\documentclass[11pt,a4paper]{article}
\usepackage[a4paper]{geometry}
\usepackage{amssymb,latexsym,mathtools,amsfonts,amsthm}
\usepackage{graphicx}
\usepackage{amsbsy}
\usepackage{hyperref}\hypersetup{hidelinks}
\usepackage{fullpage}
\usepackage{enumitem}
\usepackage{graphicx}
\usepackage{color}
\usepackage{tikz}
\usepackage{amsmath}
\usepackage{ulem}
\usepackage{subfigure}
\usepackage{booktabs,longtable}
\usepackage{bbm}
\usepackage[alphabetic,initials]{amsrefs}

\newcommand*{\dif}{\mathop{}\!\mathrm{d}}
\newcommand{\sign}[1]{\mathrm{sgn}(#1)}

\DeclareMathOperator*{\res}{Res}
\DeclareMathOperator*{\im}{Im}
\DeclareMathOperator*{\re}{Re}

\newtheorem{theorem}{Theorem}[section]
\newtheorem{lemma}[theorem]{Lemma}
\newtheorem{proposition}[theorem]{Proposition}

\newtheorem{definition}[theorem]{Definition}

\newtheorem{dbar-RHP}[theorem]{$\bar{\partial}$-RH problem}

\theoremstyle{remark}
\newtheorem{remark}[theorem]{Remark}
\theoremstyle{definition}
\numberwithin{equation}{section}
\allowdisplaybreaks[4]
\renewcommand{\thefootnote}{\fnsymbol{footnote}}

\begin{document}

\title{Confluent hypergeometric kernel determinant on multiple large intervals}

\author{Taiyang Xu\footnotemark[1], \quad Lun Zhang\footnotemark[1]~\footnotemark[2], \quad Zhengyang Zhao\footnotemark[1]}

\renewcommand{\thefootnote}{\fnsymbol{footnote}}
\footnotetext[1]{School of Mathematical Sciences, Fudan University, Shanghai 200433, China.
E-mail: \texttt{\{tyxu19,lunzhang\}@fudan.edu.cn} and \texttt{zhaozy24@m.fudan.edu.cn}.}
\footnotetext[2]{Center for Applied Mathematics and Shanghai Key Laboratory for Contemporary Applied Mathematics, Fudan University, Shanghai 200433, China.}
\date{\today}

\maketitle
	
\begin{abstract}
The confluent hypergeometric point process represents a universality class which arises in a variety of different but related areas. It particularly describes the local statistics of eigenvalues in the bulk of spectrum near a Fisher-Hartwig singular point for a broad class of unitary ensembles. It is the aim of this work to investigate large gap asymptotics of this process over a union of disjoint intervals $\cup_{j=0}^{n}(sa_j,sb_j)$, where $a_0<b_0<\dots<a_m<0<b_m<\dots<a_n<b_n$ for some $0\leq m \leq n$. 
As $s\to +\infty$, we establish a general asymptotic formula up to and including the oscillatory term of order $1$, which involves a $\theta$-functions-combination integral along a linear flow on an $n$-dimensional torus. If the linear flow has  ``good Diophantine properties'' or the ergodic properties, we further improve the error estimate or the leading term for the asymptotics of the integral. These results can be combined for the case $n=1$, which lead to a precise large gap asymptotics up to an undetermined constant.
\end{abstract}

\setcounter{tocdepth}{2} \tableofcontents

\section{Introduction}
The confluent hypergeometric kernel with two parameters $\alpha>-1/2$ and $\beta\in i\mathbb{R}$ is defined by 
\begin{align}
\label{def:CHF kernel with two parameters}
    K^{(\alpha,\hspace*{0.1em} \beta)}(x,y)=\frac{1}{2\pi i}\frac{\Gamma(1+\alpha+\beta)\Gamma(1+\alpha-\beta)}{\Gamma(1+2\alpha)^2}\frac{A(x)B(y)-A(y)B(x)}{x-y},
\end{align}
where $\Gamma(z)$ denotes the usual Gamma function and
\begin{align}
    &A(x):={{\chi}_{\beta}(x)}^{\frac12}|2x|^{\alpha}e^{-ix}\phi(1+\alpha+\beta, 1+2\alpha; 2ix),\label{def: A(x)}\\
    &B(x):={{\chi}_{\beta}(x)}^{\frac12}|2x|^{\alpha}e^{ix}\phi(1+\alpha-\beta, 1+2\alpha; -2ix),
    \label{def: B(x)}
\end{align}
with
\begin{align}
    &{\chi}_{\beta}(x):=\left\{
    \begin{aligned}
    &e^{-i\pi\beta}, \ & x\geqslant 0, \\
    &e^{i\pi\beta}, \ & x<0.
    \end{aligned}
    \right.
\end{align}
Here, 
\begin{align}\label{equ:expansion for phi}
    \phi(a,b;z)=\sum_{n=0}^{\infty}\frac{(a)_{n}}{(b)_n}\frac{z^n}{n!}, \qquad b\neq 0,-1,-2,\ldots,
\end{align}
is the confluent hypergeometric function (cf. \cite[Chapter 13]{NISTbook}), where  $(z)_{n}:=z(z+1)\cdots(z+n-1)=\Gamma(z+n)/\Gamma(z)$ is the Pochhammer symbol.

Since $\phi(a,b;z)$ is entire in $a$ and $z$, it is easily seen that the functions $A(x)$ and $B(x)$ are complex conjugated. This also implies that $K^{(\alpha,\hspace*{0.1em} \beta)}(x,y)$ is real-valued for $x,y\in\mathbb{R}$. Using the recurrence formulae for $\phi(a,b;z)$ (see \eqref{equ:CHF relation-3} and \eqref{equ:CHF relation-4} below) and the reverse formula
\begin{align}\label{equ:phi reverse formulae}
    \phi(a,b;z)=\phi(b-a,b;-z)e^{z},
\end{align} 
one has the following equivalent form of the confluent hypergeometric kernel:
\begin{multline}\label{def:equivalent CHF kernel}
    K^{(\alpha,\hspace*{0.1em} \beta)}(x,y)=\frac{1}{\pi}\frac{\Gamma(1+\alpha+\beta)\Gamma(1+\alpha-\beta)}{(1+2\alpha)\Gamma(1+2\alpha)^2}{{\chi}_{\beta}}(x)^{\frac12}{{\chi}_{\beta}(y)}^{\frac12}e^{-i(x+y)}\frac{4^{\alpha}|xy|^{\alpha}}{x-y}\\
    \times\left[x\mathcal{P}(x)\mathcal{Q}(y)-y\mathcal{P}(y)\mathcal{Q}(x)\right],
\end{multline}
where 
\begin{align}\label{def:mathcalP and mathcalQ}
\mathcal{P}(x):=\phi(1+\alpha+\beta,2+2\alpha;2ix) \quad {\rm and} \quad 
\mathcal{Q}(x):=\phi(\alpha+\beta,2\alpha;2ix).
\end{align}
The functions $\mathcal{P}$ and $\mathcal{Q}$ are related to $A$ and $B$ through the relations
\begin{align}
&x\mathcal{P}(x)=\frac{1+2\alpha}{2i}\chi_{\beta}(x)^{-1/2}|2x|^{-\alpha}e^{ix}\left(A(x)-B(x)\right),\label{equ:xmathcalP(x) relation to A and B}\\
&\mathcal{Q}(x)=\frac{1}{2}\chi_{\beta}(x)^{-1/2}|2x|^{-\alpha}e^{ix}\left[\left(1+\frac{\beta}{\alpha}\right)A(x)+\left(1-\frac{\beta}{\alpha}\right)B(x)\right].\label{equ:xmathcalQ(x) relation to A and B}
\end{align}


The determinantal point process associated with the confluent hypergeometric kernel represents a universality class which arises in a variety of different but related problems. In \cites{Bor-Ol-CMP}, Borodin and Olshanski considered the Hua-Pickrell measures (a two-parameter family of unitarily invariant probability measures) on the space of infinite Hermitian matrices. They showed that the pushforwards of these measures are determinantal on $\mathbb{R}\setminus\{0\}$ with correlation kernels given by $ K^{(\alpha,\hspace*{0.1em} \beta)}(1/x,1/y)/(xy)$. Later work of Borodin and Deift \cites{Bor-Dei-CPAM} further revealed that the hypergeometric kernel for a particle system that originates in the representation theory of the infinite-dimensional unitary group degenerates in a certain limit to the confluent hypergeometric kernel \eqref{def:CHF kernel with two parameters}. 

In the context of random matrix theory, the confluent hypergeometric point process attracts the most interest since it describes the local statistics of eigenvalues in the bulk of the spectrum near a Fisher-Hartwig singular point for a broad class of unitary ensemble of random matrices. The Fisher-Hartwig singularity is a  combination of both a root-type and a jump-type singularity that are relevant to the parameters $\alpha$ and $\beta$. For unitary random matrix ensembles generated by a concrete weight function on the unit circle with a Fisher-Hartwig singularity, we refer to \cites{Dei-Kra-Vas-IMRN} for a proof. For other random matrix ensembles, this fact is known for either $\beta=0$ or $\alpha=0$. The case when $\beta=0$ corresponds to a root-type singularity in the spectrum. In this case, we have (see \cite[Equation (13.6.9)]{NISTbook})
\begin{align}
\phi(\alpha,2\alpha;2ix)=\Gamma\left(\alpha+\frac{1}{2}\right)e^{ix}\left(\frac{x}{2}\right)^{-\alpha+\frac{1}{2}}J_{\alpha-\frac{1}{2}}(x),
\end{align}
where $J_{\nu }$, $\nu\in \mathbb{C}$, is the Bessel function of the first kind (cf. \cite[Chapter 10]{NISTbook}). This, together with \eqref{def:equivalent CHF kernel}, implies that
\begin{multline}\label{def:type-I Bessel kernel}
    K^{(\alpha,\hspace*{0.1em} 0)}(x,y)\equiv K^{(\text{Bessel1})}(x,y)
    \\
    =\left(\frac{|x|}{x}\right)^{\alpha}
    \left(\frac{|y|}{y}\right)^{\alpha}\frac{\sqrt{xy}}{2}\frac{J_{\alpha+\frac{1}{2}}(x)J_{\alpha-\frac{1}{2}}(x)-J_{\alpha+\frac{1}{2}}(y)J_{\alpha-\frac{1}{2}}(x)}{x-y}.
\end{multline}
The kernel \eqref{def:type-I Bessel kernel} is called the type-I Bessel kernel\footnote{The type-II Bessel kernel $ K^{(\text{Bessel2})}(x,y):=\frac{J_\alpha(\sqrt x)\sqrt y J'_\alpha(\sqrt y)-\sqrt x J_\alpha'(\sqrt x)J_\alpha(\sqrt y)}{2(x-y)}$ was introduced by Forrester \cite{Forr93}. The associated point process characterizes the local statistics of eigenvalues near the edge of the spectrum where the density of states exhibits an inverse square root singularity and  particularly describes the smallest eigenvalue distribution of the Laguerre-type unitary ensemble \cites{Van}.}, and has appeared in \cites{Ake-Dam-Mag-Nis-NuclearB,Kui-Van-CMP,Nag-Sle-JMP,Witte-Forrester-Nonlinearity}. The case when $\alpha=0$ corresponds to a jump-type singularity in the spectrum and the kernel \eqref{def:CHF kernel with two parameters} can be found in \cite{Foul-Mart-Sousa-CostrApp,TibboelPhdThesis} up to a scaling factor. Moreover, if $\alpha=\beta=0$, it is readily seen from \eqref{equ:expansion for phi} and \eqref{equ:phi reverse formulae} that 
\eqref{def:CHF kernel with two parameters} reduces to the sine kernel, i.e., 
\begin{align}\label{def: sine kernel}
K^{(0,\hspace*{0.1em} 0)}(x,y)\equiv K^{(\text{sine})}(x,y)=\frac{\sin (x-y)}{\pi(x-y)}.
\end{align}
It is one of the most common and well-studied correlation kernels encountered in random matrix theory.

Let
\begin{align}\label{def: Contour Sigma}
\Sigma:=\bigcup_{j=0}^{n}(a_j,b_j), 
\qquad n\in\{0\}\cup\mathbb{N},
\end{align}
be a finite union of $n+1$ disjoint intervals, where $a_j$, $b_j$, $j=0,1,\dots, n$, are some fixed real numbers. Denote by $\mathcal{K}^{(\alpha,\hspace*{0.1em} \beta)}$ the integral operator with confluent hypergeometric kernel $K^{(\alpha,\hspace*{0.1em} \beta)}$.  From the general theory of determinantal point process \cite{Sosh2000}, it is well-known that the Fredholm determinant 
\begin{align}\label{equ:CHF Fredholm det}
\mathcal{F}(\Sigma):=
\det \left(1-\mathcal{K}^{(\alpha,\hspace*{0.1em} \beta)}|_{\Sigma}\right) 
\end{align}
is the probability of finding no points (a.k.a gap probability) from the confluent hypergeometric process on $\Sigma$.  As noticed in \cite{Bor-Dei-CPAM,Witte-Forrester-Nonlinearity,Xu-Zhao-Zhao-PhysD-2024}, $\mathcal{F}$ is related to the Painlev\'{e} V system; see also the pioneering work \cites{JMMS} for the sine kernel determinant.



The fundamental problem with a long history is then to establish large gap asymptotics. For $n=0$, by considering asymptotics of the Toeplitz determinant with a Fisher-Hartwig singularity, Deift, Krasovsky and Vasilevska showed in \cite{Dei-Kra-Vas-IMRN} that 
\begin{equation}\label{eq:n=0}
 \log\mathcal{F}(-s,s)=
-\frac{s^2}{2}+2\alpha s+(\beta^2-\alpha^2-\frac{1}{4})\log s+\log C+\mathcal{O}(s^{-1}),  \qquad s\to +\infty,
\end{equation}
where 
\begin{align}
C:=\frac{\sqrt{\pi}{G}^2(1/2)G(1+2\alpha)}{2^{2\alpha^2}G(1+\alpha+\beta)G(1+\alpha-\beta)}
\end{align}
is a constant with $G(\cdot)$ being the Barnes’ $G$-function. A proof of \eqref{eq:n=0} using a different approach can be found in \cite{Xu-Zhao-CMP-2020}; see also recent works \cite{Dai-Zhai-SAPM-2022,Xu-Zhao-Zhao-PhysD-2024} for asymptotic studies of the deformed $\mathcal{F}$ in the context of thinned process. Intensive studies of large gap asymptotics on a single interval  have unraveled their rich structures and elegant forms. We particularly refer to \cites{Cloi-Mehta-JMP1973, Deift-Its-Kra-Zhou-JCAM2007, Dyson-1976-CMP, Ehrhardt-CMP-2006, Widom-1971-Indiana} for the sine process, \cites{Baik-Buckingham-DiFranco-CMP2008, Deift-Its-Kra-Zhou-CMP2008-Airy, TWAiry} for the Airy process, \cites{Dei-Kra-Vas-IMRN,Ehrhardt-AdV-2010,TWBessel} for the type-II Bessel process, \cite{CLMCMP,CLMNon,CGS} for the gerneralized Bessel process, and \cite{Dai-Xu-Zhang-Pearcey-CMP2021, YZSIAM} for the Pearcey process and its hard-edge version.

The large gap asymptotics 
acting on a union of disjoint intervals, however, has been less investigated so far. This kind of  problem was first rigorously studied by Widom for the sine process in \cite{Widom-1995-CMP}, where he obtained an explicit leading term and an oscillatory subleading term in terms of a solution of a Jacobi inversion problem. The oscillatory term was later explicitly characterized by Riemann $\theta$-functions in the landmark work \cite{Dei-Its-Zhou-Ann1997} of Deift, Its and Zhou. 
More precisely, on account of the relation \eqref{def: sine kernel}, it was shown in  
 \cite[Equation (1.35)]{Dei-Its-Zhou-Ann1997} that, as $s\to +\infty$,
\begin{align}\label{asy: sine process on n intervals}
\log{\det(1-\mathcal{K}^{(0,\hspace*{0.1em} 0)}}|_{\cup_{k=0}^{n}(sa_k,sb_k)})=
-\gamma_0s^2+\log\theta(\vec{V}(s))+\hat{G}_1\log s+o(\log s), 
\end{align}
where $\gamma_0$ is given in \eqref{equ: gamma_0} below, the $\theta$-function is defined in \eqref{def:theta func}, and $\hat{G}_1$ is a constant independent of $s$ that can be written in terms of a limit of an integral involving a combination of $\theta$-functions. 
The error estimate in \eqref{asy: sine process on n intervals} can be improved to be $\mathcal{O}(s^{-1})$ by considering the so-called ``good Diophantine properties'' of a certain vector (see Definition \ref{def:good Dio}).
A recent work of Fahs and Krasovsky \cite{Fahs-Krasovsky-2024CPAM} showed that $\hat{G}_1=-1/2$ for the case of two disjoint intervals, which corresponds to $n=1$ in \eqref{asy: sine process on n intervals}. Large gap asymptotics of the type-II Bessel process on $n$ intervals and of the Airy process on two intervals can be found in \cite{Blackstone-CC-JL-CPAM} and  \cites{Blackstone-CC-JL-RMTA,Blackstone-CC-JL-IMRN,Krasovsky-Maroudas-AdV2024}, respectively. It is worthwhile to  mention that, in the case of two intervals, the multiplicative constant is explicitly determined in \cite{Fahs-Krasovsky-2024CPAM} for the sine process and in \cite{Krasovsky-Maroudas-AdV2024} for the Airy process.


It is the aim of the present work to establish large-$s$ asymptotics of  
$
\mathcal{F}(s\Sigma)=\det\left(1-\mathcal{K}^{(\alpha, \hspace*{0.1em}\beta)}|_{s\Sigma}\right),
$
which particularly generalizes the known results in  \eqref{eq:n=0} and \eqref{asy: sine process on n intervals}. Our results are stated in the next section.


\section{Statement of results}\label{sec: statement of results}
\subsection{The Riemann surface and associated $\theta$-function}
To state the main results, we start with some ingredients from the theory of Riemann surface \cites{FK-Riemann-Surface}.


Let $\mathcal{W}$ be a two-sheeted Riemann surface of genus $n$ associated to 
\begin{align}\label{equ:sqrt(R(z))}
\sqrt{\mathcal{R}(z)}:=\sqrt{\prod_{j=0}^{n}\left(z-a_j\right)\left(z-b_j\right)}.
\end{align}
The two sheets are glued together through $\Sigma=\cup_{j=0}^{n}(a_j,b_j)$. For definiteness, we take $\sqrt{\mathcal{R}(z)}\sim z^{n+1}$ as $z\to\infty$ on the first sheet, while $\sqrt{\mathcal{R}(z)}\sim -z^{n+1}$ as $z\to\infty$ on the second sheet. Let $\{A_j\}_{j=0}^{n}$ and $\{B_j\}_{j=1}^{n}$ be the cycles illustrated in the Figure \ref{fig:A,B-Cycle}. The cycle $A_j$ lies entirely on the first sheet and surrounds the interval $(a_j,b_j)$, $j=0,\ldots,n$, the cycle $B_j$ passes from the first sheet, through the cut $(a_0,b_0)$ to the second sheet, and back again through $(a_j,b_j)$, $j=1,\ldots,n$. The cycles $\{A_j, B_j\}_{j=1}^{n}$ form the canonical 
homology basis for $\mathcal{W}$.

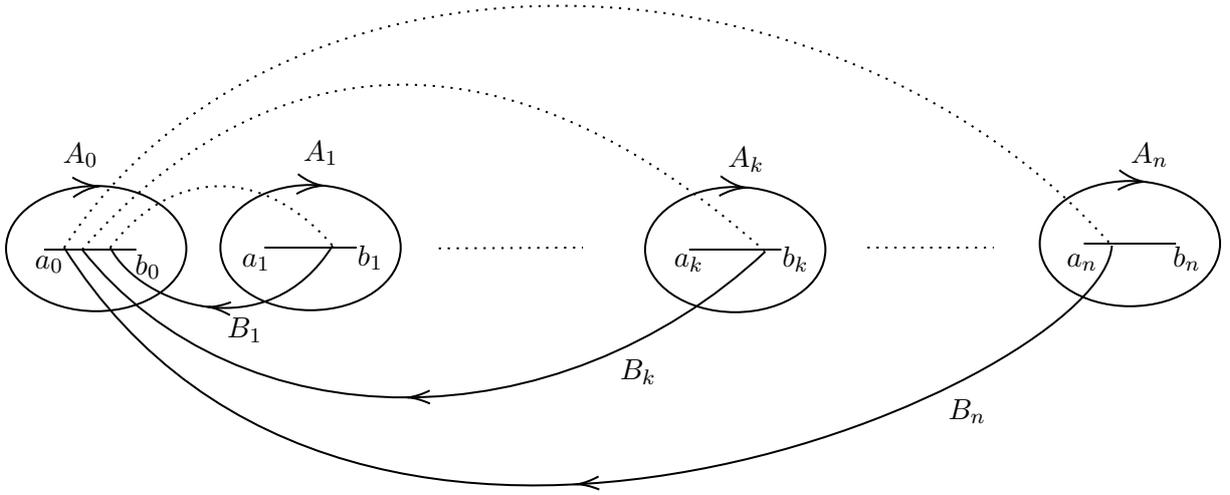
\begin{figure}[t]
\centering
\tikzset{every picture/.style={line width=0.75pt}} 
\begin{tikzpicture}[x=0.75pt,y=0.75pt,yscale=-1,xscale=1]
\draw    (33,151.6) -- (79,151.6) ;
\draw    (143,150.6) -- (189,150.6) ;
\draw    (552,148.6) -- (598,148.6) ;
\draw  [dash pattern={on 0.84pt off 2.51pt}]  (230,151) -- (302,150.6) ;
\draw  [dash pattern={on 0.84pt off 2.51pt}]  (444,150.6) -- (507,150.6) ;
\draw   (14,151.1) .. controls (14,133.7) and (34.15,119.6) .. (59,119.6) .. controls (83.85,119.6) and (104,133.7) .. (104,151.1) .. controls (104,168.5) and (83.85,182.6) .. (59,182.6) .. controls (34.15,182.6) and (14,168.5) .. (14,151.1) -- cycle ;
\draw  [dash pattern={on 0.84pt off 2.51pt}]  (43,150.6) .. controls (149,-13.4) and (426,-9.4) .. (566,149.6) ;
\draw  [dash pattern={on 0.84pt off 2.51pt}]  (52,151.6) .. controls (93,104.6) and (213,-9.4) .. (393,152.6) ;
\draw    (355,151.6) -- (401,151.6) ;
\draw  [dash pattern={on 0.84pt off 2.51pt}]  (66,150.6) .. controls (97,109.6) and (143,109.6) .. (177,149.6) ;
\draw [line width=0.75]    (66,150.6) .. controls (81,179.6) and (145,201.6) .. (177,149.6) ;
\draw [shift={(115.13,180.47)}, rotate = 4.39] [color={rgb, 255:red, 0; green, 0; blue, 0 }  ][line width=0.75]    (10.93,-3.29) .. controls (6.95,-1.4) and (3.31,-0.3) .. (0,0) .. controls (3.31,0.3) and (6.95,1.4) .. (10.93,3.29);
\draw [line width=0.75]    (52,151.6) .. controls (130,250.6) and (289,250.6) .. (393,152.6) ;
\draw [shift={(214.63,225.97)}, rotate = 358.99] [color={rgb, 255:red, 0; green, 0; blue, 0 }  ][line width=0.75]    (10.93,-3.29) .. controls (6.95,-1.4) and (3.31,-0.3) .. (0,0) .. controls (3.31,0.3) and (6.95,1.4) .. (10.93,3.29)   ;
\draw [line width=0.75]    (43,150.6) .. controls (188,382.6) and (562,217.6) .. (566,149.6) ;
\draw [shift={(299.03,269.63)}, rotate = 356.69] [color={rgb, 255:red, 0; green, 0; blue, 0 }  ][line width=0.75]    (10.93,-3.29) .. controls (6.95,-1.4) and (3.31,-0.3) .. (0,0) .. controls (3.31,0.3) and (6.95,1.4) .. (10.93,3.29)   ;
\draw   (121,150.6) .. controls (121,133.2) and (141.15,119.1) .. (166,119.1) .. controls (190.85,119.1) and (211,133.2) .. (211,150.6) .. controls (211,168) and (190.85,182.1) .. (166,182.1) .. controls (141.15,182.1) and (121,168) .. (121,150.6) -- cycle ;
\draw   (159.71,113.61) .. controls (163.7,116.69) and (167.77,118.65) .. (171.93,119.5) .. controls (167.7,119.76) and (163.39,121.13) .. (159,123.6) ;
\draw   (47.15,115.3) .. controls (51.45,117.91) and (55.72,119.42) .. (59.94,119.81) .. controls (55.76,120.53) and (51.63,122.36) .. (47.54,125.3) ;
\draw   (333,151.6) .. controls (333,134.2) and (353.15,120.1) .. (378,120.1) .. controls (402.85,120.1) and (423,134.2) .. (423,151.6) .. controls (423,169) and (402.85,183.1) .. (378,183.1) .. controls (353.15,183.1) and (333,169) .. (333,151.6) -- cycle ;
\draw   (370.71,114.61) .. controls (374.7,117.69) and (378.77,119.65) .. (382.93,120.5) .. controls (378.7,120.76) and (374.39,122.13) .. (370,124.6) ;
\draw   (530,148.6) .. controls (530,131.2) and (550.15,117.1) .. (575,117.1) .. controls (599.85,117.1) and (620,131.2) .. (620,148.6) .. controls (620,166) and (599.85,180.1) .. (575,180.1) .. controls (550.15,180.1) and (530,166) .. (530,148.6) -- cycle ;
\draw   (569.71,111.61) .. controls (573.7,114.69) and (577.77,116.65) .. (581.93,117.5) .. controls (577.7,117.76) and (573.39,119.13) .. (569,121.6) ;
\draw (27,153.4) node [anchor=north west][inner sep=0.75pt]  [font=\normalsize]  {$a_{0}$};
\draw (77,152.4) node [anchor=north west][inner sep=0.75pt]  [font=\normalsize]  {$b_{0}$};
\draw (130.2,152.4) node [anchor=north west][inner sep=0.75pt]  [font=\normalsize]  {$a_{1}$};
\draw (346,152.4) node [anchor=north west][inner sep=0.75pt]  [font=\normalsize]  {$a_{k}$};
\draw (542,152.4) node [anchor=north west][inner sep=0.75pt]  [font=\normalsize]  {$a_{n}$};
\draw (188,148) node [anchor=north west][inner sep=0.75pt]  [font=\normalsize]  {$b_{1}$};
\draw (400,148.4) node [anchor=north west][inner sep=0.75pt]  [font=\normalsize]  {$b_{k}$};
\draw (595,148.4) node [anchor=north west][inner sep=0.75pt]  [font=\normalsize]  {$b_{n}$};
\draw (41,96.4) node [anchor=north west][inner sep=0.75pt]    {$A_{0}$};
\draw (161,95.4) node [anchor=north west][inner sep=0.75pt]    {$A_{1}$};
\draw (373,98.4) node [anchor=north west][inner sep=0.75pt]    {$A_{k}$};
\draw (574,96.4) node [anchor=north west][inner sep=0.75pt]    {$A_{n}$};
\draw (123,184.4) node [anchor=north west][inner sep=0.75pt]    {$B_{1}$};
\draw (319,205.4) node [anchor=north west][inner sep=0.75pt]    {$B_{k}$};
\draw (483,225.4) node [anchor=north west][inner sep=0.75pt]    {$B_{n}$};
\end{tikzpicture}
\caption{The canonical homology basis $\{A_j, B_j\}_{j=1}^{n}$ for the Riemann surface $\mathcal{W}$.}
\label{fig:A,B-Cycle}
\end{figure}

Define an $(n+1)\times(n+1)$ matrix $\mathbb{A}$ and a column vector $\vec{a}$ by 
\begin{align}\label{def:matrix A}
\mathbb{A}:=(a_{k,l})_{0\leqslant k\leqslant n, 0\leqslant l\leqslant n}=
\begin{pmatrix}
a_{0,0} & a_{0,1} & \cdots & a_{0,n} \\
a_{1,0} & a_{1,1} & \cdots & a_{1,n} \\
\vdots  & \vdots & \vdots  & \vdots \\
a_{n,0} & a_{n,1} & \cdots & a_{n,n}
\end{pmatrix}
, \quad \vec{a}=(a_{0,n+1}, a_{1,n+1},\dots,a_{n,n+1})^{\rm T},
\end{align}
where 
\begin{align}\label{equ:a_{k,j}}
a_{k,l}:=\oint_{A_k}\frac{z^{l}}{\sqrt{\mathcal{R}(z)}}\dif z=2i(-1)^{n-k+1}\int_{a_k}^{b_k}\frac{z^{l}}{|\mathcal{R}(z)|^{\frac{1}{2}}}\dif z, 
\end{align}
and the superscript $^{\rm T}$ represents the transpose operation. As the endpoints $a_k$ and $b_k$ are real, it follows from \eqref{equ:a_{k,j}} that $a_{k,l}$ are purely imaginary for all $0\leqslant k\leqslant n, 0\leqslant l\leqslant n$. Let 
\begin{align}\label{equ:holo differential}
\vec{\omega}:=(\omega_1, \omega_2, \dots,\omega_{n})=\frac{\dif z}{\sqrt{\mathcal{R}(z)}}(1,z,\dots,z^{n-1})\tilde{\mathbb{A}}^{-1}
\end{align}
be the basis of holomorphic one-form,  where 
\begin{align}\label{def:tilde matrix A}
\tilde{\mathbb{A}}:=(a_{k, l})_{k=1,\dots,n}^{l=0,\dots,n-1}=
\begin{pmatrix}
a_{1,0} & a_{1,1} & \cdots & a_{1,n-1} \\
a_{2,0} & a_{2,1} & \cdots & a_{2,n-1} \\
\vdots  & \vdots & \vdots  & \vdots \\
a_{n,0} & a_{n,1} & \cdots & a_{n,n-1}
\end{pmatrix}
\end{align}
is invertible \cite[V.II.1, Page 324]{FK-Riemann-Surface}, and can be obtained by eliminating the first row and the last column of $\mathbb{A}$. The vector $\vec{\omega}$ satisfies the normalization conditions
\begin{align}\label{equ: normalization relation of omega}
\oint_{A_k}\omega_j=\delta_{jk}, \qquad j,k=1,\dots,n. 
\end{align}
Following the classical theory of Riemann surfaces, we introduce the associated Riemann matrix of $B_j$-period, $j=1,\dots,n$, by
\begin{equation}\label{def:tau matrix}
\tau:=\left(\tau_{ij}\right)_{i,j=1}^{n}=\left(\oint_{B_j}\omega_i\right)_{i,j=1}^{n}.
\end{equation}
By \cite[Page 63, Proposition]{FK-Riemann-Surface}, one has that $\tau$ is symmetric and has a positive definite imaginary part, i.e., $-i\tau$ is positive definite. 

The associated multi-dimensional Riemann $\theta$-function is defined by 
\begin{align}\label{def:theta func}
\theta\left(\vec{z}\right)=\sum_{\vec{m}\in\mathbb{Z}^{n}}e^{2\pi i\vec{m}^{\mathrm{T}}\vec{z}+i\pi\vec{m}^{\mathrm {T}}\tau\vec{m}}, \qquad \vec{z}=(z_1,\dots,z_n)^{\rm T}\in\mathbb{C}^{n}.
\end{align}
One could verify $\theta$ is an even entire function, and satisfies the periodic properties
\begin{align}\label{property:theta function}
\theta(\vec{z}\pm\vec{e}_j)=\theta(\vec{z}), \qquad 
\theta(\vec{z}\pm\vec{\tau}_j)=e^{\mp 2\pi iz_j-\pi i\tau_{jj}}\theta(\vec{z}),
\end{align}
where $\vec{e}_j$ denotes the standard column vector in $\mathbb{C}^{n}$ with $1$ in the $j$-th position and zero elsewhere, $\vec{\tau}_j:=\tau\vec{e}_j$ is the $j$-th column of matrix $\tau$ defined in \eqref{def:tau matrix} and $z_j$ is the $j$-th element of $\vec{z}$.

We next set
\begin{align}\label{def-intro-Omega_j}
\Omega_j:=2\sum_{k=0}^{j-1}(-1)^{n-k}\int_{b_k}^{a_{k+1}}\frac{\mathrm{p}(s)}{|\mathcal{R}(s)|^{\frac{1}{2}}}\dif s, \quad j=1,\ldots,n,
\end{align}
where 
\begin{align}\label{equ:mathbf{p}(z)}
    \mathrm{p}(z)=z^{n+1}+\sum_{j=0}^{n}p_{j}z^{j}
\end{align}
is a polynomial of degree $n+1$ satisfying
\begin{align}\label{equ:A-cycle=0}
\oint_{A_k}\frac{\mathrm{p}(s)}{\sqrt{\mathcal{R}(s)}}\dif s=0, \qquad k=0,1,\dots,n.
\end{align}
Clearly, the coefficients of $\mathrm{p}$ are uniquely given by 
\begin{align}\label{equ:coeff of p}
(p_0, \dots, p_n)^{\mathrm{T}}=-\mathbb{A}^{-1}\vec{a}.
\end{align}
We will show $\Omega_j>0$ in Proposition \ref{prop:g-func} below. 
Given $s\geqslant 0$, we then introduce a column vector 
\begin{align}\label{def-intro-vecV}
\vec{V}(s):={(V_{1}(s),V_2(s),\dots, V_n(s))}^{\rm T}\in\mathbb{R}^{n},
\end{align}
where
\begin{align}\label{def:Vj}
V_j(s)=\frac{s}{2\pi}\Omega_j+\frac{1}{2\pi}\im(\zeta_j)\in\mathbb{R}, \quad j=1,\dots,n,
\end{align}
and where $\zeta_j$ is defined in \eqref{equ: sol of zetaj} below. 

Finally, we define a function $\mathcal{L}:\mathbb{C}\times\mathbb{R}^{n}/\mathbb{Z}^{n}\to\mathbb{C}$ by 
\begin{align}\label{def-intro-mathcalL}
\mathcal{L}(z,\vec{\mu})=\frac{h(z)}{\mathrm{p}(z)}\eta(z,\vec{\mu})
\end{align}
where
\begin{align}\label{def: h(z) and eta(z,mu)}
h(z)=\prod_{k=0}^{n}(z-a_k)+\prod_{k=0}^{n}(z-b_k), \quad
\eta\left(z,\vec{\mu}\right)=\frac{\theta(\vec{0})^2\theta(\vec{\mathcal{A}}(z)+\vec{\mu}+\vec{d})\theta(\vec{\mathcal{A}}(z)-\vec{\mu}+\vec{d})}{\theta(\vec{\mu})^2\theta(\vec{\mathcal{A}}(z)+\vec{d})^2}.
\end{align}
Here, the vector $\vec{d}$ is defined in \eqref{def:choice of d} below and 
\begin{align}\label{equ:Abel map on mathbbC}
\vec{\mathcal{A}}(z):=\int_{a_0}^{z}\vec{\omega}^{\rm T},
\end{align}
where the integration is taken on the first sheet. The function $\vec{\mathcal{A}}(z)$ can be extended to the Riemann surface $\mathcal{W}$ and then becomes the Abel's map \cite{FK-Riemann-Surface}. Also, one can see that $\mathcal{L}(z,\cdot)$ is periodic in $\mathbb{R}^{n}$ with period lattice $\mathbb{Z}^{n}$.


\subsection{Large gap asymptotics}
Recall the confluent hypergeometric kernel determinant $\mathcal{F}(\Sigma)$ defined in \eqref{equ:CHF Fredholm det}. Our first theorem provides a general asymptotic formula of the determinant.


\begin{theorem}[The general case]\label{thm: general case}
Let $\Sigma=\cup_{j=0}^{n}(a_j,b_j)$ be fixed such that $a_0<b_0<\cdots<a_m<0<b_m<\cdots<a_n<b_n$ for some  $0\leqslant m\leqslant n$. For $\alpha>-1/2$ and $\beta\in i\mathbb{R}$, we have, as $s\to+\infty$,
\begin{multline}\label{asy result: general case}
\log\mathcal{F}(s\Sigma)=-\gamma_0s^2-2i\mathcal{D}_{\infty,1}s+\log\theta\left(\vec{V}(s)\right)+(\beta^2-\alpha^2)\log s 
\\ 
-\frac{1}{16}\sum_{j=0}^{n}\int_{\hat{s}}^{s}\left(\mathcal{L}\left(a_j, \vec{V}(t)\right)+\mathcal{L}\left(b_j, \vec{V}(t)\right)\right)\frac{\dif t}{t}+\breve{C}_{1}+\mathcal{O}(s^{-1}),
\end{multline}
where 
\begin{align}\label{equ: gamma_0}
\gamma_0=-\frac{1}{\pi i}\sum_{j=0}^{n}\int_{a_j}^{b_j}\frac{z\mathrm{p}(z)}{\sqrt{\mathcal{R}(z)}_{+}}\dif z\in\mathbb{R},
\end{align}
$\mathcal{D}_{\infty,1}$ given in \eqref{def:Dinfty1-a} is purely imaginary and depends on the parameters $\alpha$ and $\beta$, the vector $\vec{V}(s)\in\mathbb{R}^{n}$ is defined in \eqref{def-intro-vecV}, $\hat s>0$ is a sufficiently large number independent of $s$, $\mathcal{L}(p,\vec{V}(t))$ is real for $p\in\mathcal{I}_{e}:=\{a_j,b_j\}_{j=0}^{n}$ with the function $\mathcal{L}$ given in \eqref{def-intro-mathcalL} and  $\breve{C}_1$ is an undetermined constant independent of $s$. Moreover, for $p\in\mathcal{I}_{e}$, as $s\to+\infty$ we have 
\begin{align}\label{equ: hat{mathcalL}_p}
\int_{\hat{s}}^{s}\mathcal{L}\left(p,\vec{V}(t)\right)\frac{\dif t}{t}=\hat{\mathcal{L}}_{p}\log s+o(\log s), \quad
\hat{\mathcal{L}}_{p}:=\lim_{T\to+\infty}\frac{1}{T}\int_{0}^{T}\mathcal{L}\left(p,\vec{V}(t)\right)\dif t.
\end{align}
\end{theorem}

If $n=0$ and $\Sigma=(-1,1)$, our result reduces to \eqref{eq:n=0} except the multiplicative constant. Indeed, the terms which are associated to the $\theta$-functions, such as $\log\theta(\vec{V}(s))$ and $\vec{V}$, should be interpreted as $0$ for $n=0$. The polynomial $\mathrm{p}$ simply reads $\mathrm{p}(z)=z$ by noticing that $\Sigma=(-1,1)$. 
Thus, it follows from \eqref{equ: gamma_0} that 
\begin{align}\label{equ: gamm_0 in the case of n=0}
\gamma_0=-\frac{1}{\pi i}\int_{-1}^{1}\frac{z^2}{i\sqrt{1-z^2}}\dif z=
\frac{2}{\pi}\int_0^{1}\frac{z^2}{\sqrt{1-z^2}}\dif z=\frac{2}{\pi}\cdot\frac{\pi}{4}=\frac{1}{2}.
\end{align}
Using the definition of $\mathcal{D}_{\infty,1}$ in \eqref{def:Dinfty1-a} with $n=0$,  we have 
\begin{align}\label{equ: mathcalDinfty1 in the case of n=0}
 \mathcal{D}_{\infty,1}&=\frac{\beta}{\pi i}\int_{-1}^{1}\frac{\xi\log|\xi|}{i\sqrt{1-\xi^2}}\dif\xi+\frac{\alpha}{2}\left(\int_{-1}^{0}-\int_{0}^{1}\right)\frac{\xi}{i\sqrt{1-\xi^2}}\dif\xi
\nonumber 
\\
&=-\alpha\int_{0}^{1}\frac{\xi}{i\sqrt{1-\xi^2}}\dif\xi=i\alpha,   
\end{align}
where the first integral vanishes due to the fact that $\xi\log|\xi|/\sqrt{1-\xi^2}$ is odd on $(-1,1)$. As for the function $\mathcal{L}(z,\vec{V}(t))$, it now reads
\begin{align}\label{equ: mathcalL in the case of n=0}
\mathcal{L}(z,\vec{V}(t))=\mathcal{L}(z,0)\overset{\eqref{def-intro-mathcalL}}{=}\frac{(z+1)+(z-1)}{z} = 2.
\end{align}
Substituting \eqref{equ: gamm_0 in the case of n=0}--\eqref{equ: mathcalL in the case of n=0} into \eqref{asy result: general case}, we recover \eqref{eq:n=0} up to the constant term. 

For general $n$ with $\alpha=\beta=0$, $\mathcal{F}(s\Sigma)$ is actually the sine kernel determinant on a union of disjoint intervals, and one can check that our formula \eqref{asy result: general case} is consistent with \cite[Equation (1.34)]{Dei-Its-Zhou-Ann1997}. Due to the effect of the parameters $\alpha$ and $\beta$, we note that two new terms $-2i \mathcal{D}_{\infty,1}s$ and $(\beta^2-\alpha^2)\log s$ emerge, and the vector $\vec{V}$ also differs from that in \cite{Dei-Its-Zhou-Ann1997}.



Inspired by \cites{Blackstone-CC-JL-CPAM, Dei-Its-Zhou-Ann1997}, the asymptotics \eqref{asy result: general case} can be further improved if $\vec{V}$ in \eqref{def-intro-vecV} satisfies some suitable conditions defined in what follows. 
\begin{definition}\label{def:good Dio}
Consider the linear flow
\begin{align}\label{def: linear flow}
(0,+\infty)\ni s \mapsto (V_{1}(s) \mod 1, \quad V_{2}(s) \mod  1, ~~
\ldots, ~~ V_{n}(s) \mod  1).
\end{align}
\begin{itemize}
\item The linear flow \eqref{def: linear flow} has  ``good Diophantine properties''
if there exist $\delta_1, \delta_2>0$ such that 
\begin{align}\label{condition: diophantine}
\vert\vec{m}^{\rm T}\vec{\Omega}\vert\geqslant \delta_1{\Vert\vec{m}\Vert}_{2}^{-\delta_2} 
\end{align}
for all $\vec{m}\in\mathbb{Z}^{n\times1}$ with $\vec{m}^{\rm T}\vec{\Omega}\neq0$,
where $\vec{\Omega}:={(\Omega_1,\dots,\Omega_n)}^{\rm T}$ and ${\Vert\vec{m}\Vert}_{2}={|\vec{m}^{\rm T}\vec{m}|}^{\frac{1}{2}}$.
\item The linear flow \eqref{def: linear flow} is ergodic in the $n$-dimensional torus $\mathbb{R}^{n}/\mathbb{Z}^{n}$ if $\{\vec{V}(s) \mod \mathbb{Z}^{n}\}_{s>0}$ is dense in $\mathbb{R}^{n}/\mathbb{Z}^{n}$; cf. \cite[Definition 1.3.2]{Katok-Hassel-DS}. Equivalently, the linear flow \eqref{def: linear flow} is ergodic in $\mathbb{R}^{n}/\mathbb{Z}^{n}$ if 
$\{\Omega_j\}_{j=1}^{n}$ are rationally independent \cite[Proposition 1.5.1]{Katok-Hassel-DS}, that is, if there exist $(c_1,c_2,\dots,c_n)\in\mathbb{Z}^{n}$ such that 
\begin{align}\label{condition: ergodic}
c_1\Omega_1+c_2\Omega_2+\cdots+c_n\Omega_n=0,
\end{align}
then $c_1=c_2=\cdots=c_n=0$.
\end{itemize}
\end{definition}

Let
\begin{align}
\mathcal{S}_{D}:=\left\{\vec{\Omega}: \ \vec{\Omega} \textrm{ has ``good Diophantine properties''}  \right\}, \quad 
\mathcal{S}_{E}:=\left\{\vec{\Omega}: \ \vec{\Omega} \textrm{ is rationally independent}  \right\}.
\end{align}
be two subsets of $(0,+\infty)^n$. They are generic in the sense that the Lebesgue measures of $(0,+\infty)^{n}\setminus\mathcal{S}_{D}$ and $(0,+\infty)^{n}\setminus\mathcal{S}_{E}$ are zero. If $n=1$, we have $\mathcal{S}_{D}=\mathcal{S}_{E}=(0,+\infty)$. If $n\geq 2$, however, it comes out that $\mathcal{S}_{D}$ and $\mathcal{S}_{E}$ are non-inclusive (i.e., $\mathcal{S}_{D}\nsubseteq\mathcal{S}_{E}$, $\mathcal{S}_{E}\nsubseteq\mathcal{S}_{D}$) and $\mathcal{S}_{D}\neq{(0,+\infty)}^{n}\neq\mathcal{S}_{E}$ in general. Some concrete examples of $\mathcal{S}_{D}$ and $\mathcal{S}_{E}$ can be found in \cite[Page 3305]{Blackstone-CC-JL-CPAM}. Similar to the arguments given in \cite[Page 211]{Dei-Its-Zhou-Ann1997}, one can show that the range of the map $(a_0,b_0,\dots,a_n,b_n)\mapsto\vec{\Omega}$ contains an open ball in ${(0,+\infty)}^{n}$. This implies that all of the following four cases
\begin{align}
\vec{\Omega}\notin\mathcal{S}_{D}\cup\mathcal{S}_{E}, \quad \vec{\Omega}\in\mathcal{S}_{D}\setminus\mathcal{S}_{E}, \quad
\vec{\Omega}\in\mathcal{S}_{E}\setminus\mathcal{S}_{D},
\quad 
\vec{\Omega}\in\mathcal{S}_{D}\cap\mathcal{S}_{E},
\end{align}
can and do emerge for the some $a_j$ and $b_j$. Note that the first case does not occur for $n=2$ since $\mathcal{S}_{D}\cup\mathcal{S}_{E}={(0,+\infty)}^{2}$.

The vector $\vec{\Omega}$ such that $\eqref{condition: diophantine}$ holds true is usually called \textit{Diophantine vector}. The benefit of introducing the good  Diophantine properties is to ensure that the components of $\vec{\Omega}$ cannot be approximated by rationals too rapidly, thereby preventing excessive clustering of orbits \cite{Katok-Hassel-DS}, which finally leads to the improvement of the error estimate in \eqref{asy result: general case} as sated in the next theorem. 
\begin{theorem}[The ``good Diophantine properties'' case]
\label{thm: Diophantine case}
Let $\Sigma=\cup_{j=0}^{n}(a_j,b_j)$ be fixed such that $a_0<b_0<\cdots<a_m<0<b_m<\cdots<a_n<b_n$ for some  $0\leqslant m\leqslant n$ and assume that $\vec{\Omega}\in \mathcal{S}_D$.  As $s\to +\infty$, one has
\begin{align}\label{equ: hat{mathcalL}_p-Dio case}
\int_{\hat{s}}^{s}\mathcal{L}\left(p,\vec{V}(t)\right)\frac{\dif t}{t}=\hat{\mathcal{L}}_{p}\log s+C_p+\mathcal{O}(s^{-1}),  
\end{align}
where $p\in\mathcal{I}_{e}$ and $C_p$ is independent of $s$. Thus, we have,
as $s\to+\infty$,
\begin{multline}\label{asy result: Diophantine case}
\log\mathcal{F}(s\Sigma)=-\gamma_0s^2-2i\mathcal{D}_{\infty,1}s+\log\theta\left(\vec{V}(s)\right)
\\
+\left[\beta^2-\alpha^2-\frac{1}{16}\sum_{j=0}^{n}(\hat{\mathcal{L}}_{a_j}+\hat{\mathcal{L}}_{b_j})\right]\log s+\breve{C}_{2}+\mathcal{O}(s^{-1}),
\end{multline}
where $\breve{C}_{2}=\breve{C}_1-\frac{1}{16}\sum_{j=0}^{n}(C_{a_j}+C_{b_j})$ is a constant independent of $s$ with $\breve{C}_1$ as in 
\eqref{asy result: general case}.
\end{theorem}
\begin{remark}
Analogues of the asymptotic formula \eqref{asy result: Diophantine case} have appeared in  \cite[Equation (1.37)]{Dei-Its-Zhou-Ann1997} for the sine kernel determinant and in \cite[Theorem 1.3]{Blackstone-CC-JL-CPAM} for the type-II Bessel kernel determinant.
\end{remark}
The vector $\vec{\Omega}$ which is rationally independent is usually utilized to eliminates the influence of resonance, ensuring that orbits densely fill the entire torus \cite{Katok-Hassel-DS}. In this case, the asymptotic formula given in Theorem \ref{thm: general case} can be improved in the following way. 
\begin{theorem}[The ergodic case]\label{thm: ergodic case}
Let $\Sigma=\cup_{j=0}^{n}(a_j,b_j)$ be fixed such that $a_0<b_0<\cdots<a_m<0<b_m<\cdots<a_n<b_n$ for some  $0\leqslant m\leqslant n$ and assume that $\vec{\Omega}\in \mathcal{S}_E$. As $s\to +\infty$, one has 
\begin{align}\label{equ: hat{mathcalL}_p-ergodic case}
\int_{\hat{s}}^{s}\mathcal{L}\left(p,\vec{V}(t)\right)\frac{\dif t}{t}=\hat{\mathcal{L}}_{p}\log s+o(\log s), \quad
\hat{\mathcal{L}}_{p}=\frac{h(p)}{\mathrm{p}(p)}\int_{{[0,1)}^{n}}\eta(p; u_1, u_2,\dots u_n )\dif u_1\cdots\dif u_n.
\end{align}
Thus, we have, as $s\to+\infty$,
\begin{align}\label{asy result: ergodic case}
\log\mathcal{F}(s\Sigma)=-\gamma_0s^2-2i\mathcal{D}_{\infty,1}s+
\left[\beta^2-\alpha^2-\frac{1}{16}\sum_{j=0}^{n}(\hat{\mathcal{L}}_{a_j}+\hat{\mathcal{L}}_{b_j})\right]\log s+o(\log s),
\end{align}
where the constant $\hat{\mathcal{L}}_p$, $p\in\mathcal{I}_e$, is explicitly given by the $n$-fold integral in \eqref{equ: hat{mathcalL}_p-ergodic case}.
\end{theorem}
\begin{remark}
If $\vec{\Omega}\in \mathcal{S}_D\cap \mathcal{S}_E$,  it follows from Theorems  \ref{thm: Diophantine case} and \ref{thm: ergodic case} that the asymptotics of $\mathcal{F}(s\Sigma)$ takes the form \eqref{asy result: Diophantine case} but with  $\hat{\mathcal{L}}_p$ given in \eqref{equ: hat{mathcalL}_p-ergodic case}.
\end{remark}

Finally, we focus on the case that $n=1$. In this case, it has been shown in  \cite{Fahs-Krasovsky-2024CPAM} that $\hat{\mathcal{L}}_{p}\equiv2$ for $\alpha=\beta=0$. We show this is true for general $\alpha$ and $\beta$. This, together with the fact that $\vec{\Omega}\in \mathcal{S}_D\cap \mathcal{S}_E$ for $n=1$, implies the following result. 


\begin{theorem}[The $n=1$ case]\label{thm: n=1 case}
Let $\Sigma:=(a_0,b_0)\cup(a_1,b_1)$ be fixed such $a_0<b_0<\cdots<a_m<0<b_m<\cdots<a_1<b_1$ for some  $0\leqslant m\leqslant 1$. For $\alpha>-1/2$ and $\beta\in i\mathbb{R}$, we have,  as $s\to+\infty$,
\begin{align}
\label{asy result: n=1 case}
\log\mathcal{F}(s\Sigma)=-\gamma_0s^2-2i\mathcal{D}_{\infty,1}s+\log\theta\left(V_{1}(s)\right)
+\left(\beta^2-\alpha^2-\frac{1}{2}\right)\log s+C+\mathcal{O}(s^{-1}),
\end{align}
where $C$ is an undetermined constant independent of $s$.
\end{theorem} 

\begin{remark}
Numerical experiments suggest that the integral $\hat{\mathcal{L}}_{p}\equiv2$ in \eqref{equ: hat{mathcalL}_p-ergodic case} holds for general $n>2$, whose proof remains open. Also, we can not evaluate the constant of integration $C$ with current method, which in general is a challenging task \cite{Fahs-Krasovsky-Maroudas-inMemWidom, Krasovsky-ICMP-2009}.
\end{remark}

We conclude our main results by noting that although it is assumed that $0\in \Sigma$ throughout this paper, it is also possible to consider the case that $0 \notin \Sigma$ by a modification of our proofs.

\paragraph*{Organization of the paper}
The rest of this paper is devoted to the proofs of our main results. 
Following a general framework established in \cite{Dei-Its-Zhou-Ann1997}, the starting point is an identity that relates $\partial_{s}\log\mathcal{F}(s\Sigma)$ to a Riemann-Hilbert (RH) problem with constant jumps, which is presented in Section \ref{sec: RH characterization}. While this differential identity is easy to check for $\alpha=\beta=0$ (i.e., for the sine kernel determinant), it is highly non-trivial for general $\alpha$ and $\beta$; see Proposition \ref{prop: differential identity} below for the precise statement.  
The Deift/Zhou steepest analysis of the associated RH problem then consists the main body of this paper. After carrying out the first two
transformations in Section \ref{sec: DZ steepest analysis}, we construct global and local paramertices in Sections \ref{sec: global parametrix} and 
\ref{sec: local parametrix}, respectively. The analysis finishes with a small norm RH problem given in Section \ref{sec: small norm RH problem}. The outcome of our analysis is the proofs of our main results presented in 
Section \ref{sec: proof of theorems}.


\paragraph*{Notations} Throughout this paper, the following notations will be used.
\begin{itemize}
\item As usual, the classical Pauli matrices $\{\sigma_j\}_{j=1,2,3}$ are defined by
\begin{equation}\label{def:PauliM}
\sigma_1:=\begin{pmatrix}0 & 1 \\ 1 & 0\end{pmatrix}, \quad
\sigma_2:=\begin{pmatrix}0 & -i \\ i & 0\end{pmatrix}, \quad
\sigma_3:=\begin{pmatrix}1 & 0 \\ 0 & -1\end{pmatrix}.
\end{equation}
\item We denote by $\mathcal{I}_{e}$ be the union of endpoints of $\Sigma$, i.e., $\mathcal{I}_{e}:=\{a_j,b_j\}_{j=0}^{n}$.
\item If $M$ is a matrix, then $M_{ij}$ stands for its $(i,j)$-th entry and $M^{\rm T}$ stands for its transpose. An unimportant entry of $M$ is usually denoted by $*$. We use $I$ to denote an identity matrix.
\item For a matrix valued function $M(z)$ defined for $z\in\mathbb{C}\setminus{\Gamma}$, where $\Gamma$ is an oriented contour, we use $M_{+}(z)$ ($M_{-}(z)$) for the limits of $M(z)$ as 
$z'\to z\in\Gamma$ from the positive side (negative side). The positive side (negative side) is on our left (right) as we follow $\Gamma$ according to its orientation.
\item We use $\mathbb{C}^+$ and $\mathbb{C}^-$ to denote the upper and the lower half complex plane, respectively.
\end{itemize}

\section{An RH characterization of $\mathcal{F}(s\Sigma)$}\label{sec: RH characterization}

In this section, our goal is to establish a relation between $\partial_{s}\log\mathcal{F}(s\Sigma)$ and an RH problem with constant jumps. A key observation here is a remarkable differential identity of the scaled confluent hypergeometric kernel, which might be of independent interest.


\subsection{Preliminaries}\label{subsec: Preliminaries}
By \eqref{def:CHF kernel with two parameters}, we rewrite the confluent hypergeometric kernel $K^{(\alpha,\hspace*{0.1em}\beta)}$ as an integrable form in the sense of IIKS \cites{Its-Ize-Kore-Slav}, i.e.,
\begin{align}
    K^{(\alpha,\hspace*{0.1em}\beta)}(x,y)=\frac{\vec{f}(x)^{\text{T}}\vec{h}(y)}{x-y}=\frac{\sum_{k=1}^2f_{k}(x)h_k(y)}{x-y},
\end{align}
where 
\begin{align}\label{def:f and h}
    \vec{f}(x)=\begin{pmatrix}
        f_1(x) \\
        f_2(x)
    \end{pmatrix}:=\frac{1}{\sqrt{2\pi i}}
    \begin{pmatrix}
        \frac{\Gamma(1+\alpha-\beta)}{\Gamma(1+2\alpha)}B(x) \\
        \frac{\Gamma(1+\alpha+\beta)}{\Gamma(1+2\alpha)}A(x)
    \end{pmatrix},
~~
    \vec{h}(y)=\begin{pmatrix}
        h_1(y) \\
        h_2(y)
    \end{pmatrix}:=\frac{1}{\sqrt{2\pi i}}
    \begin{pmatrix}
        -\frac{\Gamma(1+\alpha+\beta)}{\Gamma(1+2\alpha)}A(y) \\
        \frac{\Gamma(1+\alpha-\beta)}{\Gamma(1+2\alpha)}B(y)
    \end{pmatrix},
\end{align}
and where the functions $A$ and $B$ are defined in \eqref{def: A(x)} and \eqref{def: B(x)}. Note that $\sum_{k=1}^{2}f_{k}(z)g_{k}(z)=0$.

Using the definition of Fredholm determinant and a change of variable, it's readily seen that
\begin{align}\label{equ:k|sSigma ks|Sigma}
\mathcal{F}(s\Sigma)=\det\left(1-\mathcal{K}^{(\alpha,\hspace*{0.1em}\beta)}|_{s\Sigma}\right)=\det\left(1-\mathcal{K}_{s}^{(\alpha,\hspace*{0.1em}\beta)}|_{\Sigma}\right),
\end{align}
where $\mathcal{K}^{(\alpha,\hspace*{0.1em}\beta)}_{s}$ is the trace-class integral operator acting on $L^{2}(\Sigma)$
with kernel 
\begin{align}\label{equ: scaled kernel}
K_{s}^{(\alpha,\hspace*{0.1em}\beta)}(x,y):=sK^{(\alpha,\hspace*{0.1em}\beta)}(sx,sy).
\end{align}
By the well-known differential equalities for the trace-class integral operators, it follows that
\begin{align}\label{equ: partialsF(sSigma)=Trace}
\partial_{s}\log \mathcal{F}(s\Sigma)
&\overset{\eqref{equ:k|sSigma ks|Sigma}}{=}\partial_s\log\det\left(1-\mathcal{K}_{s}^{(\alpha,\hspace*{0.1em}\beta)}|_{\Sigma}\right) =\partial_s{\rm Tr}\left(\log\left(1-\mathcal{K}_{s}^{(\alpha,\hspace*{0.1em}\beta)}|_{\Sigma}\right)\right)\nonumber\\&=-{\rm Tr}\left(\left(1-\mathcal{K}^{(\alpha,\hspace*{0.1em}\beta)}_{s}|_{\Sigma}\right)^{-1}\partial_s\mathcal{K}^{(\alpha,\hspace*{0.1em}\beta)}_{s}|_{\Sigma}\right).
\end{align}

It comes out that the resolvent operator $\mathcal{R}_{s\Sigma}:={(1-\mathcal{K}^{(\alpha,\hspace*{0.1em} \beta)}|_{s\Sigma})}^{-1}-1$ 
is also integrable. Indeed, by setting
\begin{align}
\vec{F}(u)=\begin{pmatrix}
    F_1(u) \\
    F_2(u)
\end{pmatrix}:=\left(1-\mathcal{K}^{(\alpha,\hspace*{0.1em}\beta)}|_{s\Sigma}\right)^{-1}\vec{f}(u), \quad 
\vec{H}(u)=\begin{pmatrix}
        H_1(u) \\
        H_2(u)
    \end{pmatrix}:=\left(1-\mathcal{K}^{(\alpha,\hspace*{0.1em}\beta)}|_{s\Sigma}\right)^{-1}\vec{h}(u),
\end{align}
we have that the kernel $R_{s\Sigma}(x,y)$ of the resolvent operator $\mathcal{R}_{s\Sigma}$ is given by
\begin{align}
R_{s\Sigma}(x,y)=\frac{\vec{F}^{\rm T}(x)\vec{H}(y)}{x-y}.
\end{align}
By the general theory of integrable operators \cite[Lemmas 2.8 and 2.12]{Dei-Its-Zhou-Ann1997}, one has
\begin{align}
\vec{F}(z)=Y_{+}(z)\vec{f}(z), \quad \vec{H}(z)={Y_{+}(z)}^{-\rm T}\vec{h}(z),
\end{align}
where
\begin{align}
 Y(z):=I-\int_{s\Sigma}\frac{\vec{F}(x)\vec{h}(x)^{\text{T}}}{x-z}\dif x
\end{align}
 solves the following RH problem.
\paragraph{RH problem for $Y$}
\begin{itemize}
\item[\rm {(a)}]  $Y(z)$ is holomorphic for $z\in\mathbb{C}\setminus{\cup_{j=0}^{n}[sa_j, sb_j]}$, where the orientation of $(sa_j,sb_j)$ is taken from the left to the right for $j=0,\ldots,n$.
\item[\rm {(b)}]  For $z\in\cup_{j=0}^{n}(sa_j,sb_j)$, we have 
\begin{align}
    Y_{+}(z)=Y_{-}(z)\left(I-2\pi i\vec{f}(z)\vec{h}(z)^{{\rm T}}\right),
\end{align}
where the functions $\vec{f}$ and $\vec{h}$ are defined in \eqref{def:f and h}.
\item [\rm {(c)}] As $z\rightarrow\infty$, we have 
\begin{align}
    Y(z)=I+\frac{Y_{1}(s)}{z}+\mathcal{O}(z^{-2}),
\end{align}
where
\begin{align}\label{equ: Y_1(s)}
Y_1(s)=\int_{s\Sigma}\begin{pmatrix}
F_1(x)h_1(x) & F_1(x)h_2(x)\\
F_2(x)h_1(x) & F_2(x)h_2(x)
\end{pmatrix}\dif x.
\end{align}
\item[\rm {(d)}] As $z\to sp$ for $p\in\mathcal{I}_{e}$, we have 
\begin{align}
    Y(z)=\mathcal{O}\left(\log\left(z-sp\right)\right). 
\end{align}
\end{itemize}

We now transform the RH problem for $Y$ to a new one with constant jumps by using the confluent hypergeometric parametrix $\Phi_{\rm CH}$ introduced in Appendix \ref{appendix-A-paragraph-RHP Phi}, which is also known as an undressing procedure. 
To begin with, it is readily seen from \eqref{def:f and h} and \eqref{solution4phi} that
\begin{equation}\label{def:f asso Phi}
\vec{f}(z)=e^{\frac{\beta}{2} \pi i \sigma_3} \Phi_{\rm CH}(2 z) e^{\pm \frac{\alpha}{2} \pi i \sigma_3} e^{\mp\frac{\beta}{2} \pi i \sigma_3}\left(\frac{1}{\sqrt{2 \pi i}}\right)^{\sigma_3}\left(\begin{array}{l}
1 \\
0
\end{array}\right),\end{equation}
and
\begin{equation}\label{def:h asso Phi}
\vec{h}(z)= e^{-\frac{\beta}{2} \pi i \sigma_3}{\Phi_{\rm CH}(2z)}^{\rm -T} e^{\mp \frac{\alpha}{2} \pi i \sigma_3} e^{\pm\frac{\beta}{2} \pi i \sigma_3}\left(\frac{1}{\sqrt{2 \pi i}}\right)^{-\sigma_3}\left(\begin{array}{l}
0 \\
1
\end{array}\right),\end{equation}
for $z\in \Omega_{\Phi_{\rm CH},1}$ and $\pm \re z>0$, where the region $\Omega_{\Phi_{\rm CH},1}$ is illustrated in Figure \ref{fig:RHP Phi}. Let $\varphi\in(0,\pi/4)$ and set
\begin{equation}
    \begin{aligned}
    &\Gamma_{X,r}:=b_n+e^{i\varphi}\mathbb{R}^{+}, \quad  && \Gamma_{X,r}^*:=b_n+e^{-i\varphi}\mathbb{R}^{+}, \quad
    \Gamma_{X,l}:=a_0+e^{-i\varphi}\mathbb{R}^{-},  \\
    & \Gamma_{X,l}^{*}:=a_0+e^{i\varphi}\mathbb{R}^{-}, 
 \quad  && \Gamma_{X, d}=e^{-i\pi/2}\mathbb{R}^{+}. 
\end{aligned}
\end{equation}
Clearly, the above rays divide the complex plane into five regions denoted by $\Omega_{X,j}$, $j=1,\ldots,5$, as illustrated in Figure \ref{fig:RHP X}.

We then define
\begin{equation}\label{def:RHP X}
    X(z)=(2s)^{\beta \sigma_3} e^{-\frac{\beta}{2} \pi i \sigma_3} Y(sz) e^{\frac{\beta}{2} \pi i \sigma_3} \left\{\begin{aligned}
  &\Phi_{\rm CH}(2sz), \quad &&z\in\Omega_{X,j} ,~~ j=2,5,\\
  & \widehat{\Phi}_{{\rm CH},j}(2sz),  \quad &&z \in\Omega_{X,j}, ~~ j=1,3,4,
  \end{aligned}
  \right.
  \end{equation}
where $\widehat{\Phi}_{{\rm CH},j}(z)$, $j=1,3,4$, denotes the analytic extension of $\Phi_{\rm CH}(z)$ from the region $\Omega_{\Phi_{\rm CH},j}$ to $\Omega_{X,j}$ for $j=1,3,4$, respectively; see Figures \ref{fig:RHP Phi} and \ref{fig:RHP X} for illustrations of the relevant regions.

\begin{figure}
\centering
\tikzset{every picture/.style={line width=0.75pt}} 
\begin{tikzpicture}[x=0.75pt,y=0.75pt,yscale=-1,xscale=1]
\draw    (84,126) -- (123.5,126) ;
\draw [shift={(123.5,126)}, rotate = 0] [color={rgb, 255:red, 0; green, 0; blue, 0 }  ][fill={rgb, 255:red, 0; green, 0; blue, 0 }  ][line width=0.75]      (0, 0) circle [x radius= 3.35, y radius= 3.35]   ;
\draw [shift={(109.75,126)}, rotate = 180] [color={rgb, 255:red, 0; green, 0; blue, 0 }  ][line width=0.75]    (10.93,-3.29) .. controls (6.95,-1.4) and (3.31,-0.3) .. (0,0) .. controls (3.31,0.3) and (6.95,1.4) .. (10.93,3.29)   ;
\draw [shift={(84,126)}, rotate = 0] [color={rgb, 255:red, 0; green, 0; blue, 0 }  ][fill={rgb, 255:red, 0; green, 0; blue, 0 }  ][line width=0.75]      (0, 0) circle [x radius= 3.35, y radius= 3.35]   ;
\draw    (33.5,75) -- (84,126) ;
\draw [shift={(62.97,104.76)}, rotate = 225.28] [color={rgb, 255:red, 0; green, 0; blue, 0 }  ][line width=0.75]    (10.93,-3.29) .. controls (6.95,-1.4) and (3.31,-0.3) .. (0,0) .. controls (3.31,0.3) and (6.95,1.4) .. (10.93,3.29)   ;
\draw    (34.5,175) -- (84,126) ;
\draw [shift={(63.51,146.28)}, rotate = 135.29] [color={rgb, 255:red, 0; green, 0; blue, 0 }  ][line width=0.75]    (10.93,-3.29) .. controls (6.95,-1.4) and (3.31,-0.3) .. (0,0) .. controls (3.31,0.3) and (6.95,1.4) .. (10.93,3.29)   ;
\draw    (123.5,126) -- (162.5,126) ;
\draw [shift={(162.5,126)}, rotate = 0] [color={rgb, 255:red, 0; green, 0; blue, 0 }  ][fill={rgb, 255:red, 0; green, 0; blue, 0 }  ][line width=0.75]      (0, 0) circle [x radius= 3.35, y radius= 3.35]   ;
\draw [shift={(149,126)}, rotate = 180] [color={rgb, 255:red, 0; green, 0; blue, 0 }  ][line width=0.75]    (10.93,-3.29) .. controls (6.95,-1.4) and (3.31,-0.3) .. (0,0) .. controls (3.31,0.3) and (6.95,1.4) .. (10.93,3.29)   ;
\draw    (162.5,126) -- (203.5,126) ;
\draw [shift={(203.5,126)}, rotate = 0] [color={rgb, 255:red, 0; green, 0; blue, 0 }  ][fill={rgb, 255:red, 0; green, 0; blue, 0 }  ][line width=0.75]      (0, 0) circle [x radius= 3.35, y radius= 3.35]   ;
\draw [shift={(189,126)}, rotate = 180] [color={rgb, 255:red, 0; green, 0; blue, 0 }  ][line width=0.75]    (10.93,-3.29) .. controls (6.95,-1.4) and (3.31,-0.3) .. (0,0) .. controls (3.31,0.3) and (6.95,1.4) .. (10.93,3.29)   ;
\draw  [dash pattern={on 0.84pt off 2.51pt}]  (203.5,126) -- (213.97,126.27) -- (283.5,127) ;
\draw    (320.5,127) -- (321,239) ;
\draw [shift={(320.78,189)}, rotate = 269.74] [color={rgb, 255:red, 0; green, 0; blue, 0 }  ][line width=0.75]    (10.93,-3.29) .. controls (6.95,-1.4) and (3.31,-0.3) .. (0,0) .. controls (3.31,0.3) and (6.95,1.4) .. (10.93,3.29)   ;
\draw    (441.5,127) -- (482.5,127) ;
\draw [shift={(482.5,127)}, rotate = 0] [color={rgb, 255:red, 0; green, 0; blue, 0 }  ][fill={rgb, 255:red, 0; green, 0; blue, 0 }  ][line width=0.75]      (0, 0) circle [x radius= 3.35, y radius= 3.35]   ;
\draw [shift={(468,127)}, rotate = 180] [color={rgb, 255:red, 0; green, 0; blue, 0 }  ][line width=0.75]    (10.93,-3.29) .. controls (6.95,-1.4) and (3.31,-0.3) .. (0,0) .. controls (3.31,0.3) and (6.95,1.4) .. (10.93,3.29)   ;
\draw [shift={(441.5,127)}, rotate = 0] [color={rgb, 255:red, 0; green, 0; blue, 0 }  ][fill={rgb, 255:red, 0; green, 0; blue, 0 }  ][line width=0.75]      (0, 0) circle [x radius= 3.35, y radius= 3.35]   ;
\draw    (482.5,127) -- (526.5,127) ;
\draw [shift={(526.5,127)}, rotate = 0] [color={rgb, 255:red, 0; green, 0; blue, 0 }  ][fill={rgb, 255:red, 0; green, 0; blue, 0 }  ][line width=0.75]      (0, 0) circle [x radius= 3.35, y radius= 3.35]   ;
\draw [shift={(510.5,127)}, rotate = 180] [color={rgb, 255:red, 0; green, 0; blue, 0 }  ][line width=0.75]    (10.93,-3.29) .. controls (6.95,-1.4) and (3.31,-0.3) .. (0,0) .. controls (3.31,0.3) and (6.95,1.4) .. (10.93,3.29)   ;
\draw  [dash pattern={on 0.84pt off 2.51pt}]  (362.5,127) -- (441.5,127) ;
\draw    (283.5,127) -- (320.5,127) ;
\draw [shift={(320.5,127)}, rotate = 0] [color={rgb, 255:red, 0; green, 0; blue, 0 }  ][fill={rgb, 255:red, 0; green, 0; blue, 0 }  ][line width=0.75]      (0, 0) circle [x radius= 3.35, y radius= 3.35]   ;
\draw [shift={(308,127)}, rotate = 180] [color={rgb, 255:red, 0; green, 0; blue, 0 }  ][line width=0.75]    (10.93,-3.29) .. controls (6.95,-1.4) and (3.31,-0.3) .. (0,0) .. controls (3.31,0.3) and (6.95,1.4) .. (10.93,3.29)   ;
\draw [shift={(283.5,127)}, rotate = 0] [color={rgb, 255:red, 0; green, 0; blue, 0 }  ][fill={rgb, 255:red, 0; green, 0; blue, 0 }  ][line width=0.75]      (0, 0) circle [x radius= 3.35, y radius= 3.35]   ;
\draw    (320.5,127) -- (362.5,127) ;
\draw [shift={(362.5,127)}, rotate = 0] [color={rgb, 255:red, 0; green, 0; blue, 0 }  ][fill={rgb, 255:red, 0; green, 0; blue, 0 }  ][line width=0.75]      (0, 0) circle [x radius= 3.35, y radius= 3.35]   ;
\draw [shift={(347.5,127)}, rotate = 180] [color={rgb, 255:red, 0; green, 0; blue, 0 }  ][line width=0.75]    (10.93,-3.29) .. controls (6.95,-1.4) and (3.31,-0.3) .. (0,0) .. controls (3.31,0.3) and (6.95,1.4) .. (10.93,3.29)   ;
\draw    (526.5,127) -- (570.5,127) ;
\draw [shift={(570.5,127)}, rotate = 0] [color={rgb, 255:red, 0; green, 0; blue, 0 }  ][fill={rgb, 255:red, 0; green, 0; blue, 0 }  ][line width=0.75]      (0, 0) circle [x radius= 3.35, y radius= 3.35]   ;
\draw [shift={(554.5,127)}, rotate = 180] [color={rgb, 255:red, 0; green, 0; blue, 0 }  ][line width=0.75]    (10.93,-3.29) .. controls (6.95,-1.4) and (3.31,-0.3) .. (0,0) .. controls (3.31,0.3) and (6.95,1.4) .. (10.93,3.29)   ;
\draw    (570.5,127) -- (620,78) ;
\draw [shift={(599.51,98.28)}, rotate = 135.29] [color={rgb, 255:red, 0; green, 0; blue, 0 }  ][line width=0.75]    (10.93,-3.29) .. controls (6.95,-1.4) and (3.31,-0.3) .. (0,0) .. controls (3.31,0.3) and (6.95,1.4) .. (10.93,3.29)   ;
\draw    (570.5,127) -- (621,178) ;
\draw [shift={(599.97,156.76)}, rotate = 225.28] [color={rgb, 255:red, 0; green, 0; blue, 0 }  ][line width=0.75]    (10.93,-3.29) .. controls (6.95,-1.4) and (3.31,-0.3) .. (0,0) .. controls (3.31,0.3) and (6.95,1.4) .. (10.93,3.29)   ;
\draw (77,132.4) node [anchor=north west][inner sep=0.75pt]    {$a_{0}$};
\draw (117,130.4) node [anchor=north west][inner sep=0.75pt]    {$b_{0}$};
\draw (157.5,132.4) node [anchor=north west][inner sep=0.75pt]    {$a_{1}$};
\draw (195,129.4) node [anchor=north west][inner sep=0.75pt]    {$b_{1}$};
\draw (277,132.4) node [anchor=north west][inner sep=0.75pt]    {$a_{m}$};
\draw (316,109.4) node [anchor=north west][inner sep=0.75pt]    {$0$};
\draw (353.5,129.4) node [anchor=north west][inner sep=0.75pt]    {$b_{m}$};
\draw (429.25,132.4) node [anchor=north west][inner sep=0.75pt]    {$a_{n-1}$};
\draw (471.25,130.4) node [anchor=north west][inner sep=0.75pt]    {$b_{n-1}$};
\draw (518.25,132.4) node [anchor=north west][inner sep=0.75pt]    {$a_{n}$};
\draw (564.25,131.4) node [anchor=north west][inner sep=0.75pt]    {$b_{n}$};
\draw (599,55.4) node [anchor=north west][inner sep=0.75pt]    {$\Gamma _{X,r}$};
\draw (602,183.4) node [anchor=north west][inner sep=0.75pt]    {$\Gamma _{X,r}^{*}$};
\draw (41,57.4) node [anchor=north west][inner sep=0.75pt]    {$\Gamma _{X,l}$};
\draw (47,164.4) node [anchor=north west][inner sep=0.75pt]    {$\Gamma _{X,l}^{*}$};
\draw (307,62.4) node [anchor=north west][inner sep=0.75pt]    {$\Omega _{X,1}$};
\draw (30,112.4) node [anchor=north west][inner sep=0.75pt]    {$\Omega _{X,2}$};
\draw (192,182.4) node [anchor=north west][inner sep=0.75pt]    {$\Omega _{X,3}$};
\draw (433,182.4) node [anchor=north west][inner sep=0.75pt]    {$\Omega _{X,4}$};
\draw (596,118.4) node [anchor=north west][inner sep=0.75pt]    {$\Omega _{X,5}$};
\draw (322,227.4) node [anchor=north west][inner sep=0.75pt]    {$\Gamma _{X,d}$};
\end{tikzpicture}
\caption{The jump contours and regions of the RH problem for $X$.}\label{fig:RHP X}
\end{figure}
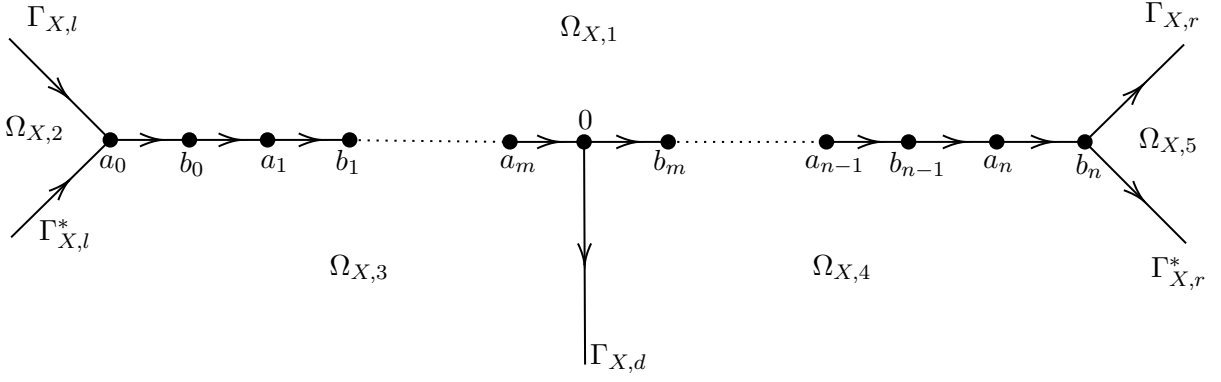
On account of the RH problems for $\Phi_{\rm CH}$ and $Y$, it is straightforward to check that $X$ satisfies the following RH problem.
  \paragraph{RH problem for $X$}
  \begin{itemize} 
      \item[\rm (a)] $X(z)$ is holomorphic for $z\in\mathbb{C}\backslash\Gamma_{X}$, 
     where 
      \begin{align}\label{def:Gamma_X}
      \Gamma_{X}:=\cup_{j\in\{r,l\}}\left(\Gamma_{X,j}\cup\Gamma_{X,j}^*\right)\cup\Gamma_{X,d}\cup[a_0,b_n];
      \end{align}
see Figure \ref{fig:RHP X} for an illustration. 
  
  \item[\rm (b)] For $z\in\Gamma_{X}$, we have 
  \begin{equation}\label{eq:Xjump}
    X_{+}(z)=X_{-}(z)J_{X}(z),
  \end{equation}
  where
  \begin{equation}\label{equ:jump of X}
      J_X(z)=\left\{\begin{array}{ll}
  {\begin{pmatrix}
  1 & 0 \\
  e^{-(\alpha-\beta)\pi i} & 1
  \end{pmatrix}},& {z \in \Gamma_{X,r}}, \\
  {\begin{pmatrix}
      1 & -e^{(\alpha-\beta)\pi i} \\
  0 & 1
  \end{pmatrix}}
  ,& {z \in \Gamma_{X,r}^*} ,\\
  {\begin{pmatrix}
      1 & 0 \\
  e^{(\alpha-\beta)\pi i} & 1
  \end{pmatrix}
  
 }, & { z \in \Gamma_{X,l}}, \\
  {\begin{pmatrix}
      1 & -e^{-(\alpha-\beta)\pi i} \\
  0 & 1
  \end{pmatrix}
  
  },& { z \in \Gamma_{X,l}^*}, \\
  {e^{2 \beta \pi i \sigma_3}},& {z \in \Gamma_{X,d} },\\
  {\begin{pmatrix}
       0 & -e^{-(\alpha-\beta)\pi i } \\
  e^{(\alpha-\beta)\pi i } & 0
  \end{pmatrix}
 
 },& {z \in \cup_{j=0}^{m-1}(a_j,b_j)}\cup(a_m,0), \\
  {\begin{pmatrix}
      0 & -e^{-(\alpha-\beta)\pi i } \\
    e^{(\alpha-\beta)\pi i } & 1
  \end{pmatrix}}   
    ,& z \in \cup_{j=0}^{m-1}(b_j,a_{j+1}), \\
  {\begin{pmatrix}
      0 & -e^{(\alpha-\beta)\pi i } \\
  e^{-(\alpha-\beta)\pi i } & 0
  \end{pmatrix}}
  ,& {z \in \cup_{j=m+1}^{n}(a_j,b_j)}\cup(0,b_m),\\
  {\begin{pmatrix}
      0 & -e^{(\alpha-\beta)\pi i } \\
    e^{-(\alpha-\beta)\pi i } & 1 
  \end{pmatrix}},& z \in \cup_{j=m}^{n-1}(b_j,a_{j+1}).
  \end{array}\right.
  \end{equation}
  
  
  \item[\rm (c)]   As $z\rightarrow \infty$, we have 
  \begin{equation}\label{equ:X asym at inf}
      X(z)=\left(I+\frac{X_1(s)}{z}+\mathcal{O}\left(z^{-2}\right)\right) z^{-\beta \sigma_3} e^{-isz \sigma_3},
  \end{equation}
  where 
  \begin{align}\label{equ: X_1(s)}
  X_{1}(s)=(2s)^{\beta\sigma_3}
  \begin{pmatrix}
    \frac{1}{2s}\left(\Phi_{{\rm CH},1}\right)_{11}+\frac{1}{s}{\left(Y_{1}(s)\right)}_{11} & \frac{1}{2s}(\Phi_{{\rm CH},1})_{12}+\frac{1}{s}{(Y_{1}(s))}_{12}e^{-\pi i\beta}\\
    \frac{1}{2s}(\Phi_{{\rm CH},1})_{21}+\frac{1}{s}{(Y_{1}(s))}_{21}e^{\pi i\beta} & \frac{1}{2s}(\Phi_{{\rm CH},1})_{22}+\frac{1}{s}{(Y_{1}(s))}_{22}
  \end{pmatrix}(2s)^{-\beta\sigma_3},
  \end{align}
  with $\Phi_{{\rm CH},1}$ and $Y_1$ given in \eqref{def:PhiCH1} and \eqref{equ: Y_1(s)}, respectively.

  \item [\rm (d)]As $z\rightarrow p$ from $\Omega_{X,1}$ with $p\in\mathcal{I}_{e}$, we have
  \begin{align}\label{equ:X asym at a_j,b_j-1}
      X(z)=X_{0}^{(p)}(z)
      \begin{pmatrix}
          1 & \sign{p}\cdot\frac{-1}{2\pi i}e^{\sign{p}(\alpha-\beta)\pi i} \log \left(z-p\right) \\
  0 & 1
      \end{pmatrix},
  \end{align}
where $X_0^{(p)}(z)$ is analytic near $z=p$, the branch cut of $\log(z-p)$ is taken to be the positive and negative real axis for $p<0$ and $p>0$, respectively. 
The behavior of $X(z)$ as $z\rightarrow p$ from other regions is determined by using the jump condition \eqref{eq:Xjump} and \eqref{equ:jump of X}.
  
\item[\rm (e)]   As $z \rightarrow 0$ from $\Omega_{X,1}$, we have
\begin{equation}\label{equ:X asym at 0}
X(z)=X_0(z)z^{\alpha\sigma_3},
\end{equation}
where $X_0(z)$ is analytic near $z=0$ and $z^\alpha$ takes the principal branch. The behavior of $X(z)$ as $z\rightarrow 0$  from other regions is  determined by using the jump condition \eqref{eq:Xjump} and \eqref{equ:jump of X}.
\end{itemize}

\subsection{The differential identity}\label{subsec: differential identity}
As in \cite{Fahs-Krasovsky-2024CPAM} for the sine kernel determinant or in \cite{Blackstone-CC-JL-RMTA,Blackstone-CC-JL-IMRN,Krasovsky-Maroudas-AdV2024} for the Airy kernel determinant, one could relate $\partial_s\log \mathcal{F}(s\Sigma)$ to the local behavior of $X$ at $z=p\in \mathcal{I}_e$. This differential identity, however, will create difficulties in further simplification of final asymptotic formula, especially for $n>1$. Instead, we are able to establish a differential identity involving the large-$z$ behavior of $X$ as stated in the following proposition. 
\begin{proposition}\label{prop: differential identity}
With $X_1$ given in \eqref{equ: X_1(s)}, we have  
    \begin{align}\label{eq:Fdiff}
    \partial_s\log\mathcal{F}(s\Sigma)=i\left((X_{1}(s))_{11}-(X_{1}(s))_{22}\right)-\frac{\alpha^2-\beta^2}{s},
    \end{align}
where $\mathcal{F}(\Sigma)$ is defined in \eqref{equ:CHF Fredholm det}.
\end{proposition} 
A key ingredient in the proof of the above proposition is the following differential identity for the scaled confluent hypergeometric kernel.
\begin{lemma} \label{lem: partial_s scaled CHF kernel}
Let $K^{(\alpha,\beta)}_s(x,y)$ be the scaled  confluent hypergeometric kernel given in \eqref{equ: scaled kernel}, we have 
\begin{multline}\label{equ:patial_s scaled kernel in terms of equivalent CHF kernel}
\partial_{s}K_{s}^{(\alpha,\beta)}(x,y)=\frac{1}{\pi}\frac{\Gamma(1+\alpha+\beta)\Gamma(1+\alpha-\beta)}{(1+2\alpha)\Gamma(1+2\alpha)^2}{\chi}_{\beta}(x)^{1/2}{\chi}_{\beta}(y)^{1/2}e^{-i(x+y)}4^{\alpha}|xy|^{\alpha}\\
\times\left[\left(1+2\alpha\right)\mathcal{Q}(sx)\mathcal{Q}(sy)
+s^2\frac{\alpha^2-\beta^2}{\alpha^2(1+2\alpha)}x\mathcal{P}(sx)\cdot y\mathcal{P}(sy)- is\frac{\beta}{\alpha}\left(x\mathcal{P}(sx)\mathcal{Q}(sy)+y\mathcal{P}(sy)\mathcal{Q}(sx)\right)\right],
\end{multline}
where the functions $\mathcal{P}$ and $\mathcal{Q}$ are given in 
\eqref{def:mathcalP and mathcalQ}. Equivalently, one has
\begin{align}\label{equ:patial_s scaled kernel in terms of A and B}
    \partial_{s}K_{s}^{(\alpha,\beta)}(x,y)=\frac{1}{2\pi}\frac{\Gamma(1+\alpha+\beta)\Gamma(1+\alpha-\beta)}{{\Gamma(1+2\alpha)}^2}\left(A(sx)B(sy)+A(sy)B(sx)\right),
\end{align}
where the functions $A$ and $B$ are given in \eqref{def: A(x)} and \eqref{def: B(x)}, respectively.
\end{lemma}
\begin{proof}
A direct computation shows that 
\begin{align}\label{equ:computations of partialsKs(x,y)}
&\partial_sK_{s}^{(\alpha,\beta)}(x,y)
\nonumber
\\
&=(2\alpha+1)K^{(\alpha,\beta)}(sx,sy)-i(x+y)sK^{(\alpha,\beta)}(sx,sy)\nonumber\\
&~~+\frac{1}{\pi}\frac{\Gamma(1+\alpha+\beta)\Gamma(1+\alpha-\beta)}{(1+2\alpha)\Gamma(1+2\alpha)^2}{\chi}_{\beta}(x)^{1/2}{\chi}_{\beta}(y)^{1/2}e^{-i(x+y)}4^{\alpha}|xy|^{\alpha}\nonumber\\
&~~\times s\left[2ix^2\mathcal{P}'(sx)\mathcal{Q}(sy)-2iy^2\mathcal{P}'(sy)\mathcal{Q}(sx)+
2ixy\left(\mathcal{P}(sx)\mathcal{Q}'(sy)-\mathcal{P}(sy)\mathcal{Q}'(sx)\right)\right].
\end{align}
To evaluate the functions $\mathcal{P}'(rx)$ and $\mathcal{Q}'(rx)$, we note that the 
confluent hypergeometric function $\phi$ satisfies the derivative formula (see \cite[Equation (13.3.15)]{NISTbook})
\begin{align}
    \frac{\dif}{\dif z}\phi(a,b;z)=\frac{a}{b}\phi(1+a,1+b;z), \label{equ:CHF relation-1}
\end{align}
the recurrence relations (see \cite[Equations (13.3.2)--(13.3.4)]{NISTbook}) 
\begin{align}
& b(b-1)\phi(a,b-1;z)+b(1-b-z)\phi(a,b;z)+z(b-a)\phi(a,b+1;z)=0, \label{equ:CHF relation-2}\\
& (a-b+1)\phi(a,b;z)-a\phi(a+1,b;z)+(b-1)\phi(a,b-1;z)=0,\label{equ:CHF relation-3}\\
& b\phi(a,b;z)-b\phi(a-1,b;z)-z\phi(a,b+1;z)=0, \label{equ:CHF relation-4}
\end{align}
and an equivalent form of Kummer's differential equation (see \cite[Equation (13.3.13)]{NISTbook})
\begin{align}\label{equ:CHF relation-5-Kummer}
(a+1)z\phi(a+2,b+2;z)+(b+1)(b-z)\phi(a+1,b+1;z)-b(b+1)\phi(a,b;z)=0.
\end{align}
With the aid of the above formulae, we have
\begin{align}\label{equ: computation of mathcalP'(sx)}
\mathcal{P}'(sx)&\overset{\eqref{def:mathcalP and mathcalQ}}{=}\phi'(1+\alpha+\beta,2+2\alpha;2isx)\nonumber\\
&\overset{\eqref{equ:CHF relation-1}}{=}\frac{1+\alpha+\beta}{2+2\alpha}\phi(2+\alpha+\beta,3+2\alpha;2isx)\nonumber\\
&\overset{\eqref{equ:CHF relation-5-Kummer}}{=}\frac{1+2\alpha}{2isx}\phi(\alpha+\beta,1+2\alpha;2isx)-\frac{1+2\alpha-2isx}{2isx}\phi(1+\alpha+\beta,2\alpha+2;2isx)\nonumber\\
&\overset{\eqref{equ:CHF relation-3}}{=}\frac{(1+2\alpha)(\alpha+\beta)}{2isx(\beta-\alpha)}\phi(1+\alpha+\beta,1+2\alpha;2isx)-\frac{2\alpha(1+2\alpha)}{2isx(\beta-\alpha)}\phi(\alpha+\beta,2\alpha;2isx)\nonumber\\
&\qquad -\frac{1+2\alpha-2isx}{2isx}\phi(1+\alpha+\beta,2\alpha+2;2isx),\nonumber
\end{align}
or equivalently, by \eqref{def:mathcalP and mathcalQ}, 
\begin{equation}\label{equ:rep of mathcalP'(sx)}
    \mathcal{P}'(sx)=\frac{\alpha(1+2\alpha)}{isx(\beta-\alpha)}\mathcal{Q}'(sx)-\frac{\alpha(1+2\alpha)}{isx(\beta-\alpha)}\mathcal{Q}(sx)-\frac{1+2\alpha-2isx}{2isx}\mathcal{P}(sx).
\end{equation}
Similarly, 
\begin{align}
\mathcal{Q}'(sx)&\overset{\eqref{def:mathcalP and mathcalQ}}{=}\phi'(\alpha+\beta,2\alpha;2isx)\nonumber\\
&\overset{\eqref{equ:CHF relation-1}}=\frac{\alpha+\beta}{2\alpha}\phi(1+\alpha+\beta,1+2\alpha;2isx) \nonumber\\
&\overset{\eqref{equ:CHF relation-2}}{=}\frac{\alpha+\beta}{2\alpha+2isx}\phi(1+\alpha+\beta,2\alpha;2isx)+\frac{isx(\alpha^2-\beta^2)}{2\alpha(1+2\alpha)(\alpha+isx)}\phi(1+\alpha+\beta,2+2\alpha;2isx)\nonumber\\
&\overset{\eqref{equ:CHF relation-4}}{=}\frac{\alpha+\beta}{2\alpha(2\alpha+2isx)}\left[2\alpha\phi(\alpha+\beta,2\alpha;2isx)+2isx\phi(1+\alpha+\beta,1+2\alpha;2isx)\right]\nonumber
\\
&\qquad +\frac{isx(\alpha^2-\beta^2)}{2\alpha(1+2\alpha)(\alpha+isx)}\phi(1+\alpha+\beta,2+2\alpha;2isx). \nonumber
\end{align}
This, together with \eqref{def:mathcalP and mathcalQ}, implies that
\begin{align}\label{equ:rep of mathcalQ'(sx)}
\frac{\alpha}{\alpha+isx}\mathcal{Q}'(sx)=\frac{isx(\alpha^2-\beta^2)}{2\alpha(1+2\alpha)(\alpha+isx)}\mathcal{P}(sx)+\frac{\alpha+\beta}{2\alpha+2isx}\mathcal{Q}(sx).
\end{align}
By \eqref{equ:rep of mathcalP'(sx)} and \eqref{equ:rep of mathcalQ'(sx)}, it follows that
\begin{align}
&\mathcal{P}'(sx)=\left(1-\frac{\alpha+\beta}{2\alpha}-\frac{1+2\alpha}{2isx}\right)\mathcal{P}(sx)+\frac{1+2\alpha}{2isx}\mathcal{Q}(sx),\label{equ:expression of mathcalP'(sx)} \\
&\mathcal{Q}'(sx)=\frac{isx(\alpha^2-\beta^2)}{2\alpha^2(1+2\alpha)}\mathcal{P}(sx)+\frac{\alpha+\beta}{2\alpha}\mathcal{Q}(sx).\label{equ:expression of mathcalQ'(sx)}
\end{align}
Substituting \eqref{equ:expression of mathcalP'(sx)} and \eqref{equ:expression of mathcalQ'(sx)} into \eqref{equ:computations of partialsKs(x,y)}, we obtain 
\eqref{equ:patial_s scaled kernel in terms of equivalent CHF kernel} after straightforward calculations. Moreover, by noticing the 
relations \eqref{equ:xmathcalP(x) relation to A and B} and \eqref{equ:xmathcalQ(x) relation to A and B}, the equivalent form \eqref{equ:patial_s scaled kernel in terms of A and B} follows.
\end{proof}

\begin{remark}
    If $\alpha=\beta=0$, Lemma \ref{lem: partial_s scaled CHF kernel} simply reads
    $\partial_{s}K_{s}^{(0,\hspace*{0.1em} 0)}(x,y)=\frac{\cos (s(x-y))}{\pi}$. 
\end{remark}
\paragraph{Proof of Proposition \ref{prop: differential identity}.} 
By \eqref{equ:patial_s scaled kernel in terms of A and B}, we have 
\begin{align}
\partial_{s}K^{(\alpha,\hspace*{0.1em}\beta)}_{s}(x,y)=i{\vec{f}(sx)}^{\rm T}
\begin{pmatrix}
    -1 & 0 \\
    0  & 1
\end{pmatrix}
\vec{h}(sy),
\end{align}
where $\vec{f}$ and $\vec{h}$ are defined by \eqref{def:f and h}. Thus, the kernel of the operator $\left(1-\mathcal{K}^{(\alpha,\hspace*{0.1em}\beta)}_{s}|_{\Sigma}\right)^{-1}\partial_s\mathcal{K}^{(\alpha,\hspace*{0.1em}\beta)}_{s}|_{\Sigma}$ in \eqref{equ: partialsF(sSigma)=Trace} is given by 
\begin{align}\label{equ:kernel of a necessary trace operator}
i{\vec{F}(sx)}^{\rm T}
\begin{pmatrix}
    -1 & 0 \\
    0  & 1
\end{pmatrix}
\vec{h}(sy)=-iF_1(sx)h_1(sy)+iF_{2}(sx)h_2(sy),
\end{align}
where 
\begin{align*}
\vec{F}(u)=\begin{pmatrix}
        F_1(u) \\
        F_2(u)
    \end{pmatrix}:=\left(1-\mathcal{K}^{(\alpha,\hspace*{0.1em}\beta)}|_{s\Sigma}\right)^{-1}\vec{f}(u).
\end{align*}
This, together with \eqref{equ: partialsF(sSigma)=Trace}, implies that 
\begin{multline}\label{equ: partialsF(sSigma) relation to Fh}
\partial_{s}\log \mathcal{F}(s\Sigma)=-{\rm Tr}\left(\left(1-\mathcal{K}^{(\alpha,\hspace*{0.1em}\beta)}_{s}|_{\Sigma}\right)^{-1}\partial_s\mathcal{K}^{(\alpha,\hspace*{0.1em}\beta)}_{s}|_{\Sigma}\right)\\
\overset{\eqref{equ:kernel of a necessary trace operator}}{=}i\int_{\Sigma}F_1(sx)h_1(sx)-F_2(sx)h_2(sx)\dif x=\frac{i}{s}\int_{s\Sigma}F_1(x)h_1(x)-F_2(x)h_2(x)\dif x.
\end{multline}
In view of \eqref{equ: Y_1(s)}, it follows that  
\begin{align}\label{equ:partialsF(sSigma) = Y}
\partial_{s}\log \mathcal{F}(s\Sigma)=\frac{i}{s}\left({(Y_{1}(s))}_{11}-{(Y_{1}(s))}_{22}\right).
\end{align}
By \eqref{equ: X_1(s)}, we have
\begin{align*}
\left(Y_1(s)\right)_{11}=s\left(X_1(s)\right)_{11}-\frac{1}{2}\left(\Phi_{{\rm CH},1}\right)_{11}, \quad 
\left(Y_1(s)\right)_{22}=s\left(X_1(s)\right)_{22}-\frac{1}{2}\left(\Phi_{{\rm CH},1}\right)_{22}.
\end{align*}
Substituting the above formulae into \eqref{equ:partialsF(sSigma) = Y}, we finally arrive at \eqref{eq:Fdiff} by noting that $\left(\Phi_{{\rm CH},1}\right)_{22}=-\left(\Phi_{{\rm CH},1}\right)_{11}=(\alpha^2-\beta^2)i$.
\qed 

\section{Asymptotic analysis of the RH problem for $X$ -- the first two transformations}
\label{sec: DZ steepest analysis}
We intend to perform Deift/Zhou steepest analysis to the  RH problem for $X$, and start with an introduction of the so-called $g$-fucntion.


\subsection{The $g$-function}
Recall the function $\mathcal{R}(z)$ defined in \eqref{equ:sqrt(R(z))}, it is easily  seen that 
\begin{align}
&\sqrt{\mathcal{R}(z)}_{+}+\sqrt{\mathcal{R}(z)}_{-}=0, && \ z\in\cup_{j=0}^{n}(a_j,b_j), \label{equ:mathcalR-jump-relation1}\\
&\sqrt{\mathcal{R}(z)}_{+}=\sqrt{\mathcal{R}(z)}_{-}, && \ z\in\cup_{j=0}^{n-1}(b_j,a_{j+1}),
\label{equ:mathcalR-jump-relation2}
\end{align}
where $\sqrt{\mathcal{R}(z)}_{+}$ and $\sqrt{\mathcal{R}(z)}_{-}$ denote the non-tangential limits of 
$\sqrt{\mathcal{R}(z')}$ as $z'\to z$ from the upper and lower half-plane of the first sheet of $\mathcal{W}$, respectively. We then introduce
\begin{align}\label{equ:g-func}
    g(z):=\int_{a_0}^{z}\frac{\mathrm{p}(x)}{\sqrt{\mathcal{R}(x)}}\dif x,
\end{align}
where the polynomial $\mathrm{p}(z)$ is defined through \eqref{equ:mathbf{p}(z)} and \eqref{equ:A-cycle=0}. Some properties of the $g$-function are collected in the following proposition for later use. 

\begin{proposition}\label{prop:g-func}
The $g$-function defined in \eqref{equ:g-func} satisfies the following properties:
\begin{itemize}
\item[\rm (a)]$g(z)$ is holomorphic for $z\in\mathbb{C}\setminus\overline{\Sigma}$.
\item[\rm (b)]For $z\in\Sigma$, $g(z)$ satisfies the jump relation
\begin{align}\label{equ:g-jump-relation2}
 g_{+}(z)+g_{-}(z)=\Omega_j, \quad z\in(a_j,b_j), \ j=0,\dots, n,
\end{align}
where $\Omega_j$ is defined in \eqref{def-intro-Omega_j}. Moreover, we have
\begin{align}\label{equ:Omega_j}
\Omega_0=0  \quad {\rm and} \quad  \Omega_j=\sum_{k=0}^{j-1}\hat{\Omega}_k, \quad j=1,\dots, n, 
\end{align}
where
\begin{align}\label{equ:hatOmega_k}
\hat{\Omega}_k=2\int_{b_k}^{a_{k+1}}\frac{\mathrm{p}(s)}{\sqrt{\mathcal{R}(s)}}\dif s=2(-1)^{n-k}\int_{b_k}^{a_{k+1}}\frac{\mathrm{p}(s)}{|\mathcal{R}(s)|^{\frac{1}{2}}}\dif s>0. 
\end{align}
\item[\rm (c)] As $z\to\infty$, we have 
\begin{align}\label{eq:gAsy}
g(z)=z+\ell-\frac{\gamma_0}{z}+\mathcal{O}(z^{-2}),
\end{align}
where 
\begin{equation}\label{def:ell}
    \ell=\int_{a_0}^{\infty}\left(\frac{\mathrm{p}(s)}{\sqrt{\mathcal{R}(s)}}-1\right)\dif s-a_0
\end{equation}
and $\gamma_0$ is given by \eqref{equ: gamma_0}.

\item[\rm (d)] There exists a neighborhood $U$ of $\cup_{j=0}^{n-1}(b_j,a_{j+1})$ such that
\begin{equation}\label{eq:Uprop}
\im g(z)\geqslant 0 \ {\rm for} \ z\in U\cap\overline{\mathbb{C}^{+}}, \quad {\rm and} \quad 
\im g(z)\leqslant 0 \ {\rm for} \ z\in U\cap\overline{\mathbb{C}^{-}},
\end{equation}
where the equality holds only for $z\in\cup_{j=0}^{n-1}(b_j,a_{j+1})$; see Figure \ref{fig:signature table of img} for an illustration with $n=2$.
\end{itemize}
\end{proposition}

\begin{proof}
We start with the observation that \eqref{equ:A-cycle=0} can be rewritten as 
\begin{align}\label{equ:aj-bjInt=0}
\int_{a_j}^{b_j}\frac{\mathrm{p}(s)}{\sqrt{\mathcal{R}(s)}_{+}}\dif s=0, \quad j=0,\dots,n,
\end{align}
where the path lies on the first sheet of $\mathcal{W}$. By a straightforward calculation, it is clear that 
\begin{align}
\sqrt{\mathcal{R}(z)}_{+}=i(-1)^{n-j}|\mathcal{R}(z)|^{1/2}\in i\mathbb{R}, \quad z\in(a_j,b_j), \quad j=0,1,\dots,n.
\end{align}
Since $\mathrm{p}(z) \in \mathbb{R}$ for $z\in (a_j,b_j)$, it follows from \eqref{equ:aj-bjInt=0} that $\mathrm{p}$ must have at least one zero on the intervals $(a_j,b_j)$. Note that $\mathrm{p}(z)$ is a polynomial of degree $n+1$, we conclude that $\mathrm{p}(z)$ has exactly one simple zero on $(a_j,b_j)$, which  we denote by $x_j\in(a_j,b_j)$ for $j=0,1,\dots,n$. Note that $\mathrm{p}(z)\to+\infty$ as $z\to+\infty$, it follows that for $j=0,1,\dots,n$,
\begin{align}
\sign{\mathrm{p}(a_j)}=(-1)^{n-j+1}, 
\end{align}
and 
\begin{align}\label{equ: sgn(mathbf{p}(z))}
\sign{\mathrm{p}(z)}=(-1)^{n-j}, \qquad z\in(b_j,a_{j+1}).
\end{align}

We next prove our claims one by one. 
\begin{itemize}

\item[\rm (a)] It  follows directly from the definition of $g$ in \eqref{equ:g-func}. In particular, by \eqref{equ:mathcalR-jump-relation2} and \eqref{equ:aj-bjInt=0}, we 
have $g_{+}(x)=g_{-}(x)$ for $x \in \left(b_j,a_{j+1}\right)$, $j=0,\dots,n-1$.

    
\item[\rm (b)] 
If $z\in(a_0,b_0)$, we have $g_{+}(z)+g_{-}(z)=0$ by \eqref{equ:mathcalR-jump-relation1}. If $z\in(a_j,b_j)$, $j=1,\ldots,n$, it also follows from \eqref{equ:mathcalR-jump-relation2} that 
    \begin{align}\label{eq:g++g-}
    g_{+}(z)+g_{-}(z)=2\left(\int_{b_{0}}^{a_1}+\int_{b_{1}}^{a_{2}}+\dots\int_{b_{j-1}}^{a_{j}}\right)\frac{\mathrm{p}(s)}{\sqrt{\mathcal{R}(s)}}\dif s.
    \end{align}
By a direct calculation, we have
    \begin{align}
    \int_{b_k}^{a_{k+1}}\frac{\mathrm{p}(s)}{\sqrt{\mathcal{R}(s)}}\dif s=
    (-1)^{n-k}\int_{b_k}^{a_{k+1}}\frac{\mathrm{p}(s)}{|\mathcal{R}(s)|^{\frac{1}{2}}}\dif s, \quad k=0,\dots,n-1,
    \end{align}
This, together with \eqref{eq:g++g-}, implies \eqref{equ:g-jump-relation2}--\eqref{equ:hatOmega_k}. Finally, we observe from 
\eqref{equ:hatOmega_k} and \eqref{equ: sgn(mathbf{p}(z))}
that 
    \begin{align*}
    \sign {\hat{\Omega}_k}={(-1)}^{n-k}\cdot{(-1)}^{n-k}=1,
    \end{align*}
    thus $\hat{\Omega}_k>0$. This also shows $\Omega_j>0$ by using \eqref{equ:Omega_j}.

    \item[\rm (c)]
It follows from \eqref{equ:A-cycle=0} that $\mathrm{p}(z)/\sqrt{\mathcal{R}(z)}$ has 
no residue at infinity, which yields 
    \begin{align}\label{equ: asy of p/mathcalR at infty}
    \frac{\mathrm{p}(z)}{\sqrt{\mathcal{R}(z)}}=1+\frac{\gamma_0}{z^2}+\mathcal{O}\left(\frac{1}{z^3}\right), \qquad  z\to\infty,
    \end{align}
where $\gamma_0=\frac{1}{2\pi i}\oint_{\mathcal{C}}\frac{z\mathrm{p}(z)}{\sqrt{\mathcal{R}(z)}}\dif z$ with  $\mathcal{C}$ being any sufficiently large circle oriented in the counterclockwise manner. By a contour deformation, we see the equivalent definition of $\gamma_0$ given in \eqref{equ: gamma_0}. Since $\mathrm{p}$ is real-valued on the real line and $1/{\sqrt{\mathcal{R}(z)}}_{+}\in i\mathbb{R}$ for $z\in(a_j,b_j)$, $j=0,1,\dots,n$, it is immediate that $\gamma_0\in\mathbb{R}$. Inserting \eqref{equ: asy of p/mathcalR at infty} into \eqref{equ:g-func}, we obtain \eqref{eq:gAsy} and \eqref{def:ell}.

    
\item[\rm (d)] By \eqref{equ:g-func}, \eqref{equ:aj-bjInt=0} and \eqref{equ: sgn(mathbf{p}(z))}, one can check directly that $g(x)$ is real and strictly increasing on $\cup_{j=0}^{n-1}(b_j,a_{j+1})$. The existence of a neighborhood $U$ such that \eqref{eq:Uprop} holds is then assured by the Cauchy-Riemann equations; see also Figure \ref{fig:signature table of img} for an illustration. 
\end{itemize}
This completes the proof of Proposition \ref{prop:g-func}.
\end{proof}
\begin{figure}[htbp]
    \centering
    \includegraphics[height=6.28cm,width=9.99cm]{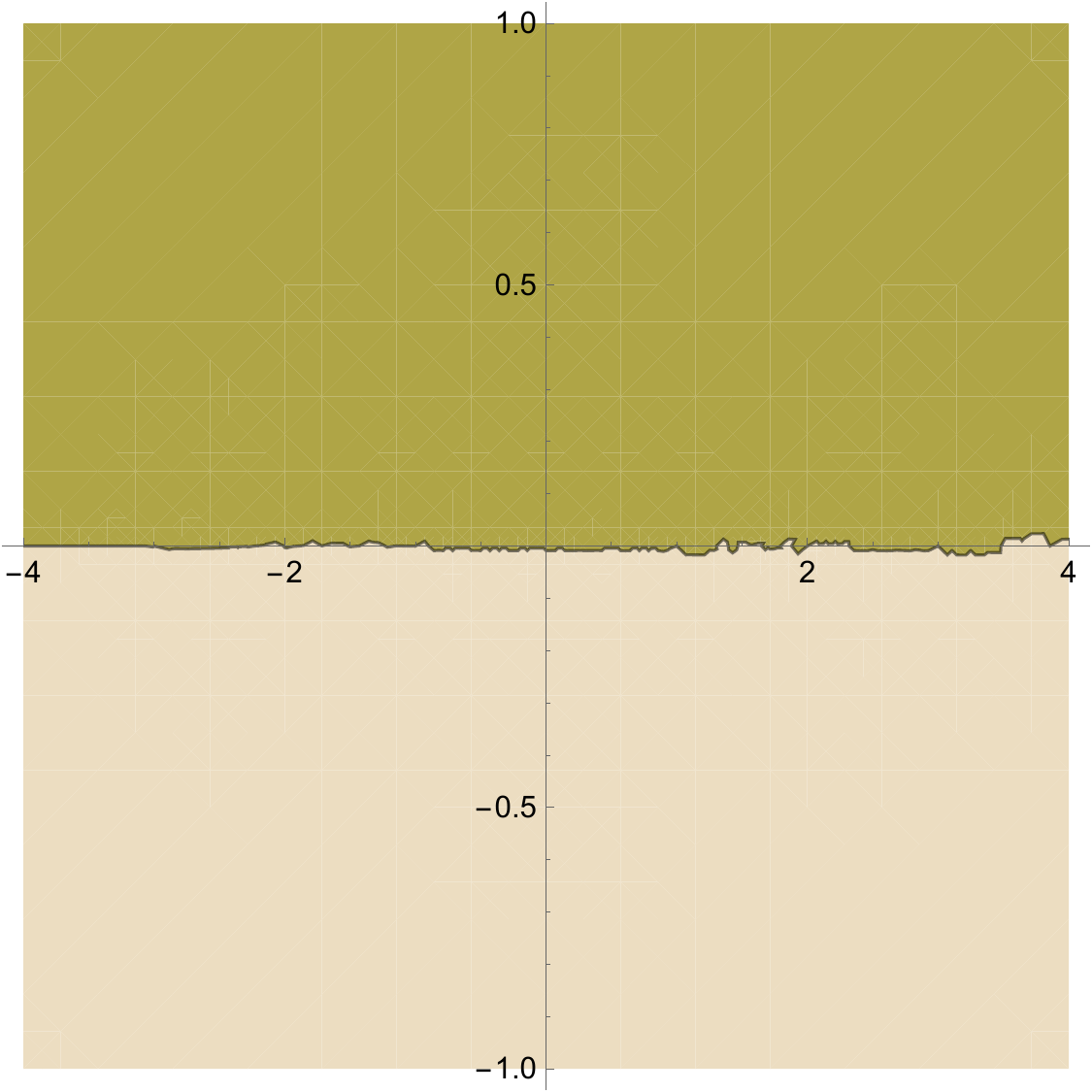}
    \caption{Signature table of $\im g$ on the complex plane for $n=2$ with $a_0=-3$, $b_0=-2$, $a_1=-1$, $b_1=1$, $a_2=2$, $b_2=3$ in \eqref{equ:g-func}. The dark and light colored parts denote where $\im g>0$ and $\im g<0$, respectively.}
    \label{fig:signature table of img}
\end{figure}

\subsection{First transformation: $X\to T$}
With the aid of the $g$-function defined in \eqref{equ:g-func}, the first transformation is defined by 
\begin{align}\label{def: T}
T(z):=e^{-is\ell\sigma_3}X(z)e^{isg(z)\sigma_3}, 
\end{align}
where $\ell$ is given in \eqref{def:ell}. 
Combining the RH problem for $X$ and Proposition \ref{prop:g-func}, it follows that $T$ satisfies the following RH problem.

\paragraph{RH problem for $T$}
\begin{itemize}
\item [\rm (a)] $T(z)$ is holomorphic for $z\in\mathbb{C}\backslash\Gamma_{X}$, 
      where $\Gamma_{X}$ is defined in \eqref{def:Gamma_X}.
      
\item [\rm (b)] For $z\in\Gamma_{X}$, we have 
  \begin{equation}
    T_{+}(z)=T_{-}(z)J_{T}(z),
  \end{equation}
  where
  \begin{equation}\label{equ:jump of T}
      J_T(z)=\left\{\begin{array}{ll}
  \begin{pmatrix}
  1 & 0 \\
  e^{-(\alpha-\beta)\pi i+2isg(z)} & 1
  \end{pmatrix},& {z \in \Gamma_{X,r}}, \\
  {\begin{pmatrix}
       1 & -e^{(\alpha-\beta)\pi i-2isg(z)} \\
  0 & 1
  \end{pmatrix}},& {z \in \Gamma_{X,r}^*} ,\\
  {\begin{pmatrix}
      1 & 0 \\
  e^{(\alpha-\beta)\pi i+2isg(z)} & 1
  \end{pmatrix}}, & { z \in \Gamma_{X,l}}, \\
  {\begin{pmatrix}
       1 & -e^{-(\alpha-\beta)\pi i-2isg(z)} \\
  0 & 1
  \end{pmatrix}},& { z \in \Gamma_{X,l}^*}, \\
  {e^{2 \beta \pi i \sigma_3}},& {z \in \Gamma_{X,d} },\\
  {\begin{pmatrix}
       0 & -e^{-(\alpha-\beta)\pi i-is\Omega_j } \\
  e^{(\alpha-\beta)\pi i+is\Omega_j } & 0
  \end{pmatrix}},& z \in \cup_{j=0}^{m-1}(a_j,b_j)\cup(a_m,0), \\
  {\begin{pmatrix}
       0 & -e^{-(\alpha-\beta)\pi i-2isg(z) } \\
    e^{(\alpha-\beta)\pi i+2isg(z) } & 1
  \end{pmatrix}},& z \in \cup_{j=0}^{m-1}(b_j,a_{j+1}), \\
  {\begin{pmatrix}
      0 & -e^{(\alpha-\beta)\pi i-is\Omega_j } \\
  e^{-(\alpha-\beta)\pi i+is\Omega_j } & 0
  \end{pmatrix}},& z \in \cup_{j=m+1}^{n}(a_j,b_j)\cup(0,b_m),\\
  {\begin{pmatrix}
      0 & -e^{(\alpha-\beta)\pi i-2isg(z) } \\
    e^{-(\alpha-\beta)\pi i+2isg(z) } & 1
  \end{pmatrix}},& z \in \cup_{j=m}^{n-1}(b_j,a_{j+1}).
  \end{array}\right.
\end{equation}
  

\item[\rm (c)]   As $z\rightarrow \infty$ with  $\arg z\in(-\frac{\pi}{2},\frac{3\pi}{2})$, we have
  \begin{equation}\label{equ:T asym at inf}
      T(z)=\left(I+\frac{T_1(s)}{z}+\mathcal{O}\left(z^{-2}\right)\right) z^{-\beta \sigma_3},
 \end{equation}
where $T_1$ is independent of $z$.

\item[\rm (d)] As $z \to p$ with $p\in\mathcal{I}_{e}$, we have
\begin{align}
T(z)=\mathcal{O}(\log(z-p)). 
\end{align}

\item[\rm (e)] As $z \rightarrow 0$ from $\Omega_{X,1}$, we have
  \begin{equation}\label{equ:T asym at 0}
      T(z)=T_0(z)z^{\alpha\sigma_3},
  \end{equation}
 where $T_0(z)$ is holomorphic in the neighborhood of $z=0$.

\end{itemize}

\subsection{Second transformation: $T\to S$}
The second transformation involves lenses-opening around $\cup_{j=0}^{n-1}(b_j,a_{j+1})$. 
To proceed, we note the following two decompositions of $J_{T}$ given in \eqref{equ:jump of T}.
\begin{align}\label{equ: decomposition of J_T-1}
   & \begin{pmatrix}
    0 & -e^{-(\alpha-\beta)\pi i-2isg(z) } \\
    e^{(\alpha-\beta)\pi i+2isg(z) } & 1
    \end{pmatrix}
    \nonumber \\ 
    &=
   \begin{pmatrix}
  1 & -e^{-(\alpha-\beta)\pi i-2isg(z)} \\
  0 & 1
  \end{pmatrix}
  \begin{pmatrix}
  1 & 0 \\
  e^{(\alpha-\beta)\pi i+2isg(z)} & 1
  \end{pmatrix},  \qquad  z\in\cup_{j=0}^{m-1}(b_j,a_{j+1}),
\end{align}
and
\begin{align}\label{equ: decomposition of J_T-2}
   & \begin{pmatrix}
    0 & -e^{(\alpha-\beta)\pi i-2isg(z) } \\
    e^{-(\alpha-\beta)\pi i+2isg(z) } & 1
    \end{pmatrix}
    \nonumber
    \\
    &=
  \begin{pmatrix}
  1 & -e^{(\alpha-\beta)\pi i-2isg(z)} \\
  0 & 1
  \end{pmatrix}
  \begin{pmatrix}
  1 & 0 \\
  e^{-(\alpha-\beta)\pi i+2isg(z)} & 1
  \end{pmatrix}, \qquad z\in\cup_{j=m}^{n-1}(b_j,a_{j+1}).
\end{align}

Let $\Gamma_{S,j}$ and $\Gamma_{S,j}^{*}$ be two curves starting at $b_j$ and ending at $a_{j+1}$ for $j=0,1, \dots, n-1$, where $\Gamma_{S,j}\subseteq \mathbb{C}^{+}$ and $\Gamma^{*}_{S,j}\subseteq \mathbb{C}^{-}$; see Figure \ref{fig:RHP S} for an illustration. We also denote the ``lens" domains which are delimited by boundaries $\Gamma_{S,j}$ and $\Gamma_{S,j}^{*}$ by $\mathrm{L}_j$. It is required that the angles between contours $\Gamma_{S,j}$, $\Gamma_{S,j}^{*}$ and the real axis are chosen to be sufficiently small such that $\mathrm{L}_j$ belongs
to the neighborhood $U$ given in item (d) of Proposition \ref{prop:g-func}.


The second transformation is defined by 
\begin{equation}\label{def: S}
   S(z)=T(z)\left\{\begin{array}{ll}
  \begin{pmatrix}
  1 & 0 \\
  -e^{(\alpha-\beta)\pi i+2isg(z)} & 1
  \end{pmatrix},& z \in \cup_{j=0}^{m-1}\mathrm{L}_j\cap\mathbb{C}^{+}, \\
  \begin{pmatrix}
       1 & 0 \\
  -e^{-(\alpha-\beta)\pi i+2isg(z)} & 1
  \end{pmatrix},& z \in \cup_{j=m}^{n-1}\mathrm{L}_j\cap\mathbb{C}^{+}, \\
  \begin{pmatrix}
      1 & -e^{-(\alpha-\beta)\pi i-2isg(z)} \\
  0 & 1
  \end{pmatrix},& z \in \cup_{j=0}^{m-1}\mathrm{L}_j\cap\mathbb{C}^{-} ,\\
  \begin{pmatrix}
      1 & -e^{(\alpha-\beta)\pi i-2isg(z)} \\
  0 & 1 
  \end{pmatrix},& z \in \cup_{j=m}^{n-1}\mathrm{L}_j\cap\mathbb{C}^{-} ,\\
  I,& {\rm elsewhere}.
  \end{array}\right. 
\end{equation}
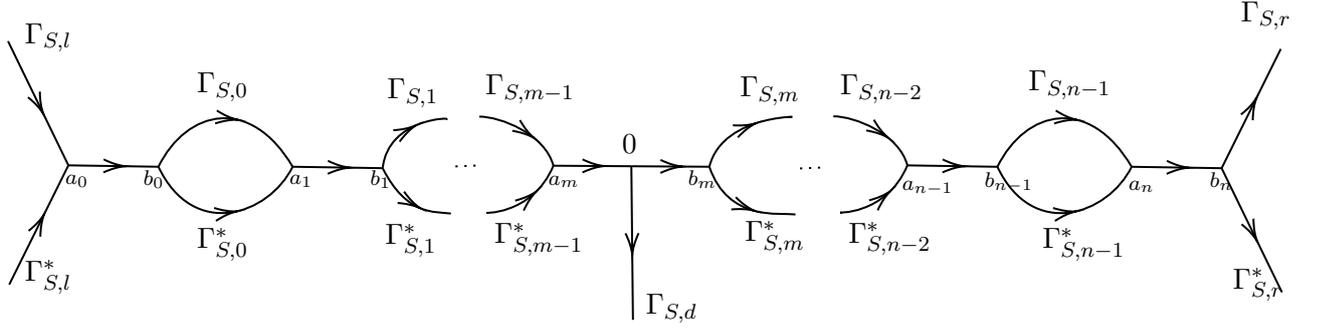
\begin{figure}
\tikzset{every picture/.style={line width=0.75pt}} 
\begin{tikzpicture}[x=0.75pt,y=0.75pt,yscale=-1,xscale=1]
\draw    (10.07,92.85) -- (40.21,155.29) ;
\draw [shift={(27.74,129.47)}, rotate = 244.23] [color={rgb, 255:red, 0; green, 0; blue, 0 }  ][line width=0.75]    (10.93,-3.29) .. controls (6.95,-1.4) and (3.31,-0.3) .. (0,0) .. controls (3.31,0.3) and (6.95,1.4) .. (10.93,3.29)   ;
\draw    (10.66,215.28) -- (40.21,155.29) ;
\draw [shift={(28.09,179.9)}, rotate = 116.22] [color={rgb, 255:red, 0; green, 0; blue, 0 }  ][line width=0.75]    (10.93,-3.29) .. controls (6.95,-1.4) and (3.31,-0.3) .. (0,0) .. controls (3.31,0.3) and (6.95,1.4) .. (10.93,3.29)   ;
\draw    (40.21,155.29) -- (85.14,156) ;
\draw [shift={(68.67,155.74)}, rotate = 180.91] [color={rgb, 255:red, 0; green, 0; blue, 0 }  ][line width=0.75]    (10.93,-3.29) .. controls (6.95,-1.4) and (3.31,-0.3) .. (0,0) .. controls (3.31,0.3) and (6.95,1.4) .. (10.93,3.29)   ;
\draw    (85.14,156) .. controls (99.38,126.99) and (131.14,120.08) .. (152.14,156) ;
\draw [shift={(124.71,132.38)}, rotate = 187.44] [color={rgb, 255:red, 0; green, 0; blue, 0 }  ][line width=0.75]    (10.93,-3.29) .. controls (6.95,-1.4) and (3.31,-0.3) .. (0,0) .. controls (3.31,0.3) and (6.95,1.4) .. (10.93,3.29)   ;
\draw    (85.14,156) .. controls (102.66,193.18) and (140.17,180.34) .. (152.14,156) ;
\draw [shift={(124.23,178.45)}, rotate = 172.02] [color={rgb, 255:red, 0; green, 0; blue, 0 }  ][line width=0.75]    (10.93,-3.29) .. controls (6.95,-1.4) and (3.31,-0.3) .. (0,0) .. controls (3.31,0.3) and (6.95,1.4) .. (10.93,3.29)   ;
\draw    (197.07,156.71) .. controls (209.07,180.71) and (223.27,178.54) .. (231.13,179.53) ;
\draw [shift={(215.65,176.75)}, rotate = 209.54] [color={rgb, 255:red, 0; green, 0; blue, 0 }  ][line width=0.75]    (10.93,-3.29) .. controls (6.95,-1.4) and (3.31,-0.3) .. (0,0) .. controls (3.31,0.3) and (6.95,1.4) .. (10.93,3.29)   ;
\draw  [dash pattern={on 0.84pt off 2.51pt}]  (233.41,155.83) -- (239.25,155.59) -- (245.29,155.34) ;
\draw    (245.33,130.78) .. controls (256.01,130.78) and (278.18,139.67) .. (282.29,155.48) ;
\draw [shift={(272.6,141.51)}, rotate = 215.16] [color={rgb, 255:red, 0; green, 0; blue, 0 }  ][line width=0.75]    (10.93,-3.29) .. controls (6.95,-1.4) and (3.31,-0.3) .. (0,0) .. controls (3.31,0.3) and (6.95,1.4) .. (10.93,3.29)   ;
\draw    (249.44,179.19) .. controls (245.33,179.19) and (271.61,181.16) .. (282.29,155.48) ;
\draw [shift={(273.69,168.5)}, rotate = 139.82] [color={rgb, 255:red, 0; green, 0; blue, 0 }  ][line width=0.75]    (10.93,-3.29) .. controls (6.95,-1.4) and (3.31,-0.3) .. (0,0) .. controls (3.31,0.3) and (6.95,1.4) .. (10.93,3.29)   ;
\draw    (360.14,156) .. controls (367.14,128) and (412.14,131) .. (399.14,131) ;
\draw [shift={(383.27,134.15)}, rotate = 158.3] [color={rgb, 255:red, 0; green, 0; blue, 0 }  ][line width=0.75]    (10.93,-3.29) .. controls (6.95,-1.4) and (3.31,-0.3) .. (0,0) .. controls (3.31,0.3) and (6.95,1.4) .. (10.93,3.29)   ;
\draw    (360.14,156) .. controls (369.56,183.76) and (395.28,179.01) .. (403.14,180) ;
\draw [shift={(382.48,178.11)}, rotate = 199.9] [color={rgb, 255:red, 0; green, 0; blue, 0 }  ][line width=0.75]    (10.93,-3.29) .. controls (6.95,-1.4) and (3.31,-0.3) .. (0,0) .. controls (3.31,0.3) and (6.95,1.4) .. (10.93,3.29)   ;
\draw  [dash pattern={on 0.84pt off 2.51pt}]  (405.41,156.85) -- (411.25,156.6) -- (417.29,156.35) ;
\draw    (422.04,130.62) .. controls (432.71,130.62) and (454.89,139.51) .. (458.99,155.31) ;
\draw [shift={(449.31,141.34)}, rotate = 215.16] [color={rgb, 255:red, 0; green, 0; blue, 0 }  ][line width=0.75]    (10.93,-3.29) .. controls (6.95,-1.4) and (3.31,-0.3) .. (0,0) .. controls (3.31,0.3) and (6.95,1.4) .. (10.93,3.29)   ;
\draw    (426.14,179.02) .. controls (422.04,179.02) and (448.32,181) .. (458.99,155.31) ;
\draw [shift={(450.4,168.34)}, rotate = 139.82] [color={rgb, 255:red, 0; green, 0; blue, 0 }  ][line width=0.75]    (10.93,-3.29) .. controls (6.95,-1.4) and (3.31,-0.3) .. (0,0) .. controls (3.31,0.3) and (6.95,1.4) .. (10.93,3.29)   ;
\draw    (152.14,156) -- (197.07,156.71) ;
\draw [shift={(180.61,156.45)}, rotate = 180.91] [color={rgb, 255:red, 0; green, 0; blue, 0 }  ][line width=0.75]    (10.93,-3.29) .. controls (6.95,-1.4) and (3.31,-0.3) .. (0,0) .. controls (3.31,0.3) and (6.95,1.4) .. (10.93,3.29)   ;
\draw    (197.07,156.71) .. controls (199.92,131.12) and (232.77,131.12) .. (228.66,132.11) ;
\draw [shift={(213.53,135.25)}, rotate = 151.26] [color={rgb, 255:red, 0; green, 0; blue, 0 }  ][line width=0.75]    (10.93,-3.29) .. controls (6.95,-1.4) and (3.31,-0.3) .. (0,0) .. controls (3.31,0.3) and (6.95,1.4) .. (10.93,3.29)   ;
\draw    (282.29,155.48) -- (321.21,155.74) ;
\draw [shift={(307.75,155.65)}, rotate = 180.39] [color={rgb, 255:red, 0; green, 0; blue, 0 }  ][line width=0.75]    (10.93,-3.29) .. controls (6.95,-1.4) and (3.31,-0.3) .. (0,0) .. controls (3.31,0.3) and (6.95,1.4) .. (10.93,3.29)   ;
\draw    (321.21,155.74) -- (360.14,156) ;
\draw [shift={(346.68,155.91)}, rotate = 180.39] [color={rgb, 255:red, 0; green, 0; blue, 0 }  ][line width=0.75]    (10.93,-3.29) .. controls (6.95,-1.4) and (3.31,-0.3) .. (0,0) .. controls (3.31,0.3) and (6.95,1.4) .. (10.93,3.29)   ;
\draw    (458.99,155.31) -- (503.93,156.02) ;
\draw [shift={(487.46,155.76)}, rotate = 180.91] [color={rgb, 255:red, 0; green, 0; blue, 0 }  ][line width=0.75]    (10.93,-3.29) .. controls (6.95,-1.4) and (3.31,-0.3) .. (0,0) .. controls (3.31,0.3) and (6.95,1.4) .. (10.93,3.29)   ;
\draw    (503.93,156.02) .. controls (518.16,127.02) and (549.93,120.1) .. (570.93,156.02) ;
\draw [shift={(543.5,132.4)}, rotate = 187.44] [color={rgb, 255:red, 0; green, 0; blue, 0 }  ][line width=0.75]    (10.93,-3.29) .. controls (6.95,-1.4) and (3.31,-0.3) .. (0,0) .. controls (3.31,0.3) and (6.95,1.4) .. (10.93,3.29)   ;
\draw    (503.93,156.02) .. controls (521.45,193.21) and (558.96,180.36) .. (570.93,156.02) ;
\draw [shift={(543.02,178.48)}, rotate = 172.02] [color={rgb, 255:red, 0; green, 0; blue, 0 }  ][line width=0.75]    (10.93,-3.29) .. controls (6.95,-1.4) and (3.31,-0.3) .. (0,0) .. controls (3.31,0.3) and (6.95,1.4) .. (10.93,3.29)   ;
\draw    (570.93,156.02) -- (615.86,156.73) ;
\draw [shift={(599.39,156.47)}, rotate = 180.91] [color={rgb, 255:red, 0; green, 0; blue, 0 }  ][line width=0.75]    (10.93,-3.29) .. controls (6.95,-1.4) and (3.31,-0.3) .. (0,0) .. controls (3.31,0.3) and (6.95,1.4) .. (10.93,3.29)   ;
\draw    (615.86,156.73) -- (645.4,96.74) ;
\draw [shift={(633.28,121.36)}, rotate = 116.22] [color={rgb, 255:red, 0; green, 0; blue, 0 }  ][line width=0.75]    (10.93,-3.29) .. controls (6.95,-1.4) and (3.31,-0.3) .. (0,0) .. controls (3.31,0.3) and (6.95,1.4) .. (10.93,3.29)   ;
\draw    (615.86,156.73) -- (646,219.17) ;
\draw [shift={(633.54,193.36)}, rotate = 244.23] [color={rgb, 255:red, 0; green, 0; blue, 0 }  ][line width=0.75]    (10.93,-3.29) .. controls (6.95,-1.4) and (3.31,-0.3) .. (0,0) .. controls (3.31,0.3) and (6.95,1.4) .. (10.93,3.29)   ;
\draw    (321.21,155.74) -- (322,233) ;
\draw [shift={(321.67,200.37)}, rotate = 269.42] [color={rgb, 255:red, 0; green, 0; blue, 0 }  ][line width=0.75]    (10.93,-3.29) .. controls (6.95,-1.4) and (3.31,-0.3) .. (0,0) .. controls (3.31,0.3) and (6.95,1.4) .. (10.93,3.29);
\draw (37.14,158.4) node [anchor=north west][inner sep=0.75pt]  [font=\scriptsize]  {$a_{0}$};
\draw (76.14,156.4) node [anchor=north west][inner sep=0.75pt]  [font=\scriptsize]  {$b_{0}$};
\draw (149.14,158.4) node [anchor=north west][inner sep=0.75pt]  [font=\scriptsize]  {$a_{1}$};
\draw (189.32,156.4) node [anchor=north west][inner sep=0.75pt]  [font=\scriptsize]  {$b_{1}$};
\draw (315.14,138.4) node [anchor=north west][inner sep=0.75pt]    {$0$};
\draw (278.29,158.88) node [anchor=north west][inner sep=0.75pt]  [font=\scriptsize]  {$a_{m}$};
\draw (348.14,156.4) node [anchor=north west][inner sep=0.75pt]  [font=\scriptsize]  {$b_{m}$};
\draw (455.03,161.42) node [anchor=north west][inner sep=0.75pt]  [font=\scriptsize]  {$a_{n-1}$};
\draw (496.03,157.4) node [anchor=north west][inner sep=0.75pt]  [font=\scriptsize]  {$b_{n-1}$};
\draw (568.03,161.42) node [anchor=north west][inner sep=0.75pt]  [font=\scriptsize]  {$a_{n}$};
\draw (608.45,157.4) node [anchor=north west][inner sep=0.75pt]  [font=\scriptsize]  {$b_{n}$};
\draw (17,82.23) node [anchor=north west][inner sep=0.75pt]    {$\Gamma _{S,l}$};
\draw (17,201.57) node [anchor=north west][inner sep=0.75pt]    {$\Gamma _{S,l}^{*}$};
\draw (103,107.23) node [anchor=north west][inner sep=0.75pt]    {$\Gamma _{S,0}$};
\draw (103,185.23) node [anchor=north west][inner sep=0.75pt]    {$\Gamma _{S,0}^{*}$};
\draw (624,71.4) node [anchor=north west][inner sep=0.75pt]    {$\Gamma _{S,r}$};
\draw (620,205.4) node [anchor=north west][inner sep=0.75pt]    {$\Gamma _{S,r}^{*}$};
\draw (198,110.23) node [anchor=north west][inner sep=0.75pt]    {$\Gamma _{S,1}$};
\draw (197,183.23) node [anchor=north west][inner sep=0.75pt]    {$\Gamma _{S,1}^{*}$};
\draw (247,108.23) node [anchor=north west][inner sep=0.75pt]    {$\Gamma _{S,m-1}$};
\draw (251.44,182.59) node [anchor=north west][inner sep=0.75pt]    {$\Gamma _{S,m-1}^{*}$};
\draw (374,109.23) node [anchor=north west][inner sep=0.75pt]    {$\Gamma _{S,m}$};
\draw (376.44,181.59) node [anchor=north west][inner sep=0.75pt]    {$\Gamma _{S,m}^{*}$};
\draw (424,108.23) node [anchor=north west][inner sep=0.75pt]    {$\Gamma _{S,n-2}$};
\draw (428.14,182.42) node [anchor=north west][inner sep=0.75pt]    {$\Gamma _{S,n-2}^{*}$};
\draw (518,106.23) node [anchor=north west][inner sep=0.75pt]    {$\Gamma _{S,n-1}$};
\draw (525,185.23) node [anchor=north west][inner sep=0.75pt]    {$\Gamma _{S,n-1}^{*}$};
\draw (327,218.4) node [anchor=north west][inner sep=0.75pt]    {$\Gamma _{S,d}$};
\end{tikzpicture}
\caption{The jump contour $\Gamma_S$ of the RH problem for $S$.}
\label{fig:RHP S}
\end{figure}
In view of the RH problem for $T$ and the decompositions \eqref{equ: decomposition of J_T-1} and \eqref{equ: decomposition of J_T-2}, it is straightforward to verify that $S$ satisfies the following RH problem.
\paragraph{RH problem for $S$}
\begin{itemize}
\item [\rm (a)] $S(z)$ is holomorphic for $z\in\mathbb{C}\backslash\Gamma_{S}$, where 
\begin{align}\label{def:Gamma_s}
\Gamma_{S}:=\left(\cup_{j\in\{r,l\}}\left(\Gamma_{S,j}\cup\Gamma_{S,j}^*\right)\right)
\bigcup\left(\cup_{j=0}^{n-1}\left(\Gamma_{S,j}\cup\Gamma_{S,j}^*\right)\right)\bigcup
\left(\cup_{j=0}^{n}\overline{(a_j,b_j)}\right)\bigcup\Gamma_{S,d};
\end{align}
see Figure \ref{fig:RHP S} for an illustration. 
\item [\rm (b)] For $z\in\Gamma_{S}$, we have 
  \begin{equation}
    S_{+}(z)=S_{-}(z)J_{S}(z),
  \end{equation}
  where
\begin{equation}\label{equ:jump of S}
      J_S(z)=\left\{\begin{array}{ll}
  {\begin{pmatrix}
    1 & 0 \\
  e^{-(\alpha-\beta)\pi i+2isg(z)} & 1
  \end{pmatrix}},& {z \in \Gamma_{S,r}}, \\
  {\begin{pmatrix}
    1 & -e^{(\alpha-\beta)\pi i-2isg(z)} \\
    0 & 1
  \end{pmatrix}},& {z \in \Gamma_{S,r}^*} ,\\
  {\begin{pmatrix}
    1 & 0 \\
    e^{(\alpha-\beta)\pi i+2isg(z)} & 1 
  \end{pmatrix}}, & { z \in \Gamma_{S,l}}, \\
  {\begin{pmatrix}
    1 & -e^{-(\alpha-\beta)\pi i-2isg(z)} \\
    0 & 1
  \end{pmatrix}},& { z \in \Gamma_{S,l}^*}, \\
  
  {\begin{pmatrix}
      0 & -e^{-(\alpha-\beta)\pi i-is\Omega_j } \\
  e^{(\alpha-\beta)\pi i+is\Omega_j } & 0
  \end{pmatrix}},& z \in \cup_{j=0}^{m-1}(a_j,b_j)\cup(a_m,0), \\
  {\begin{pmatrix}
      0 & -e^{(\alpha-\beta)\pi i-is\Omega_j } \\
  e^{-(\alpha-\beta)\pi i+is\Omega_j } & 0
  \end{pmatrix}},& z \in \cup_{j=m+1}^{n}(a_j,b_j)\cup(0,b_m),\\
  \begin{pmatrix}
       1 & 0 \\
    e^{(\alpha-\beta)\pi i+2isg(z) } & 1
  \end{pmatrix},& z \in \cup_{j=0}^{m-1}\Gamma_{S,j}, \\  
   \begin{pmatrix}
       1 & 0 \\
    e^{-(\alpha-\beta)\pi i+2isg(z) } & 1
   \end{pmatrix},& z \in \cup_{j=m}^{n-1}\Gamma_{S,j}, \\  
\begin{pmatrix}
    1 & -e^{-(\alpha-\beta)\pi i-2isg(z)} \\
    0  & 1
\end{pmatrix},& z \in \cup_{j=0}^{m-1}\Gamma_{S,j}^*, \\
\begin{pmatrix}
    1 & -e^{(\alpha-\beta)\pi i-2isg(z)} \\
    0  & 1
\end{pmatrix},& z \in \cup_{j=m}^{n-1}\Gamma_{S,j}^*,\\
{e^{2 \beta \pi i \sigma_3}},& z \in \Gamma_{S,d}.
\end{array}\right.
\end{equation}

\item[\rm (c)] As $z\rightarrow \infty$ with $\arg z\in(-\frac{\pi}{2},\frac{3\pi}{2})$, we have
  \begin{equation}\label{equ:S asym at inf}
      S(z)=\left(I+\frac{S_1(s)}{z}+\mathcal{O}\left(z^{-2}\right)\right) z^{-\beta \sigma_3},
  \end{equation}
  where $S_1$ is independent of $s$. 

\item[\rm (d)] As $z \to p$ with $p\in\mathcal{I}_{e}$, we have
\begin{align}\label{equ:S asym at endpoints}
S(z)=\mathcal{O}(\log(z-p)). 
\end{align}

\item[\rm (e)] As $z \rightarrow 0$ from $\mathbb{C}^+$ , we have
  \begin{equation}\label{equ:S asym at 0}
      S(z)=S_0(z)z^{\alpha\sigma_3},
  \end{equation}
   where $S_0(z)$ is holomorphic in a neighborhood of $z=0$. 
\end{itemize}

\section{Global parametrix}\label{sec: global parametrix}
In view of items (c) and (d) of Proposition \ref{prop:g-func}, it follows that there exists some constant $c_0>0$ such that 
\begin{equation}\label{equ:JS-I}
J_{S}(z)=I+\mathcal{O}(e^{-c_0s|z|}), \qquad s\to+\infty,
\end{equation}
uniformly for $z$ bounded away from fixed neighborhoods of  $a_{j}$ and $b_{j}$, $j=0,\dots, n$. This leads us to  consider the following global parametrix.

\paragraph{RH problem for $P^{(\infty)}$}
\begin{itemize}
\item [\rm (a)] $P^{(\infty)}$ is holomorphic for $z\in\mathbb{C}\setminus(\overline{\Sigma}\cup(-i\infty,0])$.
\item[\rm (b)] For $z\in\Sigma \cup(-i\infty,0)$, we have
  \begin{equation}
    P^{(\infty)}_{+}(z)= P^{(\infty)}_{-}(z)J_{ P^{(\infty)}}(z),
  \end{equation}
  where
\begin{equation}\label{equ:jump of Pinfty}
      J_{P^{(\infty)}}(z)=\left\{\begin{array}{ll}
  {e^{-\frac{(\alpha-\beta)\pi i}{2}\sigma_3}
\begin{pmatrix}
0 & -1\\
1 & 0
\end{pmatrix}
e^{\frac{(\alpha-\beta)\pi i}{2}\sigma_3}},& {z\in(a_0,b_0)}, \\
  {e^{-\frac{(\alpha-\beta)\pi i}{2}\sigma_3}
\begin{pmatrix}
    0 & -e^{-is\Omega_j}\\
    e^{is\Omega_j} & 0
\end{pmatrix}e^{\frac{(\alpha-\beta)\pi i}{2}\sigma_3}},& {z\in\cup_{j=1}^{m-1}(a_j,b_j)\cup(a_m,0)} ,\\
  {e^{\frac{(\alpha-\beta)\pi i}{2}\sigma_3}
\begin{pmatrix}
0 & -e^{-is\Omega_j}\\
e^{is\Omega_j} & 0
\end{pmatrix} e^{-\frac{(\alpha-\beta)\pi i}{2}\sigma_3}}, & { z\in\cup_{j=m+1}^{n}(a_j,b_j)\cup(0,b_m)}, \\
{e^{2 \beta \pi i \sigma_3}},& z \in (-i\infty,0).
\end{array}\right.
\end{equation}

\item [\rm (c)] As $z\to\infty$, we have 
\begin{align}\label{equ:asy Pinfty at z=infty}
P^{(\infty)}(z)=\left(I+\frac{P^{(\infty)}_{1}(s)}{z}+\mathcal{O}(z^{-2})\right)z^{-\beta\sigma_3},
\end{align}
where $P^{(\infty)}_{1}$ is independent of $z$.
\item [\rm (d)] As $z\to p$ with $p\in \mathcal{I}_e$, we have
\begin{align}\label{equ:asy Pinfty at z=endpoints}
P^{(\infty)}(z)=\mathcal{O}\left(\left(z-p\right)^{-\frac{1}{4}}\right).
\end{align}
\item [\rm (e)] As $z\to 0$ from $\mathbb{C}^+$, we have 
\begin{align}\label{equ:asy Pinfty at z=0}
P^{(\infty)}(z)=
P^{(\infty)}_{0}(z)z^{\alpha\sigma_3}
\end{align}
where $P^{(\infty)}_{0}(z)$ is holomorphic in a neighborhood of the origin. 
\end{itemize}

In what follows, we will give an explicit solution of the global parametrix.


\paragraph{The Szeg\H{o} function $\mathcal{D}$}  
Let 
\begin{align}\label{def:D-function}
\mathcal{D}(z):=\exp\left\{\frac{\sqrt{\mathcal{R}(z)}}{2\pi i}\int_{\Sigma}\frac{\mathcal{H}(\xi)}{\xi-z}\dif\xi\right\}
\end{align}
be a  Szeg\H{o} type function, where  
\begin{equation}\label{def:mathcalH}
\mathcal{H}(z):=
\frac{\log z^{-2\beta}+\sign{z}(\alpha-\beta)\pi i+\zeta_j}{\sqrt{\mathcal{R}(z)}_{+}}, \quad z\in(a_j,b_j), \quad j=0,1,\dots,n.
\end{equation}
Here, $\sqrt{\mathcal{R}(z)}$ is given by \eqref{equ:sqrt(R(z))}, $\zeta_0:=0$ and $\zeta_{j}$, $j=1,\dots,n$, is determined by \eqref{equ: sol of zetaj} below to ensure $\mathcal{D}(z)$ is bounded as $z\to\infty$. We will use $\mathcal{D}$ to reduce the RH problem for $P^{(\infty)}$ to a solvable one. Some properties of $\mathcal{D}$ are listed in the following proposition.
\begin{proposition}\label{prop:D-function} 
The $\mathcal{D}$ function defined in \eqref{def:D-function} satisfies the following properties.
\begin{itemize}
    \item [\rm (a)] $\mathcal{D}(z)$ is holomorphic for $z\in\mathbb{C}\setminus\overline{\Sigma}$.
    \item [\rm (b)] For $z\in(a_j,b_j)$, $j=0,1,\dots,n$, we have 
    \begin{equation}\label{eq:Djump}
        \mathcal{D}_{+}(z)\mathcal{D}_{-}(z)=z^{-2\beta}e^{\sign{z}(\alpha-\beta)\pi i+\zeta_j}.
    \end{equation} 
    \item [\rm (c)] As $z\to\infty$, we have 
    \begin{align}\label{equ: expansion of mathcalD at z=infty}
    \mathcal{D}(z)=\mathcal{D}_{\infty}\left(1+\frac{\mathcal{D}_{\infty,1}}{z}+\mathcal{O}(z^{-2})\right),
    \end{align}
    where 
    \begin{align}
    &\mathcal{D}_{\infty}:=\exp\left\{-\frac{1}{2\pi i}\int_{\Sigma}\xi^{n}\mathcal{H}(\xi)\dif\xi\right\},\label{def:Dinfty}\\
    &\mathcal{D}_{\infty,1}:=\frac{\sum_{k=0}^{n}(a_k+b_k)}{4\pi i}\int_{\Sigma}\xi^n\mathcal{H}(\xi)\dif\xi-\frac{1}{2\pi i}\int_{\Sigma}\xi^{n+1}\mathcal{H}(\xi)\dif\xi \label{def:Dinfty1-pre}\\
    &=\frac{1}{4\pi i}\left(\frac{\sum_{k=0}^{n}(a_k+b_k)}{2}\sum_{j=1}^{n}\zeta_ja_{j,n}-\sum_{j=1}^{n}\zeta_ja_{j,n+1}\right)\nonumber\\
    &\hspace{1em}+\frac{\beta}{\pi i}\int_{\Sigma}\frac{\xi^{n+1}\log|\xi|}{\sqrt{\mathcal{R}(\xi)}_{+}}\dif\xi-\frac{\beta\sum_{k=0}^{n}(a_k+b_k)}{2\pi i}\int_{\Sigma}\frac{\xi^{n}\log|\xi|}{\sqrt{\mathcal{R}(\xi)}_{+}}\dif\xi\nonumber\\
    &\hspace{1em}+\frac{\alpha}{2}\left(\sum_{j=0}^{m-1}\int_{a_j}^{b_j}+\int_{a_m}^{0}-\int_{0}^{b_m}-\sum_{j=m+1}^{n}\int_{a_j}^{b_j}\right)\frac{\xi^{n+1}}{\sqrt{\mathcal{R}(\xi)}_{+}}\dif\xi\nonumber\\
    &\hspace{1em}-\frac{\alpha\sum_{k=0}^{n}(a_k+b_k)}{4}\left(\sum_{j=0}^{m-1}\int_{a_j}^{b_j}+\int_{a_m}^{0}-\int_{0}^{b_m}-\sum_{j=m+1}^{n}\int_{a_j}^{b_j}\right)\frac{\xi^{n}}{\sqrt{\mathcal{R}(\xi)}_{+}}\dif\xi
    \in i \mathbb{R}.\label{def:Dinfty1-a}  
    \end{align}
   Here, $\zeta_1,\dots,\zeta_{n}$, are
    all purely imaginary and  uniquely determined by 
    \begin{align}\label{equ: sol of zetaj}
    \begin{pmatrix}
        \zeta_1\\
        \zeta_2\\ 
        \vdots\\ 
        \zeta_{n}
    \end{pmatrix}
    =2{\tilde{\mathbb{A}}}^{-\rm T}
    \begin{pmatrix}
        \mathcal{F}_0\\
        \mathcal{F}_1\\
        \vdots\\
        \mathcal{F}_{n-1}
    \end{pmatrix},
    \end{align}
    where $\tilde{\mathbb{A}}$ is given by \eqref{def:tilde matrix A} and 
    \begin{align}\label{def:mathcalFl}
    \mathcal{F}_{l}=2\beta\int_{\Sigma}\frac{\xi^{l}\log|\xi|}{\sqrt{\mathcal{R}(\xi)}_{+}}\dif\xi-\pi i\alpha\int_{\Sigma}\xi^{l}\frac{\sign{\xi}}{\sqrt{\mathcal{R}(\xi)}_{+}}\dif\xi, \quad l=0,1,\dots,n-1.
    \end{align}
    \item [\rm (d)] As $z\to p$ from $\mathbb{C}^+$ with $p\in\mathcal{I}_{e}$, we have 
    \begin{align}\label{equ: asy mathcalD at endpoints}
    \mathcal{D}(z)=\mathcal{D}_{p}\left(1+\mathcal{O}({(z-p)}^{\frac{1}{2}})\right),
    \end{align}
    where
    \begin{align}\label{equ:mathcalD_p}
    \mathcal{D}_{p}=\exp\left\{\frac{\log p^{-2\beta}+\sign{p}(\alpha-\beta)\pi i+\zeta_j}{2}\right\}.
    \end{align}
    \item [\rm (e)] As $z\to 0$ from $\mathbb{C}^+$, we have
    \begin{align}\label{equ: asy mathcalD at 0}
     \mathcal{D}(z)=\mathcal{O}(1)\cdot z^{-\alpha-\beta}.
    \end{align}
\end{itemize}
\end{proposition}

\begin{proof}
Since items (a) and (b) follow directly from \eqref{def:D-function} and the Sokhotski-Plemlj formula, we focus on the proofs of items (c)--(e). 

\begin{itemize}
\item[\rm (c)] As $z\rightarrow\infty$, it is readily seen that
\begin{align}\label{equ: asy of integral at z=infty}
\int_{\Sigma}\frac{\mathcal{H}(\xi)}{\xi-z}\dif\xi
&=-\frac{1}{z}\int_{\Sigma}\mathcal{H}(\xi)\dif\xi
-\frac{1}{z^2}\int_{\Sigma}\xi\mathcal{H}(\xi)\dif\xi-
\frac{1}{z^3}\int_{\Sigma}\xi^2\mathcal{H}(\xi)\dif\xi
\nonumber 
\\
&\quad -\dots-\frac{1}{z^n}\int_{\Sigma}\xi^{n-1}\mathcal{H}(\xi)\dif\xi
-\frac{1}{z^{n+1}}\int_{\Sigma}\xi^n\mathcal{H}(\xi)\dif\xi+\mathcal{O}\left(\frac{1}{z^{n+2}}\right),
\end{align}
and 
\begin{align}\label{equ: asy of sqrtmathcalR+ at z=infty}
\sqrt{\mathcal{R}(z)}=z^{n+1}-\frac{1}{2}\sum_{j=0}^{n}(a_j+b_j)z^{n}+\mathcal{O}(z^n). 
\end{align}
To ensure $\mathcal{D}(z)$ is bounded as $z\to\infty$, one has 
\begin{align}
\int_{\Sigma}\xi^{l}\mathcal{H}(\xi)\dif\xi=0,  \qquad l=0,1,\dots,n-1.
\end{align}
Inserting \eqref{def:mathcalH} into the above formula, it follows that
\begin{align}\label{equ:sumzetajint=mathFl}
\sum_{j=1}^{n}\zeta_{j}\int_{a_j}^{b_j}\frac{\xi^{l}}{\sqrt{\mathcal{R}(\xi)}_{+}}\dif\xi=\mathcal{F}_l, \qquad l=0,\dots,n-1, \ \ j=1,\dots,n,
\end{align}
where 
\begin{align}
\mathcal{F}_{l}=-\int_{\Sigma}\xi^{l}\frac{\log\xi^{-2\beta}}{\sqrt{\mathcal{R}(\xi)}_{+}}\dif\xi-\pi i(\alpha-\beta)\int_{\Sigma}\xi^l\frac{\sign{\xi}}{\sqrt{\mathcal{R}(\xi)}_{+}}\dif\xi.
\end{align}
A straightforward calculation shows that
\begin{align}\label{equ:mathcalFl cal}
\mathcal{F}_{l}=2\beta\int_{\Sigma}\frac{\xi^{l}\log|\xi|}{\sqrt{\mathcal{R}(\xi)}_{+}}\dif\xi-\pi i\alpha\int_{\Sigma}\xi^{l}\frac{\sign{\xi}}{\sqrt{\mathcal{R}(\xi)}_{+}}\dif\xi+
\pi i\beta\int_{\Sigma}\frac{\xi^{l}}{\sqrt{\mathcal{R}(\xi)}_{+}}\dif\xi.
\end{align}
Note that $\sqrt{\mathcal{R}(\xi)}_{+}$ is purely imaginary for 
$\xi\in\Sigma$, $\alpha\in\mathbb{R}$ and $\beta\in i\mathbb{R}$,
the first two terms on the right-hand side of \eqref{equ:mathcalFl cal} are all real. To evaluate the third term, we see from the residue theorem that
\begin{align}\label{eq:3rdint}
\pi i\beta\int_{\Sigma}\frac{\xi^{l}}{\sqrt{\mathcal{R}(\xi)}_{+}}\dif\xi
=-2\pi i\res_{z=\infty} \left[\frac{z^{l}}{\sqrt{\mathcal{R}(z)}}\right]\cdot\pi i\beta.
\end{align}
As it is easily seen from \eqref{equ: asy of sqrtmathcalR+ at z=infty}  that
\begin{align}\label{equ: 1/sqrtR asy at z=infty}
\frac{1}{\sqrt{\mathcal{R}(z)}}=z^{-n-1}+\frac{\sum_{k=0}^{n}(a_k+b_k)}{2}z^{-n-2}+o(z^{-n-2}), \qquad \ z\to\infty,
\end{align}
it is immediate that $\res_{z=\infty} [{z^{l}}/\sqrt{\mathcal{R}(z)}]=0$ 
for $l=0,1,\dots,n-1$. This, together with \eqref{equ:mathcalFl cal} and \eqref{eq:3rdint}, implies \eqref{def:mathcalFl} and $\mathcal{F}_{l}\in \mathbb{R}$.

Next, observe that 
\begin{align*}
    \int_{a_j}^{b_j}\frac{\xi^{l}}{\sqrt{\mathcal{R}(\xi)}_{+}}\dif\xi=\frac{1}{2}\oint_{A_j}\frac{\xi^{l}}{\sqrt{\mathcal{R}(\xi)}}\dif\xi, \quad l,j=0,\dots,n-1,
\end{align*}
one can rewrite \eqref{equ:sumzetajint=mathFl} as
\begin{equation}\label{equ:system of zetaj and mathcalFl}
    \sum_{j=1}^{n}\frac{\zeta_j}{2}a_{j,l}=\mathcal{F}_l,
\end{equation}
where $a_{j,l}$ is given by \eqref{equ:a_{k,j}}. Solving the linear system 
\eqref{equ:system of zetaj and mathcalFl} gives us \eqref{equ: sol of zetaj}. 
As $\mathcal{F}_{l}\in\mathbb{R}$, we conclude $\zeta_j\in i\mathbb{R}$ by noticing that the elements of ${\tilde{\mathbb{A}}}^{-\rm T}$ are purely imaginary. Using \eqref{equ: asy of integral at z=infty} and \eqref{equ: asy of sqrtmathcalR+ at z=infty}, we obtain \eqref{equ: expansion of mathcalD at z=infty}--\eqref{def:Dinfty1-pre}. 

We finally show \eqref{def:Dinfty1-a} and prove $\mathcal{D}_{\infty,1}$ is purely imaginary. From \eqref{def:mathcalH}, \eqref{equ:a_{k,j}} and
\eqref{def:Dinfty1-pre}, it follows that  
\begin{align*}
&\mathcal{D}_{\infty,1}=
\frac{1}{4\pi i}\left(\frac{\sum_{k=0}^{n}(a_k+b_k)}{2}\sum_{j=1}^{n}\zeta_ja_{j,n}-\sum_{j=1}^{n}\zeta_ja_{j,n+1}\right)\nonumber\\
&\hspace{3.5em}+\frac{\beta}{\pi i}\left(\sum_{j=0}^{n}\int_{a_j}^{b_j}\frac{\xi^{n+1}\log\xi}{\sqrt{\mathcal{R}(\xi)}_{+}}\dif\xi-\frac{\sum_{k=0}^{n}(a_k+b_k)}{2}\sum_{j=0}^{n}\int_{a_j}^{b_j}\frac{\xi^n\log\xi}{\sqrt{\mathcal{R}(\xi)}_{+}}\dif\xi\right)\nonumber\\
&\hspace{3.5em}+\frac{\alpha-\beta}{2}\left(\frac{\sum_{k=0}^{n}(a_k+b_k)}{2}\sum_{j=0}^{n}\int_{a_j}^{b_j}\frac{\xi^n\sign{\xi}}{\sqrt{\mathcal{R}(\xi)}_{+}}\dif\xi-\sum_{j=0}^{n}\int_{a_j}^{b_j}\frac{\xi^{n+1}\sign{\xi}}{\sqrt{\mathcal{R}(\xi)}_{+}}\dif\xi\right), 
\end{align*}
or equivalently,  
\begin{align}\label{equ: mathcalDinfty,1 cal}
&\mathcal{D}_{\infty,1}:=
\frac{1}{4\pi i}\left(\frac{\sum_{k=0}^{n}(a_k+b_k)}{2}\sum_{j=1}^{n}\zeta_ja_{j,n}-\sum_{j=1}^{n}\zeta_ja_{j,n+1}\right)\nonumber\\
&\hspace{3.5em}+\frac{\beta}{\pi i}\int_{\Sigma}\frac{\xi^{n+1}\log|\xi|}{\sqrt{\mathcal{R}(\xi)}_{+}}\dif\xi-\frac{\beta\sum_{k=0}^{n}(a_k+b_k)}{2\pi i}\int_{\Sigma}\frac{\xi^{n}\log|\xi|}{\sqrt{\mathcal{R}(\xi)}_{+}}\dif\xi\nonumber\\
&\hspace{3.5em}+\frac{\alpha}{2}\left(\sum_{j=0}^{m-1}\int_{a_j}^{b_j}+\int_{a_m}^{0}-\int_{0}^{b_m}-\sum_{j=m+1}^{n}\int_{a_j}^{b_j}\right)\frac{\xi^{n+1}}{\sqrt{\mathcal{R}(\xi)}_{+}}\dif\xi\nonumber\\
&\hspace{3.5em}-\frac{\alpha\sum_{k=0}^{n}(a_k+b_k)}{4}\left(\sum_{j=0}^{m-1}\int_{a_j}^{b_j}+\int_{a_m}^{0}-\int_{0}^{b_m}-\sum_{j=m+1}^{n}\int_{a_j}^{b_j}\right)\frac{\xi^{n}}{\sqrt{\mathcal{R}(\xi)}_{+}}\dif\xi\nonumber\\
&\hspace{3.5em}+\uwave{
\frac{\beta}{2}\int_{\Sigma}\frac{\xi^{n+1}}{\sqrt{\mathcal{R}(\xi)}_{+}}\dif\xi
-\frac{\beta}{4}\sum_{k=0}^{n}(a_k+b_k)\int_{\Sigma}\frac{\xi^{n}}{\sqrt{\mathcal{R}(\xi)}_{+}}\dif\xi}.
\end{align}
Since $\zeta_j,a_{k,l},\beta$ and $\sqrt{\mathcal{R}(\xi)}_{+}$, $\xi\in\Sigma$, are all purely imaginary and $\alpha\in\mathbb{R}$, it follows that 
the first four terms on the right-hand side of \eqref{equ: mathcalDinfty,1 cal} are all purely imaginary. Now we claim that the underlined part in \eqref{equ: mathcalDinfty,1 cal} vanishes, which leads to  \eqref{def:Dinfty1-a}. Indeed, we see from the residue theorem that 
\begin{multline}\label{equ: an equ ho to 0}
\frac{\beta}{2}\int_{\Sigma}\frac{\xi^{n+1}}{\sqrt{\mathcal{R}(\xi)}_{+}}\dif\xi
-\frac{\beta}{4}\sum_{k=0}^{n}(a_k+b_k)\int_{\Sigma}\frac{\xi^{n}}{\sqrt{\mathcal{R}(\xi)}_{+}}\dif\xi\\
=-2\pi i\cdot\frac{\beta}{2}\res_{z=\infty}\left[\frac{z^{n+1}}{\sqrt{\mathcal{R}(z)}}\right]+2\pi i\cdot\frac{\beta}{4}\sum_{k=0}^{n}(a_k+b_k)\res_{z=\infty}\left[\frac{z^{n}}{\sqrt{\mathcal{R}(z)}}\right].
\end{multline}
Using \eqref{equ: 1/sqrtR asy at z=infty} again, it's readily seen 
that $\res_{z=\infty}\left[{z^{n+1}}/{\sqrt{\mathcal{R}(z)}}\right]=\frac{1}{2}\sum_{k=0}^{n}(a_k+b_k)$ as well as $\res_{z=\infty}\left[{z^{n}}/{\sqrt{\mathcal{R}(z)}}\right]=1$, which implies that \eqref{equ: an equ ho to 0} is equal to $0$ and $\mathcal{D}_{\infty,1}\in i\mathbb{R}$.

\item[\rm (d)] For $j\in\{0,1,\dots,n\}$, we note that
\begin{align}\label{equ:split the integral of H/(xi-z)}
\frac{1}{2\pi i}\int_{\Sigma}\frac{\mathcal{H}(\xi)}{\xi-z}\dif\xi
&=\frac{1}{2\pi i}\int_{a_j}^{b_j}\frac{\log\xi^{-2\beta}+\sign{\xi}(\alpha-\beta)\pi i+\zeta_j}{\sqrt{\mathcal{R}(\xi)}_{+}(\xi-z)}\dif\xi
\nonumber 
\\
&\quad +
\frac{1}{2\pi i}\int_{\Sigma\setminus[a_j,b_j]}\frac{\mathcal{H}(\xi)}{\xi-z}\dif\xi.
\end{align}
Applying the formula (29.5) from \cite[Chapter 4, $\S$29]{Mus-SingularIntegrals-book} for $\gamma=\frac{1}{2}$ to the first term on the right-hand side of \eqref{equ:split the integral of H/(xi-z)}, one has, as $z\to a_j$,
\begin{multline}\label{equ:first term asy at aj}
\frac{1}{2\pi i}\int_{a_j}^{b_j}\frac{\log \xi^{-2\beta}+\zeta_j+\sign{\xi}(\alpha-\beta)\pi i}{(\xi-z)\sqrt{\mathcal{R}(\xi)}_{+}}\\
\sim\frac{1}{2}\cdot\frac{\log a_j^{-2\beta}+\zeta_j+\sign{a_j}(\alpha-\beta)\pi i}{i(-1)^{n-j}|\prod_{\substack{k=0 \\ k\neq j}}^{n}(a_j-a_k)
\prod_{k=0}^{n}(a_j-b_k)|^{\frac{1}{2}}}\cdot (z-a_j)^{-\frac{1}{2}}.
\end{multline}
Meanwhile, as $z\to a_j$ from $\mathbb{C}^+$,  direct computations show that 
\begin{align}\label{equ:second term asy at aj}
\frac{1}{2\pi i}\int_{\Sigma\setminus[a_j,b_j]}\frac{\mathcal{H}(\xi)}{\xi-z}\dif\xi=\frac{1}{2\pi i}\int_{\Sigma\setminus[a_j,b_j]}\frac{\mathcal{H}(\xi)}{\xi-a_j}\dif\xi+\mathcal{O}(z-a_j),
\end{align}
and 
\begin{align}\label{equ:sqrtR asy at aj}
\sqrt{\mathcal{R}(z)}=i(-1)^{n-j}(z-a_j)^{\frac{1}{2}}
\prod_{\substack{k=0 \\ k\neq j}}^{n}|a_j-a_k|^{\frac{1}{2}}
\prod_{k=0}^{n}|a_j-b_k|^{\frac{1}{2}}
\left(1+\mathcal{O}\left(z-a_j\right)\right).
\end{align}
A combination of \eqref{equ:split the integral of H/(xi-z)}--\eqref{equ:sqrtR asy at aj} gives us the statement in item (d) with $p=a_j$. Similarly, one can show the claim holds if $z\to b_j$ from $\mathbb{C}^+$.

\item[\rm (e)] 
To proceed, we note that 
\begin{align}\label{equ:split the integral of H/(xi-z) zto0}
&\frac{\sqrt{\mathcal{R}(z)}}{2\pi i}\int_{\Sigma}\frac{\mathcal{H}(\xi)}{\xi-z}\dif\xi
\nonumber
\\
&  =\frac{\sqrt{\mathcal{R}(z)}}{2\pi i}\int_{\Sigma\setminus[a_m,b_m]}
\frac{\mathcal{H}(\xi)}{\xi-z}\dif\xi+\frac{\sqrt{\mathcal{R}(z)}}{2\pi i}\int_{a_m}^{0}\frac{\zeta_m-(\alpha-\beta)\pi i}{\sqrt{\mathcal{R}(\xi)}_{+}(\xi-z)}
\nonumber \\
& \quad 
+\frac{\sqrt{\mathcal{R}(z)}}{2\pi i}\int_{0}^{b_m}\frac{\zeta_m+(\alpha-\beta)\pi i}{\sqrt{\mathcal{R}(\xi)}_{+}(\xi-z)}
+\frac{\sqrt{\mathcal{R}(z)}}{2\pi i}\int_{a_m}^{b_m}\frac{-2\beta\log\xi}{\sqrt{\mathcal{R}(\xi)}_{+}(\xi-z)}\dif\xi.
\end{align}
As $z\to 0$ from $\mathbb{C}^+$, we have
\begin{align}\label{equ: est D(z) zto0-a}
\frac{\sqrt{\mathcal{R}(z)}}{2\pi i}\int_{\Sigma\setminus[a_m,b_m]}
\frac{\mathcal{H}(\xi)}{\xi-z}\dif\xi=\mathcal{O}(1),
\end{align}
and by the formula (29.4) in \cite[Chapter 4, $\S$29]{Mus-SingularIntegrals-book} with $\gamma=0$,
\begin{multline}\label{equ: est D(z) zto0-b}
\frac{\sqrt{\mathcal{R}(z)}}{2\pi i}\int_{a_m}^{0}\frac{\zeta_m-(\alpha-\beta)\pi i}{\sqrt{\mathcal{R}(\xi)}_{+}(\xi-z)}\dif\xi+\frac{\sqrt{\mathcal{R}(z)}}{2\pi i}\int_{0}^{b_m}\frac{\zeta_m+(\alpha-\beta)\pi i}{\sqrt{\mathcal{R}(\xi)}_{+}(\xi-z)}\dif\xi\\
\sim -(\alpha-\beta)\log z.
\end{multline}
To deal with the last term on the right-hand side of \eqref{equ:split the integral of H/(xi-z) zto0}, we extend the function $\frac{\log \xi}{\sqrt{\mathcal{R}(\xi)}_{+}}$ analytically  to the upper half plane $\mathbb{C}^{+}$ and obtain
\begin{multline}\label{equ: Cauchy integral}
\frac{\sqrt{\mathcal{R}(z)}}{2\pi i}\int_{a_m}^{b_m}\frac{-2\beta\log\xi}{\sqrt{\mathcal{R}(\xi)}_{+}(\xi-z)}\dif\xi\\
=\frac{\sqrt{\mathcal{R}(z)}}{2\pi i}\oint_{\Gamma'}\frac{-2\beta\log\xi}{\sqrt{\mathcal{R}(\xi)}(\xi-z)}\dif\xi-\frac{\sqrt{\mathcal{R}(z)}}{2\pi i}\int_{\Gamma}\frac{-2\beta\log\xi}{\sqrt{\mathcal{R}(\xi)}(\xi-z)}\dif\xi,
\end{multline}
where $\Gamma':=\Gamma \cup (a_m,b_m)$; see Figure \ref{fig: Cauchy integral} for an illustration. As $z\to 0$ from $\mathbb{C}^+$, we have by Cauchy's formula that 
\begin{align}\label{equ: est D(z) zto0-c}
\frac{\sqrt{\mathcal{R}(z)}}{2\pi i}\int_{a_m}^{b_m}\frac{-2\beta\log\xi}{\sqrt{\mathcal{R}(\xi)}_{+}(\xi-z)}\dif\xi
=-2\beta\log z+\mathcal{O}(1).
\end{align}
The estimate \eqref{equ: asy mathcalD at 0} follows by combining \eqref{equ:split the integral of H/(xi-z) zto0}, \eqref{equ: est D(z) zto0-a}, \eqref{equ: est D(z) zto0-b} and 
\eqref{equ: est D(z) zto0-c}.
\end{itemize}
This completes the proof of Proposition \ref{prop:D-function}. 
\end{proof}
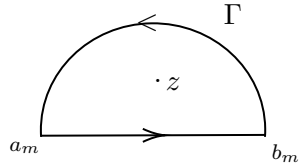
\begin{figure}[htbp]
\centering
\tikzset{every picture/.style={line width=0.75pt}} 
\begin{tikzpicture}[x=0.75pt,y=0.75pt,yscale=-1,xscale=1]
\draw    (248,153) -- (360,152.6) ;
\draw [shift={(310,152.78)}, rotate = 179.8] [color={rgb, 255:red, 0; green, 0; blue, 0 }  ][line width=0.75]    (10.93,-3.29) .. controls (6.95,-1.4) and (3.31,-0.3) .. (0,0) .. controls (3.31,0.3) and (6.95,1.4) .. (10.93,3.29)   ;
\draw  [draw opacity=0] (248.28,153.58) .. controls (248.22,152.7) and (248.18,151.8) .. (248.17,150.91) .. controls (247.76,121.27) and (272.49,96.9) .. (303.4,96.47) .. controls (334.32,96.04) and (359.72,119.71) .. (360.14,149.34) .. controls (360.16,150.77) and (360.12,152.18) .. (360.03,153.57) -- (304.15,150.13) -- cycle ; \draw   (248.28,153.58) .. controls (248.22,152.7) and (248.18,151.8) .. (248.17,150.91) .. controls (247.76,121.27) and (272.49,96.9) .. (303.4,96.47) .. controls (334.32,96.04) and (359.72,119.71) .. (360.14,149.34) .. controls (360.16,150.77) and (360.12,152.18) .. (360.03,153.57) ;  
\draw (231,154.4) node [anchor=north west][inner sep=0.75pt]  [font=\scriptsize]  {$a_m$};
\draw (362,156) node [anchor=north west][inner sep=0.75pt]  [font=\scriptsize]  {$b_m$};
\draw (294,90.4) node [anchor=north west][inner sep=0.75pt]    {$< $};
\draw (309,122.4) node [anchor=north west][inner sep=0.75pt]    {$z$};
\draw (303,123) node [anchor=north west][inner sep=0.75pt]   [align=left] {.};
\draw (338,85.4) node [anchor=north west][inner sep=0.75pt]    {$\Gamma $};
\end{tikzpicture}
\caption{The contour of integral in \eqref{equ: Cauchy integral}.}
\label{fig: Cauchy integral}
\end{figure}

With $\mathcal{D}_{\infty}$ and $\mathcal{D}$ given in \eqref{def:Dinfty1-pre} and \eqref{def:D-function}, we define
\begin{align}\label{def:hatPinfty}
\hat{P}^{(\infty)}(z):=\mathcal{D}_{\infty}^{-\sigma_3}P^{(\infty)}(z)z^{\beta\sigma_3}\mathcal{D}(z)^{\sigma_3}.
\end{align}
It's readily seen that $\hat{P}^{(\infty)}$ satisfies the following RH problem.

\paragraph{RH problem for $\hat{P}^{(\infty)}$}
\begin{itemize} 
\item [\rm (a)] $\hat{P}^{(\infty)}(z)$ is holomorphic for $z\in\mathbb{C}\setminus\overline{\Sigma}$.
\item[\rm (b)] For $z\in\Sigma$, we have 
 \begin{equation}
    \hat{P}^{(\infty)}_{+}(z)= \hat{P}^{(\infty)}_{-}(z)J_{\hat{P}^{(\infty)}}(z),
  \end{equation}
  where
\begin{equation}
      J_{\hat{P}^{(\infty)}}(z)=\left\{\begin{array}{ll}
  {\begin{pmatrix}
0 & -1\\
1 & 0
\end{pmatrix}},& {z\in(a_0,b_0)}, \\
  {
\begin{pmatrix}
0 & -e^{-2\pi iV_j}\\
e^{2\pi iV_j} & 0
\end{pmatrix}},& {z\in\cup_{j=1}^{n}(a_j,b_j)},
\end{array}\right.
\end{equation}
and $V_j$ is given in \eqref{def:Vj}.
\item [\rm (c)] As $z\to\infty$, we have 
\begin{align}
\hat{P}^{(\infty)}(z)=I+\frac{\hat{P}^{(\infty)}_1(s)}{z}+\mathcal{O}{(z^{-2})},
\end{align}
where $\hat{P}^{(\infty)}_1$ is independent of $z$.
\item [\rm (d)] As $z\to p $ with $p\in\mathcal{I}_{e}$, we have
\begin{align}
\hat{P}^{(\infty)}(z)=\mathcal{O}(\left(z-p\right)^{-\frac{1}{4}}).
\end{align}
\end{itemize}
Following \cite{Dei-Its-Zhou-Ann1997}, we can solve the above RH problem explicitly by using multi-dimensional Riemann $\theta$-functions, hence the RH problem for $P^{(\infty)}$ by \eqref{def:hatPinfty}. 

\paragraph{$\mathcal{E}$ model problem} 
Define 
\begin{align}\label{def:mathcalE}
\mathcal{E}(z):=\prod_{k=0}^{n}\left(\frac{z-b_k}{z-a_k}\right)^{\frac{1}{4}}, \quad z\in\mathbb{C}\setminus\overline{\Sigma},
\end{align}
where the branch cut is chosen such that $\mathcal{E}(z)>0$ for $z>b_n$. Clearly, $\mathcal{E}_{+}(z)=i\mathcal{E}_{-}(z)$ for $z\in\Sigma$. Let 
\begin{align}\label{def:mathcalN}
\mathcal{N}(z):=\begin{pmatrix}
\frac{\mathcal{E}(z)+\mathcal{E}(z)^{-1}}{2} & -\frac{\mathcal{E}(z)-\mathcal{E}(z)^{-1}}{2i}\\
\frac{\mathcal{E}(z)-\mathcal{E}(z)^{-1}}{2i} & \frac{\mathcal{E}(z)+\mathcal{E}(z)^{-1}}{2}
\end{pmatrix}.
\end{align}
It is readily seen that $\mathcal{N}$ satisfies the following RH problem.
\paragraph{RH problem for $\mathcal{N}$}
\begin{itemize}
\item [\rm (a)] $\mathcal{N}(z)$ is holomorphic for $z\in\mathbb{C}\setminus\overline{\Sigma}$.
\item [\rm (b)] For $z\in\Sigma$, we have  
\begin{align}\label{equ:jump of mathcalN}
\mathcal{N}_{+}(z)=\mathcal{N}_{-}(z)
\begin{pmatrix}
    0 & -1 \\
    1 & 0
\end{pmatrix}.
\end{align}
\item [\rm (c)] As $z\to\infty$, we have $\mathcal{N}(z)=I+\mathcal{O}(z^{-1})$.
\end{itemize}
The following proposition reveals the distribution of zeros for the functions $\mathcal{E}(z)\pm \mathcal{E}(z)^{-1}$.
\begin{proposition}\label{prop: zeros of mathcal{E}}
    The function $\mathcal{E}(z)-\mathcal{E}(z)^{-1}$ has exactly $n$ zeros denoted by $z_j$ such that $z_j\in (b_{j-1}, a_j)$ for $j=1,\dots,n$. The function $\mathcal{E}(z)+\mathcal{E}(z)^{-1}$ has no zeros in 
    $\mathbb{C}\setminus\overline{\Sigma}$.
\end{proposition}
\begin{proof}
If $\mathcal{E}(z)\pm\mathcal{E}(z)^{-1}=0$, we have $\mathcal{E}(z)^2\pm 1=0$, or equivalently, $\mathcal{E}(z)^{4}=1$. This implies that 
\begin{align}
\varsigma(z):=\prod_{i=0}^{n}(z-b_i)-\prod_{i=0}^{n}(z-a_i)=0.
\end{align}
By a straightforward calculation, it follows that $\sign{\varsigma(b_{j-1})}=(-1)^{n-j}=-\sign{\varsigma(a_j)}$ for $j=1,\dots,n$. Hence $\varsigma(z)$ has at least one zero $z_j$ on the gap $(b_{j-1},a_j)$. Since $\varsigma(z)$ is a polynomial of degree $n$, these $\{z_{j}\}_{j=1}^{n}$ are precisely all the zeros. Although $\mathcal{E}(z_j)^{2}=\pm 1$, the sign of the root (given our choice of the branch cut) implies that $\mathcal{E}(z)^{2}>0$ for $z$ belonging to the gap, which finally shows that $\{z_{j}\}_{j=1}^{n}$ are the zeros of $\mathcal{E}(z)-\mathcal{E}(z)^{-1}$, whereas $\mathcal{E}(z)+\mathcal{E}(z)^{-1}$ has no zeros in $\mathbb{C}\setminus\overline{\Sigma}$.
\end{proof}
\begin{remark}
In fact, all the zeros of $\mathcal{E}(z)-\mathcal{E}(z)^{-1}$ are on the first sheet of the Riemann surface $\mathcal{W}$, while all the zeros of $\mathcal{E}(z)+\mathcal{E}(z)^{-1}$ lie on the second sheet of the Riemann surface $\mathcal{W}$.
\end{remark}

\paragraph{Abel's map and related properties}
Recall the Abel's map  defined in \eqref{equ:Abel map on mathbbC}, we collect some of its properties in the next proposition.  
\begin{proposition}\label{prop:jump of mathcalA}
With $\vec{\mathcal{A}}(z)$ defined in \eqref{equ:Abel map on mathbbC}, we have, for $j=1,\dots,n$,
    \begin{align}
    &\vec{\mathcal{A}}_{+}(z)+\vec{\mathcal{A}}_{-}(z)=-\vec{\tau}_j, && z\in(a_j,b_j),\label{equ:jump of mathcalA-equ1}\\
    &\vec{\mathcal{A}}_{+}(z)+\vec{\mathcal{A}}_{-}(z)=\vec{0}, && z\in(a_0,b_0),
    \label{equ:jump of mathcalA-equ2}\\
    &\vec{\mathcal{A}}_{+}(z)-\vec{\mathcal{A}}_{-}(z)=-\sum_{k=j}^{n}\vec{e}_k, &&
    z\in(b_{j-1},a_j), \label{equ:jump of mathcalA-equ3}\
    \end{align}
where $\vec{\tau}_j$ is $j$-th column of the matrix $\tau$ defined in \eqref{def:tau matrix} and $\vec{e}_j$ denotes the standard column vector in $\mathbb{C}^{n}$ with $1$ in the $j$-th position and zero elsewhere. Moreover, we have
\begin{align}
& \vec{\mathcal{A}}(a_0)=\vec{0}, \quad \vec{\mathcal{A}}(b_n)=-\frac{1}{2}\vec{\tau}_n,
\label{equ: vecA(a_j) and vecA(b_j)-a}\\
&\vec{\mathcal{A}}_{\pm}(a_j)=\mp\frac{1}{2}\sum_{k=j}^{n}\vec{e}_k-\frac{1}{2}\vec{\tau}_j,  &&j=1,\dots,n,\label{equ: vecA(a_j) and vecA(b_j)-b}\\
& \vec{\mathcal{A}}_{\pm}(b_j)=\mp\frac{1}{2}\sum_{k=j+1}^{n}\vec{e}_k-\frac{1}{2}\vec{\tau}_j,  &&j=0, 1,\dots,n-1.\label{equ: vecA(a_j) and vecA(b_j)-c}
\end{align}
\end{proposition}
\begin{proof}
All the claims follow directly from the definition of the Abel's map. We only give the proofs of \eqref{equ:jump of mathcalA-equ1}
and \eqref{equ:jump of mathcalA-equ3}, and leave the other proofs to the interested readers. 

By a straightforward calculation, it follows that, for  $j=1,\dots,n$,
    \begin{align}
    \vec{\mathcal{A}}_{+}(z)+\vec{\mathcal{A}}_{-}(z)
&=-2\left(\int_{a_1}^{b_0}+\int_{a_2}^{b_1}+\dots+\int_{a_j}^{b_{j-1}}\right)\vec{\omega}^{\rm T}_{+}
\nonumber 
\\
    &=-\oint_{B_j}\vec{\omega}^{\rm T}=-\vec{\tau}_j, \qquad z\in(a_j,b_j),
    \end{align}
which is  \eqref{equ:jump of mathcalA-equ1}.  Similarly, 
    \begin{align}\label{equ:abelA+-abelA-}
    \vec{\mathcal{A}}_{+}(z)-\vec{\mathcal{A}}_{-}(z)
    &=2\left(\int_{a_0}^{b_0}+\int_{a_1}^{b_1}+\dots+\int_{a_{j-1}}^{b_{j-1}}\right)\vec{\omega}^{\rm T}_{+}
    \nonumber
    \\
    &=2\int_{a_0}^{b_0}\vec{\omega}^{\rm T}_{+}+\sum_{k=1}^{j-1}\vec{e}_k, \qquad z\in (b_{j-1},a_j).
    \end{align}
Note that
\begin{align}\label{equ:2inta_0b_0vecomega}
2\int_{a_0}^{b_0}\vec{\omega}^{\rm T}=-\sum_{k=j}^{n}\oint_{A_k}\vec{\omega}^{\rm T}=-\sum_{k=1}^{n}\vec{e}_k,
\end{align}
we are led to \eqref{equ:jump of mathcalA-equ3} by substituting \eqref{equ:2inta_0b_0vecomega} into \eqref{equ:abelA+-abelA-}.
\end{proof}
By \cite[Page 303, Corollary]{FK-Riemann-Surface}, there are a total of $2^{2n}$ half periods associated to the Jacobian variety $J(\mathcal{W}):=\mathbb{C}^n/\Lambda(\mathcal{W})$ with $\Lambda(\mathcal{W}):=\mathbb{Z}^{n}+\tau\mathbb{Z}^{n}$, and 
$2^{n-1}(2^n-1)$ are odd half periods. From \cite[Page 304, Remark 2]{FK-Riemann-Surface}, it follows that the $\theta$-function vanishes at each odd half-period. One can see that 
the points $\vec{\mathcal{A}}(a_j)$, $j=1,2,\dots,n$, are 
odd half-periods. Indeed, it follows from \eqref{equ: vecA(a_j) and vecA(b_j)-b} that 
\begin{align}
\vec{\mathcal{A}}(a_j)=\frac{1}{2}\vec{\lambda}_{j}+\frac{1}{2}\tau\vec{\mu}_j \quad 
{\rm mod} \ \Lambda(\mathcal{W}),
\end{align}
for certain vectors $\vec{\lambda}_{j}$ and $\vec{\mu}_j$ belonging to $\mathbb{Z}^{n}$. As $\vec{\lambda}_{j}^{\rm T}\mu_j$ is an odd number, we have that
$\vec{\mathcal{A}}(a_j)$ are odd half-periods. Now we define
\begin{align}\label{def:choice of d}
\vec{d}:=\mathcal{K}+\sum_{j=1}^{n}\int_{a_0}^{P_1(z_j)}\vec{\omega}^{\rm T},
\end{align}
where $\mathcal{K}$ is the vector of Riemann constants to be determined and $P_{1}(z_j)$ denotes the pre-images of the zero $z_j\in(b_{j-1},a_j)$ given in Proposition \ref{prop: zeros of mathcal{E}} on the first sheet of $\mathcal{W}$. Similarly, we also denote by $P_{2}(z_j)$ to be the pre-images of $z_j$ on the second sheet. 
Recall that we have $n$ points $\vec{\mathcal{A}}(a_j)$, $j=1,\dots,n$, which are odd half periods in $J(\mathcal{W})$. Utilizing
\cite[Page 325, VII.1.2, Equation (1.2.1)]{FK-Riemann-Surface}, we obtain
\begin{align}\label{equ: Riemann Constant}
\mathcal{K}:=\sum_{j=1}^{n}\vec{\mathcal{A}}(a_j)=
\sum_{j=1}^{n}\int_{a_0}^{a_j}\vec{\omega}^{\rm T}
\quad {\rm mod} \ \Lambda(\mathcal{W}).
\end{align}

The final proposition in this part is an analogue of \cite[Lemmas 3.27 and 3.36]{Dei-Its-Zhou-Ann1997} and we omit the proof here. 
\begin{proposition}\label{prop:zeros of theta function}
The multi-valued function $\theta(\vec{\mathcal{A}}(P)\pm \vec{d})$ is not identically zero on $\mathcal{W}$. $\theta(\vec{\mathcal{A}}(P)+\vec{d})=0$ if and only if $P\in\left\{P_2(z_k)\right\}_{k=1}^{n}$, and $\theta(\vec{\mathcal{A}}(P)-\vec{d})=0$ if and only if $P\in\left\{P_1(z_k)\right\}_{k=1}^{n}$.
In addition,  one has
    \begin{align}\label{equ: vecA(infty)+vecd== mod Lambda}
    \vec{\mathcal{A}}(\infty)+\vec{d}=0 \quad {\rm mod} \ \Lambda,
    \end{align}
and 
    \begin{align}
    \theta\left(\vec{\mathcal{A}}(\infty)+\vec{V}(s)+\vec{d}\right)\neq 0, \qquad  s\in\mathbb{R},
    \end{align}
where $\vec{V}(s)$ is defined in \eqref{def-intro-vecV} and \eqref{def:Vj} and $\vec{d}$ is given in \eqref{def:choice of d}.
\end{proposition}

\paragraph{The $\theta$-function parametrix} 
Define
\begin{align}\label{def:mathcalG}
\mathcal{G}(\vec{z}):=\mathcal{G}(\vec{z};s)=\frac{\theta\left(\vec{z}+\vec{V}(s)\right)}{\theta(\vec{z})}, \quad \vec{z}\in\mathbb{C}^{n}.
\end{align}
By Proposition \ref{prop:jump of mathcalA} and properties of  the $\theta$-function, the next proposition follows.
\begin{proposition}
With $\mathcal{G}$ defined in \eqref{def:mathcalG}, the functions 
$\mathcal{G}(\vec{\mathcal{A}}(z)\pm\vec{d})$ and $\mathcal{G}(-\vec{\mathcal{A}}(z)\pm\vec{d})$ are single-valued and meromorphic for $z\in\mathbb{C}\setminus\overline{\Sigma}$, and 
satisfy the jump relations
\begin{align}
&\mathcal{G}_{+}\left(\vec{\mathcal{A}}(z)\pm\vec{d}\right)=\mathcal{G}_{-}\left(-\vec{\mathcal{A}}(z)\pm\vec{d}\right)e^{2\pi iV_{j}(s)}, & z\in(a_j,b_j), \label{equ:jump of mathcalG-1}\ j=0,1,\dots,n,\\
&\mathcal{G}_{+}\left(-\vec{\mathcal{A}}(z)\pm\vec{d}\right)=\mathcal{G}_{-}\left(\vec{\mathcal{A}}(z)\pm\vec{d}\right)e^{-2\pi iV_{j}(s)}, & z\in(a_j,b_j), \ j=0,1,\dots,n, \label{equ:jump of mathcalG-2}
\end{align}
where $V_0=0$ and $V_j$, $j=1,\dots,n$, are defined in \eqref{def:Vj}. Particularly, the relations \eqref{equ:jump of mathcalG-1} and \eqref{equ:jump of mathcalG-2} are also valid for $z\in\mathcal{I}_{e}$.
\end{proposition}
\begin{proof}
For $z\in(b_{j-1},a_j)$, $j=1,\dots,n$, it follows from \eqref{equ:jump of mathcalA-equ3} and the periodicity of the $\theta$-function that 
\begin{align*}
    \theta(\vec{\mathcal{A}}_{+}(z)+\vec{V}(s))&=\theta(\vec{\mathcal{A}}_{-}(z)-\sum_{k=j}^{n}\vec{e}_k+\vec{V}(s))= \theta(\vec{\mathcal{A}}_{-}(z)+\vec{V}(s)),
    \\
    \theta(\vec{\mathcal{A}}_{+}(z))&=\theta(\vec{\mathcal{A}}_{-}(z)-\sum_{k=j}^{n}\vec{e}_k)=\theta(\vec{\mathcal{A}}_{-}(z)). 
\end{align*}
Moreover, $\theta(\vec{\mathcal{A}}(z)+\vec{d})$ and $\theta(-\vec{\mathcal{A}}(z)+\vec{d})=\theta(\vec{\mathcal{A}}(z)-\vec{d})$ are not identically equal to zero by Proposition \ref{prop:zeros of theta function}. As a consequence, we conclude that $\mathcal{G}(\pm\vec{\mathcal{A}}(z)+\vec{d})$ are single-valued and meromorphic for $z\in\mathbb{C}\setminus\overline{\Sigma}$. The arguments for $\mathcal{G}(\pm\vec{\mathcal{A}}(z)-\vec{d})$ are similar.

To show the jump realtions  \eqref{equ:jump of mathcalG-1}, we see from  \eqref{equ:jump of mathcalA-equ2} that it holds for $j=0$. If $z\in(a_j,b_j)$, $j=1,\dots,n$, one has
\begin{align}
\mathcal{G}_{+}\left(\vec{\mathcal{A}}(z)\pm\vec{d}\right)
&=\frac{\theta\left(\vec{\mathcal{A}}_{+}(z)+\vec{V}(s)\pm\vec{d}\right)}{\theta\left(\vec{\mathcal{A}}_{+}(z)\pm\vec{d}\right)}
\overset{\eqref{equ:jump of mathcalA-equ1}}{=}
\frac{\theta\left(-\vec{\mathcal{A}}_{-}(z)-\vec{\tau}_j+\vec{V}(s)\pm\vec{d}\right)}{\theta\left(-\vec{\mathcal{A}}_{-}(z)-\vec{\tau}_j\pm\vec{d}\right)}\nonumber\\
&\overset{\eqref{property:theta function}}{=}
\frac{e^{2\pi i\left(-\vec{\mathcal{A}}_{-,j}(z)+V_{j}(s)\pm d_{j}\right)-\pi i\tau_{jj}}\theta\left(-\vec{\mathcal{A}}_{-}(z)+\vec{V}(s)\pm\vec{d}\right)}{e^{2\pi i\left(-\vec{\mathcal{A}}_{-,j}(z)\pm d_{j}\right)-\pi i\tau_{jj}}\theta\left(-\vec{\mathcal{A}}_{-}(z)\pm\vec{d}\right)}\nonumber\\
&=\mathcal{G}_{-}\left(-\vec{\mathcal{A}}(z)\pm \vec{d}\right)e^{2\pi iV_{j}(s)},
\end{align}
as required, where the subscript $_j$ denotes the $j$-th element of the corresponding vector. A further appeal to \eqref{equ: vecA(a_j) and vecA(b_j)-a}--\eqref{equ: vecA(a_j) and vecA(b_j)-c} shows that \eqref{equ:jump of mathcalG-1} is also valid for $z\in\mathcal{I}_{e}$.

Finally, the proof of \eqref{equ:jump of mathcalG-2} is similar to that of  \eqref{equ:jump of mathcalG-1}, and we omit the details here.  
\end{proof}

With the aid of $\mathcal{G}$, we define
\begin{align}\label{def:Q(z)}
Q(z):=
\begin{pmatrix}
    \mathcal{N}_{11}(z)\mathcal{G}\left(\vec{\mathcal{A}}(z)+\vec{d}\right) & \mathcal{N}_{12}(z)\mathcal{G}\left(-\vec{\mathcal{A}}(z)+\vec{d}\right) \\ 
    \mathcal{N}_{21}(z)\mathcal{G}\left(\vec{\mathcal{A}}(z)-\vec{d}\right) &
    \mathcal{N}_{22}(z)\mathcal{G}\left(-\vec{\mathcal{A}}(z)-\vec{d}\right)
\end{pmatrix},
\end{align}
where recall that $\mathcal{N}_{ij}$ is the $(i,j)$-entry of $\mathcal{N}$ defined in \eqref{def:mathcalN}. The next proposition shows that $Q$ is the solution of the following RH problem. 
\begin{proposition}\label{Prop:Q}
The function $Q$ defined in \eqref{def:Q(z)} satisfies the following RH problem.
\paragraph{RH problem for $Q$}
\begin{itemize}
\item [\rm (a)] $Q(z)$ is holomorphic for $z\in\mathbb{C}\setminus\overline{\Sigma}$.
\item [\rm (b)] For $z\in(a_j,b_j)$, $j=0,1,\dots,n$, we have
\begin{align}\label{equ:jump of Q}
Q_{+}(z)=Q_{-}(z)
\begin{pmatrix}
    0 & -e^{-2\pi iV_{j}(s)} \\
    e^{2\pi iV_{j}(s)} & 0
\end{pmatrix},
\end{align}
where $V_{0}=0$, and $V_{j}$, $j=1,\dots,n$, are given in \eqref{def:Vj}.
\item [\rm (c)] As $z\to\infty$, we have
\begin{align}\label{eq:Qasy}
Q(z)=Q_{\infty}+\mathcal{O}(z^{-1}),
\end{align}
where
\begin{align}\label{def:Qinfty}
Q_{\infty}:=
\begin{pmatrix}
    \mathcal{G}\left(\vec{\mathcal{A}}(\infty)+\vec{d}\right) & 0\\
    0 & \mathcal{G}\left(-\vec{\mathcal{A}}(\infty)-\vec{d}\right)
\end{pmatrix}.
\end{align}
\item [\rm (d)] As $z\to 0$ from $\mathbb{C}^{+}$, we have
\begin{align}\label{equ:Qasyto0+}
Q(z)=Q_0+\mathcal{O}(z),
\end{align}
where $Q_0$ is a constant matrix.
\end{itemize}
\end{proposition}
\begin{proof}
\begin{itemize}
\item [\rm (a)] To check the analyticity in $\mathbb{C}\setminus\overline{\Sigma}$, one only needs to show that there are no poles introduced from the zeros of the $\theta$-function appeared in the denominators of $\mathcal{G}$. Since $\theta(\vec{\mathcal{A}}(z)+\vec{d})\neq 0$ for $z\in\mathbb{C}\setminus\overline{{\Sigma}}$, it follows that the diagonal entries $Q_{11}$ and $Q_{22}$ have no poles. Next, we have by Proposition \ref{prop:zeros of theta function} that $\theta(\vec{\mathcal{A}}(z)-\vec{d})=0$ if and only if $z=z_j(=P_1(z_j))$ for $j=1,\dots,n$. However, the poles resulted from the denominators of $\mathcal{G}\left(-\vec{\mathcal{A}}(z)+\vec{d}\right)$ and $\mathcal{G}\left(\vec{\mathcal{A}}(z)-\vec{d}\right)$ are exactly canceled by the zeros of $\mathcal{E}(z)-\mathcal{E}(z)^{-1}$ in $\mathcal{N}_{12}$ and $\mathcal{N}_{21}$. This implies that $Q_{12}$ and $Q_{21}$ have no poles in $\mathbb{C}\setminus\overline{\Sigma}$ as well. 
    \item [\rm (b)] By  \eqref{equ:jump of mathcalN}, \eqref{equ:jump of mathcalG-1} and \eqref{equ:jump of mathcalG-2}, the jump relation \eqref{equ:jump of Q} for $Q$
    follows immediately.
    \item [\rm (c)] Since $\mathcal{N}(z)=I+\mathcal{O}(z^{-1})$ as $z\to\infty$, the asymptotics \eqref{eq:Qasy} and \eqref{def:Qinfty} follows directly from \eqref{def:Q(z)}.
    \item [\rm (d)] As $z\to 0$ from $\mathbb{C}^{+}$, one has
    \begin{align}\label{equ:mathcaN asy to 0}
        \mathcal{N}(z)=\frac{1}{2}\begin{pmatrix}
            c_0^{(+)}+\mathcal{O}(z) & ic_0^{(-)}+\mathcal{O}(z) \\
            -ic_0^{(-)}+\mathcal{O}(z) & c_0^{(+)}+\mathcal{O}(z)
        \end{pmatrix},
    \end{align}
    where $c_0^{(\pm)}=e^{i\pi/4}\prod_{k=0}^{n}|b_k/a_k|^{1/4}\pm e^{-i\pi/4}\prod_{k=0}^{n}|a_k/b_k|^{1/4}$. Also using \eqref{equ:holo differential}, we have 
    $\mathcal{G}(\pm\vec{\mathcal{A}}(z)\pm\vec{d})=\pm\tilde{c}+\mathcal{O}(z)$,  as $z\to 0$ from $\mathbb{C}^{+}$, where $\tilde{c}$ denotes some constant. Combining this with \eqref{equ:mathcaN asy to 0}, we obtain \eqref{equ:Qasyto0+}.
\end{itemize}
This completes the proof of Proposition \ref{Prop:Q}. 
\end{proof}

We are now ready to solve the RH problem for $P^{(\infty)}$.
\begin{proposition}
    The solution to the RH problem for $P^{(\infty)}$ is given by  
\begin{align}\label{equ:sol of Pinfty}
P^{(\infty)}(z):=\mathcal{D}_{\infty}^{\sigma_3}Q_{\infty}^{-1}Q(z)\mathcal{D}(z)^{-\sigma_3}z^{-\beta\sigma_3},
\end{align}
where $\mathcal{D}$, $\mathcal{D}_{\infty}$, $Q$, and $Q_{\infty}$ are defined in \eqref{def:D-function}, \eqref{def:Dinfty}, \eqref{def:Q(z)}, \eqref{def:Qinfty},  respectively.
\end{proposition}
\begin{proof}
We only need to check $P^{(\infty)}$ satisfies items (a)--(e) of the RH problem. 
\begin{itemize}
\item [(a)] Note that the branch cut of $z^{-\beta\sigma_3}$ is taken along the negative imaginary axis,  this follows directly from the fact that $\mathcal{D}(z)$ and $Q(z)$ are analytic in $\mathbb{C}\setminus\overline{\Sigma}$. 


\item [(b)] For $z\in(a_j,b_j)$, $j=0,1,\dots,n$, we have
\begin{align}\label{eq:Pinfty}
&{P^{(\infty)}}(z)^{-1}P^{(\infty)}(z)
=z^{\beta\sigma_3}\mathcal{D}_{-}^{\sigma_3}(z)Q_{-}(z)^{-1}Q_{+}(z)\mathcal{D}_{+}(z)^{-\sigma_3}z^{-\beta\sigma_3}
\nonumber\\
&\overset{\eqref{equ:jump of Q}}{=}
\begin{pmatrix}
    0 & -e^{-2\pi iV_{j}(s)}\mathcal{D}(z)_{-}\mathcal{D}_{+}(z)z^{2\beta}
    \\
    e^{2\pi iV_{j}(s)}\mathcal{D}_{-}(z)^{-1}\mathcal{D}_{+}(z)^{-1}z^{-2\beta} & 0
\end{pmatrix}.
\end{align}
This, together with the definition of $V_{j}(s)$ given in \eqref{def:Vj} and \eqref{eq:Djump}, implies the jump relation \eqref{equ:jump of Pinfty} on $\Sigma$.
As $Q(z)$ and $\mathcal{D}(z)$ have no jumps on $(-i\infty,i0)$, we have 
for $z\in(-i\infty,0)$,
\begin{align}
{P^{(\infty)}}(z)^{-1}P^{(\infty)}(z)=z_{-}^{\beta\sigma_3}z_{+}^{-\beta\sigma_3}=e^{2\pi i\beta\sigma_3}. 
\end{align}
\item [(c)] 
As $\vec{\omega}$ given in \eqref{equ:holo differential} is a holomorphic differential, it follows from \eqref{equ:Abel map on mathbbC} that 
\begin{align}\label{equ: asy of mathcalA at infty}
\vec{\mathcal{A}}(z)=\vec{\mathcal{A}}(\infty)-
\frac{\left(\tilde{\mathbb{A}}^{-1}\right)_{n}^{\rm T}}{z}+\mathcal{O}(z^{-2}), \qquad z\to\infty,
\end{align}
where $\tilde{\mathbb{A}}$ is given in \eqref{def:tilde matrix A} and $(\tilde{\mathbb{A}}^{-1})_{n}$ denotes the $n$-th row of $\tilde{\mathbb{A}}^{-1}$. With the aid of \eqref{equ: asy of mathcalA at infty}, we have, as $z\to\infty$, 
\begin{align}
&\mathcal{G}(\vec{\mathcal{A}}(z)+\vec{d})=\mathcal{G}(\vec{\mathcal{A}}(\infty)+\vec{d})-\left(\tilde{\mathbb{A}}^{-1}\right)_{n}\nabla\mathcal{G}(\vec{\mathcal{A}}(\infty)+\vec{d})\cdot z^{-1}+\mathcal{O}(z^{-2}), \label{equ: G expansion at infty-1}
\\
&\mathcal{G}(-\vec{\mathcal{A}}(z)-\vec{d})=\mathcal{G}(-\vec{\mathcal{A}}(\infty)-\vec{d})+\left(\tilde{\mathbb{A}}^{-1}\right)_{n}\nabla\mathcal{G}(-\vec{\mathcal{A}}(\infty)-\vec{d})\cdot z^{-1}+\mathcal{O}(z^{-2}), \label{equ: G expansion at infty-2}
\end{align} where $\nabla\mathcal{G}$ is the gradient of $\mathcal{G}$ that is of size $n\times 1$. Inserting \eqref{equ: expansion of mathcalD at z=infty} and the above two formulae into \eqref{equ:sol of Pinfty} gives us \eqref{equ:asy Pinfty at z=infty}. In particular, by \eqref{equ: expansion of mathcalD at z=infty}, \eqref{equ: G expansion at infty-1} and \eqref{equ: G expansion at infty-2}, we have
\begin{align}
& \left(P^{(\infty)}_{1}\right)_{11}=-\frac{\left(\tilde{\mathbb{A}}^{-1}\right)_{n}\cdot\nabla\mathcal{G}(\vec{\mathcal{A}}(\infty)+\vec{d})}{\mathcal{G}(\vec{\mathcal{A}}(\infty)+\vec{d})}-\mathcal{D}_{\infty,1},\label{equ:P(infty)1,11}\\
& \left(P^{(\infty)}_{1}\right)_{22}=\frac{\left(\tilde{\mathbb{A}}^{-1}\right)_{n}\cdot\nabla\mathcal{G}(-\vec{\mathcal{A}}(\infty)-\vec{d})}{\mathcal{G}(-\vec{\mathcal{A}}(\infty)-\vec{d})}+\mathcal{D}_{\infty,1}.\label{equ:P(infty)1,22}
\end{align}
\item [(d)] 
For $j=0,1,\dots,n$, it follows from \eqref{def:mathcalE} that,  if $z\to a_j$ from $\mathbb{C}^+$,
\begin{align}\label{equ: asy of mathcalE at z=a_j}
\mathcal{E}(z)=\mathcal{E}_{a_j}^{(-\frac{1}{4})}(z-a_j)^{-\frac{1}{4}}+\mathcal{O}((z-a_j)^{\frac{3}{4}}), \quad \mathcal{E}_{a_j}^{(-\frac{1}{4})}=e^{\frac{i\pi}{4}}\frac{\prod_{k=0}^{n}|a_j-b_k|^{\frac{1}{4}}}{\prod_{\substack{k=0 \\ k\neq j}}^{n}|a_j-a_k|^{\frac{1}{4}}},
\end{align}
and if $z\to b_j$ from $\mathbb{C}^+$,
\begin{align}
{\mathcal{E}(z)}^{-1}=\mathcal{E}_{b_j}^{(-\frac{1}{4})}(z-b_j)^{-\frac{1}{4}}+\mathcal{O}((z-b_j)^{\frac{3}{4}}), \quad \mathcal{E}_{b_j}^{(-\frac{1}{4})}=\frac{\prod_{k=0}^{n}|b_j-a_k|^{\frac{1}{4}}}{\prod_{\substack{k=0 \\ k\neq j}}^{n}|b_j-b_k|^{\frac{1}{4}}}.
\end{align}
We then obtain from \eqref{def:mathcalN} that if $z\to a_j$ from $\mathbb{C}^+$,
\begin{align}\label{equ: asy expansion of mathcalN at a_j}
\mathcal{N}(z)=\mathcal{N}_{a_j}^{(-\frac{1}{4})}(z-a_j)^{-\frac{1}{4}}+\mathcal{O}((z-a_j)^{\frac{1}{4}}),
\quad \mathcal{N}_{a_j}^{(-\frac{1}{4})}=\frac{1}{2}\mathcal{E}_{a_j}^{(-\frac{1}{4})}
\begin{pmatrix}
    1 & i \\
    -i & 1
\end{pmatrix}.
\end{align}
and if $z\to b_j$ from  $\mathbb{C}^+$, 
\begin{align}\label{equ: asy expansion of mathcalN at b_j}
\mathcal{N}(z)=\mathcal{N}_{b_j}^{(-\frac{1}{4})}(z-b_j)^{-\frac{1}{4}}+\mathcal{O}((z-b_j)^{\frac{1}{4}}),
\quad
\mathcal{N}_{b_j}^{(-\frac{1}{4})}=\frac{1}{2}\mathcal{E}_{b_j}^{(-\frac{1}{4})}
\begin{pmatrix}
    1 & -i \\
    i & 1
\end{pmatrix}.
\end{align}
These, together with \eqref{def:Q(z)}, \eqref{equ:sol of Pinfty} and \eqref{equ: asy mathcalD at endpoints}, imply that
\begin{align}\label{equ: Pinfty expansion at z=p}
    P^{(\infty)}(z)=\left(P^{(\infty)}\right)_{p}^{(-\frac{1}{4})}(z-p)^{-\frac{1}{4}}+\mathcal{O}(\left(z-p\right)^{\frac{1}{4}}), \qquad z\to p, 
\end{align}
where 
\begin{align}
&\left(P^{(\infty)}\right)_{a_j}^{(-\frac{1}{4})}=\frac{1}{2}\mathcal{E}_{a_j}^{(-\frac{1}{4})}\mathcal{D}_{\infty}^{\sigma_3}Q_{\infty}^{-1}
\begin{pmatrix}
\mathcal{G}\left(\vec{\mathcal{A}}(a_j)+\vec{d}\right) & i\mathcal{G}(-\vec{\mathcal{A}}(a_j)+\vec{d}) \\
-i\mathcal{G}\left(\vec{\mathcal{A}}(a_j)-\vec{d}\right) & \mathcal{G}\left(-\vec{\mathcal{A}}(a_j)-\vec{d}\right)
\end{pmatrix}a_j^{-\beta\sigma_3}\mathcal{D}_{a_j}^{-\sigma_3}, 
\label{equ:Pinftya_j(-1/4)}
\\
&\left(P^{(\infty)}\right)_{b_j}^{(-\frac{1}{4})}=\frac{1}{2}\mathcal{E}_{b_j}^{(-\frac{1}{4})}\mathcal{D}_{\infty}^{\sigma_3}Q_{\infty}^{-1}
\begin{pmatrix}
\mathcal{G}\left(\vec{\mathcal{A}}(b_j)+\vec{d}\right) & -i\mathcal{G}(-\vec{\mathcal{A}}(b_j)+\vec{d}) \\
i\mathcal{G}\left(\vec{\mathcal{A}}(b_j)-\vec{d}\right) & \mathcal{G}\left(-\vec{\mathcal{A}}(b_j)-\vec{d}\right)
\end{pmatrix}b_j^{-\beta\sigma_3}\mathcal{D}_{b_j}^{-\sigma_3}.
\label{equ:Pinftyb_j(-1/4)}
\end{align}
Here, we remind the readers of the fact that $\mathcal{G}(\vec{\mathcal{A}}_{+}(p)+\vec{d})=\mathcal{G}(\vec{\mathcal{A}}_{-}(p)+\vec{d})$ for $p\in\mathcal{I}_{e}$ by combining \eqref{equ: vecA(a_j) and vecA(b_j)-a}--\eqref{equ: vecA(a_j) and vecA(b_j)-c} with \eqref{property:theta function}.
\item [(e)] Using \eqref{equ: asy mathcalD at 0} and \eqref{equ:Qasyto0+}, we have, as $z\to0$ from $\mathbb{C}^{+}$,
\begin{align}
P^{(\infty)}(z)=\mathcal{D}_{\infty}^{\sigma_3}Q_{\infty}^{-1}(*+\mathcal{O}(z))z^{\alpha\sigma_3}, 
\end{align}
where * is a constant matrix, as required.
\end{itemize}
As a consequence, the right-hand side of \eqref{equ:sol of Pinfty} indeed solves the RH problem for $P^{(\infty)}$.
\end{proof}

It is easy to check that $\det P^{(\infty)}(z)=1$, thus,
\begin{align}
{P^{(\infty)}(z)}^{-1}=\begin{pmatrix}
    P^{(\infty)}_{22}(z) & -P^{(\infty)}_{12}(z) \\
    -P^{(\infty)}_{21}(z) & P^{(\infty)}_{11}(z)
\end{pmatrix}.
\end{align}
Using \eqref{equ: Pinfty expansion at z=p}, we have, as $z\to p$ from $\mathbb{C}^+$,
\begin{align}\label{equ: Pinftyinverse expansion at z=p}
{P^{(\infty)}(z)}^{-1}=\left(P^{(\infty)}_{\rm inv}\right)_{p}^{(-\frac{1}{4})}(z-p)^{-\frac{1}{4}}+\mathcal{O}((z-p)^{\frac{1}{4}}),
\end{align}
where 
\begin{equation}
\left(P^{(\infty)}_{\rm inv}\right)_{p}^{(-\frac{1}{4})}=
\begin{pmatrix}
\left(P^{(\infty)}\right)_{p,22}^{(-\frac{1}{4})} & 
-\left(P^{(\infty)}\right)_{p,12}^{(-\frac{1}{4})} \\
-\left(P^{(\infty)}\right)_{p,21}^{(-\frac{1}{4})} &
\left(P^{(\infty)}\right)_{p,11}^{(-\frac{1}{4})}
\end{pmatrix}.
\end{equation}

\section{Local parametrices at the endpoints of $\Sigma$}
\label{sec: local parametrix}
As the convergence of $J_S$ to $I$ for large positive $s$ is not uniform near the endpoints of $\Sigma$, we need to construct a local parametrix near each $p\in\mathcal{I}_{e}$.
Let 
\begin{align}
U^{(p)}:=\left\{z: |z-p|<\varrho \right\}
\end{align}
be a small disk around $p$, where $\varrho>0$ independent of $s$ is chosen sufficiently small such that these disks do not intersect each other; see Figure \ref{fig:neigh of aj-bj} for an illustration. The local parametrix reads as follows.
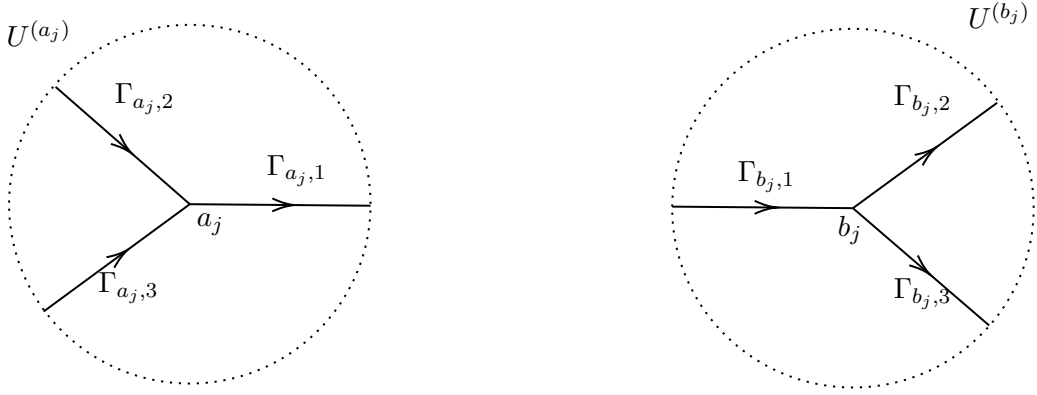
\begin{figure}[htbp]
\centering
\tikzset{every picture/.style={line width=0.75pt}} 
\begin{tikzpicture}[x=0.75pt,y=0.75pt,yscale=-1,xscale=1]
\draw    (84,80.6) -- (151,139.6) ;
\draw [shift={(122,114.07)}, rotate = 221.37] [color={rgb, 255:red, 0; green, 0; blue, 0 }  ][line width=0.75]    (10.93,-3.29) .. controls (6.95,-1.4) and (3.31,-0.3) .. (0,0) .. controls (3.31,0.3) and (6.95,1.4) .. (10.93,3.29)   ;
\draw    (151,139.6) -- (78,193.6) ;
\draw [shift={(120.13,162.44)}, rotate = 143.51] [color={rgb, 255:red, 0; green, 0; blue, 0 }  ][line width=0.75]    (10.93,-3.29) .. controls (6.95,-1.4) and (3.31,-0.3) .. (0,0) .. controls (3.31,0.3) and (6.95,1.4) .. (10.93,3.29)   ;
\draw    (241.21,140.37) -- (151,139.6) ;
\draw [shift={(203.11,140.05)}, rotate = 180.49] [color={rgb, 255:red, 0; green, 0; blue, 0 }  ][line width=0.75]    (10.93,-3.29) .. controls (6.95,-1.4) and (3.31,-0.3) .. (0,0) .. controls (3.31,0.3) and (6.95,1.4) .. (10.93,3.29)   ;
\draw    (482,141.6) -- (391.79,140.83) ;
\draw [shift={(443.89,141.27)}, rotate = 180.49] [color={rgb, 255:red, 0; green, 0; blue, 0 }  ][line width=0.75]    (10.93,-3.29) .. controls (6.95,-1.4) and (3.31,-0.3) .. (0,0) .. controls (3.31,0.3) and (6.95,1.4) .. (10.93,3.29)   ;
\draw    (554,88.6) -- (482,141.6) ;
\draw [shift={(523.64,110.95)}, rotate = 143.64] [color={rgb, 255:red, 0; green, 0; blue, 0 }  ][line width=0.75]    (10.93,-3.29) .. controls (6.95,-1.4) and (3.31,-0.3) .. (0,0) .. controls (3.31,0.3) and (6.95,1.4) .. (10.93,3.29)   ;
\draw    (482,141.6) -- (550,200.6) ;
\draw [shift={(520.53,175.03)}, rotate = 220.95] [color={rgb, 255:red, 0; green, 0; blue, 0 }  ][line width=0.75]    (10.93,-3.29) .. controls (6.95,-1.4) and (3.31,-0.3) .. (0,0) .. controls (3.31,0.3) and (6.95,1.4) .. (10.93,3.29)   ;
\draw  [dash pattern={on 0.84pt off 2.51pt}] (60.79,138.83) .. controls (61.21,89.01) and (101.95,48.96) .. (151.77,49.39) .. controls (201.59,49.81) and (241.64,90.55) .. (241.21,140.37) .. controls (240.79,190.19) and (200.05,230.24) .. (150.23,229.81) .. controls (100.41,229.39) and (60.36,188.65) .. (60.79,138.83) -- cycle ;
\draw  [dash pattern={on 0.84pt off 2.51pt}] (391.79,140.83) .. controls (392.21,91.01) and (432.95,50.96) .. (482.77,51.39) .. controls (532.59,51.81) and (572.64,92.55) .. (572.21,142.37) .. controls (571.79,192.19) and (531.05,232.24) .. (481.23,231.81) .. controls (431.41,231.39) and (391.36,190.65) .. (391.79,140.83) -- cycle ;
\draw (153,143) node [anchor=north west][inner sep=0.75pt]    {$a_{j}$};
\draw (473,144) node [anchor=north west][inner sep=0.75pt]    {$b_{j}$};
\draw (58,45.4) node [anchor=north west][inner sep=0.75pt]    {$U^{( a_{j})}$};
\draw (538,36.4) node [anchor=north west][inner sep=0.75pt]    {$U^{( b_{j})}$};
\draw (112,77.4) node [anchor=north west][inner sep=0.75pt]    {$\Gamma _{a_{j} ,2}$};
\draw (104,172.4) node [anchor=north west][inner sep=0.75pt]    {$\Gamma _{a_{j} ,3}$};
\draw (188,113.4) node [anchor=north west][inner sep=0.75pt]    {$\Gamma _{a_{j} ,1}$};
\draw (423,116.4) node [anchor=north west][inner sep=0.75pt]    {$\Gamma _{b_{j} ,1}$};
\draw (501,78.4) node [anchor=north west][inner sep=0.75pt]    {$\Gamma _{b_{j} ,2}$};
\draw (501,175.4) node [anchor=north west][inner sep=0.75pt]    {$\Gamma _{b_{j} ,3}$};
\end{tikzpicture}
\caption{The open disks centered at $a_j$ (left) and at $b_j$ (right).}
\label{fig:neigh of aj-bj}
\end{figure}

\paragraph{RH problem for $P^{(p)}$}
\begin{itemize}
\item [\rm (a)] $P^{(p)}(z)$ is holomorphic for $z\in U^{(p)} \setminus\Gamma_{S}$, where $\Gamma_S$ is defined in \eqref{def:Gamma_s}. 
\item [\rm (b)] For $z\in U^{(p)} \cap \Gamma_{S}$, we have 
\begin{align}\label{eq:localjump}
P^{(p)}_{+}(z)=P^{(p)}_{-}(z)J_{S}(z),
\end{align}
where $J_{S}(z)$ is given in \eqref{equ:jump of S}.
\item [\rm (c)] As $z\to p$, we have 
$P^{(p)}(z)=\mathcal{O}(\log (z-p))$.
\item [\rm (d)] 
As $s\to+\infty$, we have $P^{(p)}(z)\cdot {P^{(\infty)}(z)}^{-1}\to I$ 
uniformly for $z\in\partial U^{(p)}$.
\end{itemize}
Following \cites{Dei-Its-Zhou-Ann1997, KMAV-AdV-2004}, one can solve the above RH problem by using the Bessel parametrix  $\Phi_{\rm Be}$ introduced in Appendix \ref{appendix:Bessel}, as presented in what follows.


\subsection{Local parametrix at $a_j$}
Let 
\begin{align}\label{def: local map ua_j}
u_{a_j}(z):=-\left(g(z)-\frac{\Omega_j}{2}\right)^{2},
\end{align}
where $\Omega_j$ is given in \eqref{equ:Omega_j}. As $z\to a_j$, it follows from \eqref{equ:g-func} that
\begin{align}\label{equ: u_aj expansion at a_j}
u_{a_j}(z)={c}_{a_j}(z-a_j)\left(1+\mathcal{O}\left(z-a_j\right)\right), \quad 
c_{a_j}=-\frac{4\mathrm{p}(a_j)^2}{\prod_{\substack{k=0 \\ k\neq j}}^{n}(a_j-a_k)\prod_{k=0}^{n}(a_j-b_k)}>0.
\end{align}
Thus, $u_{a_j}$ is a conformal mapping from $U^{(a_j)}$ to a neighborhood of the origin. By choosing the principal branch, it is readily seen that
which implies that 
\begin{align}
 {u_{a_j}(z)}^{\frac{1}{2}}=\mp i\left(g(z)-\frac{\Omega_j}{2}\right), \qquad z\in\mathbb{C}^{\pm}\cap U^{(a_j)}.
\label{equ:relation uaj and g-omegaj/2-1} 
\end{align}

We now define
\begin{align}
 P^{(a_j)}(z):=E^{(a_j)}(z)\Phi_{\rm Be}(s^2u_{a_j}(z))H(z)e^{\left(isg(z)-\sign{a_j}\frac{(\alpha-\beta)\pi i}{2}\right)\sigma_3}, \label{def:sol of local near aj}
\end{align}
where 
\begin{equation}\label{def:analytic factor Ea_j}
    E^{(a_j)}(z):=P^{(\infty)}(z)e^{-\left(is\frac{\Omega_j}{2}-\sign{a_j}\frac{\alpha-\beta}{2}\pi i\right)\sigma_3}{H(z)}^{-1}M^{-1}\left(\pi s{u_{a_j}(z)}^{\frac{1}{2}}\right)^{\frac{1}{2}\sigma_3}, 
\end{equation}
and
\begin{align}\label{equ:constant matrix H and M}
H(z):=\left\{
\begin{aligned}
&I, &\im z>0, \\
&\begin{pmatrix}
    0 & 1\\
    -1 & 0
\end{pmatrix}, & \im z<0,
\end{aligned}
\right.
\qquad
M:=\frac{1}{\sqrt{2}}
\begin{pmatrix}
    1 & i\\
    i & 1
\end{pmatrix}.
\end{align}
\begin{proposition}\label{prop:local Pa_j}
The local parametrix defined in \eqref{def:sol of local near aj} solves the RH problem for $P^{(a_j)}$ for $j=0,1,\dots,n$.
\end{proposition}
\begin{proof}
We first show the prefactor $E^{(a_j)}$ is analytic in
$U^{(a_j)}$. By \eqref{def:analytic factor Ea_j}, it is readily seen that the only possible jump for $E^{(a_j)}$ is on $(a_j-\varrho,a_j+\varrho)$; see Figure \ref{fig:neigh of aj-bj}. For $z\in(a_{j},a_{j}+\varrho)$, recall the  jump matrix of $P^{(\infty)}$ given in \eqref{equ:jump of Pinfty}, we have 
\begin{multline}
{E^{(a_j)}_{-}(z)}^{-1}{E^{(a_j)}_{+}}(z)=
\left(\pi su_{a_j}(z)^{\frac{1}{2}}\right)^{-\frac{1}{2}\sigma_3}M
\begin{pmatrix}
0 & 1\\
-1 & 0
\end{pmatrix}e^{(is\Omega_j-\sign{a_j}(\alpha-\beta)\pi i)\frac{\sigma_3}{2}}\times 
\\
{P^{(\infty)}_{-}(z)}^{-1}{P^{(\infty)}_{+}}(z)
e^{-(is\Omega_j-\sign{a_j}(\alpha-\beta)\pi i)\frac{\sigma_3}{2}}
\begin{pmatrix}
    1 & 0 \\
    0 & 1
\end{pmatrix}
M^{-1}
\left(\pi su_{a_j}(z)^{\frac{1}{2}}\right)^{\frac{1}{2}\sigma_3}=I.
\end{multline}
For $z\in(a_j-\varrho,a_j)$, since ${P^{(\infty)}(z)}^{-1}P^{(\infty)}(z)=I$, it follows that
\begin{align}
&{E^{(a_j)}_{-}(z)}^{-1}{E^{(a_j)}_{+}}(z)=
\left(\pi s{u_{a_j,-}(z)}^{\frac{1}{2}}\right)^{-\frac{1}{2}\sigma_3}
\begin{pmatrix}
    -i & 0 \\
    0 & i
\end{pmatrix}
\left(\pi s{u_{a_j,+}(z)}^{\frac{1}{2}}\right)^{\frac{1}{2}\sigma_3}\nonumber\\
&=\begin{pmatrix}
    -iu_{a_j,-}(z)^{-\frac{1}{4}}u_{a_j,+}(z)^{\frac{1}{4}} & 0 \\
    0 & iu_{a_j,-}(z)^{\frac{1}{4}}u_{a_j,+}(z)^{-\frac{1}{4}}
\end{pmatrix}=I.
\end{align}
Also note the expansions of $P^{(\infty)}$ and $u_{a_j}$ at $a_j$ given in \eqref{equ: Pinfty expansion at z=p} and \eqref{equ: u_aj expansion at a_j}, we conclude that $a_j$ is a removable singular point and $E^{(a_j)}$ is analytic in $U^{(a_j)}$. The analyticity
of $E^{(a_j)}$, together with  item (b) of Proposition \ref{prop:g-func} and \eqref{jump of PhiBe}, gives us the jump relation \eqref{eq:localjump}.

By considering \eqref{asy: PhiBe at 0}, the asymptotic behavior of $P^{(a_j)}$ near $a_j$ follows directly from \eqref{def:sol of local near aj} and \eqref{equ: u_aj expansion at a_j}. For the matching condition, we see from \eqref{equ:asy of PhiBe at infty} and \eqref{equ:relation uaj and g-omegaj/2-1} that, as $s\to +\infty$,
\begin{multline}\label{equ:PajPinfty-1 est}
P^{(a_j)}(z){P^{(\infty)}(z)}^{-1}=
I\pm\\ 
\frac{P^{(\infty)}(z)e^{\left(-is\frac{\Omega_j}{2}+\sign{a_j}\frac{(\alpha-\beta)\pi i}{2}\right)\sigma_3}\Phi_{\rm Be,1}
e^{\left(is\frac{\Omega_j}{2}-\sign{a_j}\frac{(\alpha-\beta)\pi i} {2}\right)\sigma_3}{P^{(\infty)}(z)}^{-1}}{su_{a_j}(z)^{\frac{1}{2}}}+\mathcal{O}(s^{-2}),
\end{multline}
uniformly for $z\in\partial U^{(a_j)}$, as required. Here, $\Phi_{\rm Be,1}$ is given in \eqref{equ: Bessel coe M,be1,be2} and the $\pm$ signs correspond to $z\in\mathbb{C}^{\pm}\cap U^{(a_j)}$, respectively.  
\end{proof}

\subsection{Local parametrix at $b_j$}
Let 
\begin{align}\label{def: local map ub_j}
u_{b_j}(z):=-\left(g(z)-\frac{\Omega_j}{2}\right)^{2},
\end{align}
where $\Omega_j$ is given in \eqref{equ:Omega_j} and define
\begin{align}
  P^{(b_j)}(z):=E^{(b_j)}(z)\sigma_3\Phi_{\rm Be}(s^2u_{b_j}(z))\sigma_3H(z)e^{\left(isg(z)-\sign{b_j}\frac{(\alpha-\beta)\pi i}{2}\right)\sigma_3}, \label{def:sol of local near bj}
  \end{align}
where
\begin{align}
E^{(b_j)}(z):=P^{(\infty)}(z)e^{-\left(is\frac{\Omega_j}{2}-\sign{b_j}\frac{\alpha-\beta}{2}\pi i\right)\sigma_3}{H(z)}^{-1}M\left(\pi su_{b_j}(z)^{\frac{1}{2}}\right)^{\frac{1}{2}\sigma_3}, 
 \label{def:analytic factor Eb_j}
\end{align}
$H$ and $M$ are given in \eqref{equ:constant matrix H and M}. 
As $z\to b_j$, one has
\begin{align}\label{equ: u_bj expansion at b_j}
u_{b_j}(z)=c_{b_j}(z-b_j)\left(1+\mathcal{O}\left(z-b_j\right)\right), \quad 
c_{b_j}=-\frac{4\mathrm{p}(b_j)^2}{\prod_{\substack{k=0 \\ k\neq j}}^{n}(b_j-b_k)\prod_{k=0}^{n}(b_j-a_k)}<0,
\end{align}
we then have the following analogue of Proposition \ref{prop:local Pa_j}.
\begin{proposition}\label{prop:local Pb_j}
    The local parametrix defined in 
    \eqref{def:sol of local near bj} solves the RH problem for $P^{(b_j)}$ for $j=0,1,\dots,n$.
\end{proposition}
For later use, we note that, as $s \to +\infty$,
\begin{multline}\label{equ:PbjPinfty-1 est}
 P^{(b_j)}(z){P^{(\infty)}(z)}^{-1}=\\
 I\pm
 \frac{P^{(\infty)}(z)e^{\left(-is\frac{\Omega_j}{2}+\sign{b_j}\frac{(\alpha-\beta)\pi i}{2}\right)\sigma_3}\sigma_3\Phi_{\rm Be,1}\sigma_3
 e^{\left(is\frac{\Omega_j}{2}-\sign{b_j}\frac{(\alpha-\beta)\pi i} {2}\right)\sigma_3}{P^{(\infty)}(z)}^{-1}}{su_{b_j}(z)^{\frac{1}{2}}}
 +\mathcal{O}(s^{-2}),
\end{multline}
uniformly for $z\in\partial U^{(b_j)}$, where the $\pm$ signs correspond to $z\in\mathbb{C}^{\pm}\cap U^{(b_j)}$, respectively.  

\section{The small norm RH problem}\label{sec: small norm RH problem}
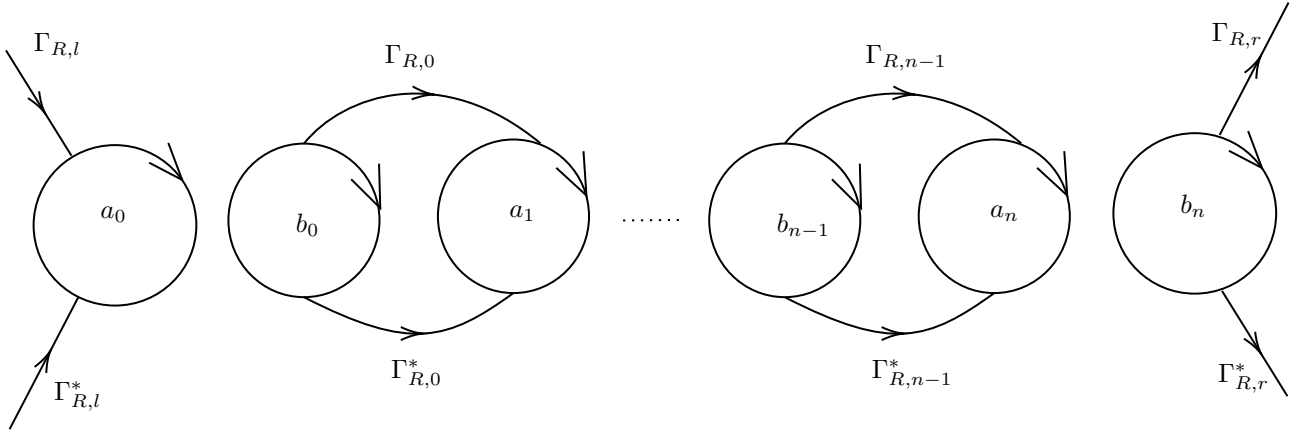
\begin{figure}[htbp]
\centering
\tikzset{every picture/.style={line width=0.75pt}} 
\begin{tikzpicture}[x=0.75pt,y=0.75pt,yscale=-1,xscale=1]
\draw   (23.76,141.46) .. controls (23.76,119.17) and (41.95,101.1) .. (64.38,101.1) .. controls (86.81,101.1) and (105,119.17) .. (105,141.46) .. controls (105,163.74) and (86.81,181.81) .. (64.38,181.81) .. controls (41.95,181.81) and (23.76,163.74) .. (23.76,141.46) -- cycle ;
\draw   (121,138.96) .. controls (121,117.63) and (137.88,100.34) .. (158.71,100.34) .. controls (179.54,100.34) and (196.42,117.63) .. (196.42,138.96) .. controls (196.42,160.29) and (179.54,177.58) .. (158.71,177.58) .. controls (137.88,177.58) and (121,160.29) .. (121,138.96) -- cycle ;
\draw   (225.58,136.9) .. controls (225.58,115.57) and (242.46,98.28) .. (263.29,98.28) .. controls (284.12,98.28) and (301,115.57) .. (301,136.9) .. controls (301,158.23) and (284.12,175.52) .. (263.29,175.52) .. controls (242.46,175.52) and (225.58,158.23) .. (225.58,136.9) -- cycle ;
\draw    (158.71,100.34) .. controls (178.32,75.63) and (227.48,59.01) .. (276.75,100.2) ;
\draw [shift={(223.6,75.7)}, rotate = 184.46] [color={rgb, 255:red, 0; green, 0; blue, 0 }  ][line width=0.75]    (10.93,-3.29) .. controls (6.95,-1.4) and (3.31,-0.3) .. (0,0) .. controls (3.31,0.3) and (6.95,1.4) .. (10.93,3.29)   ;
\draw    (158.71,177.58) .. controls (212.51,205.39) and (230.61,199.21) .. (263.29,175.52) ;
\draw [shift={(218,195.91)}, rotate = 180.06] [color={rgb, 255:red, 0; green, 0; blue, 0 }  ][line width=0.75]    (10.93,-3.29) .. controls (6.95,-1.4) and (3.31,-0.3) .. (0,0) .. controls (3.31,0.3) and (6.95,1.4) .. (10.93,3.29)   ;
\draw   (91,100) -- (97.93,118.62) -- (81.41,109.93) ;
\draw   (196.18,110.44) -- (196.8,132.79) -- (182.28,118.29) ;
\draw   (298.24,104.15) -- (299.72,126.46) -- (284.66,112.52) ;
\draw    (10,53.47) -- (42.63,106.39) ;
\draw [shift={(29.46,85.04)}, rotate = 238.35] [color={rgb, 255:red, 0; green, 0; blue, 0 }  ][line width=0.75]    (10.93,-3.29) .. controls (6.95,-1.4) and (3.31,-0.3) .. (0,0) .. controls (3.31,0.3) and (6.95,1.4) .. (10.93,3.29)   ;
\draw    (46.18,177.4) -- (11.78,244) ;
\draw [shift={(32.19,204.48)}, rotate = 117.32] [color={rgb, 255:red, 0; green, 0; blue, 0 }  ][line width=0.75]    (10.93,-3.29) .. controls (6.95,-1.4) and (3.31,-0.3) .. (0,0) .. controls (3.31,0.3) and (6.95,1.4) .. (10.93,3.29)   ;
\draw  [dash pattern={on 0.84pt off 2.51pt}]  (318,139) -- (349.02,138.72) ;
\draw   (361,138.96) .. controls (361,117.63) and (377.88,100.34) .. (398.71,100.34) .. controls (419.54,100.34) and (436.42,117.63) .. (436.42,138.96) .. controls (436.42,160.29) and (419.54,177.58) .. (398.71,177.58) .. controls (377.88,177.58) and (361,160.29) .. (361,138.96) -- cycle ;
\draw   (465.58,136.9) .. controls (465.58,115.57) and (482.46,98.28) .. (503.29,98.28) .. controls (524.12,98.28) and (541,115.57) .. (541,136.9) .. controls (541,158.23) and (524.12,175.52) .. (503.29,175.52) .. controls (482.46,175.52) and (465.58,158.23) .. (465.58,136.9) -- cycle ;
\draw    (398.71,100.34) .. controls (418.32,75.63) and (467.48,59.01) .. (516.75,100.2) ;
\draw [shift={(463.6,75.7)}, rotate = 184.46] [color={rgb, 255:red, 0; green, 0; blue, 0 }  ][line width=0.75]    (10.93,-3.29) .. controls (6.95,-1.4) and (3.31,-0.3) .. (0,0) .. controls (3.31,0.3) and (6.95,1.4) .. (10.93,3.29)   ;
\draw    (398.71,177.58) .. controls (452.51,205.39) and (470.61,199.21) .. (503.29,175.52) ;
\draw [shift={(458,195.91)}, rotate = 180.06] [color={rgb, 255:red, 0; green, 0; blue, 0 }  ][line width=0.75]    (10.93,-3.29) .. controls (6.95,-1.4) and (3.31,-0.3) .. (0,0) .. controls (3.31,0.3) and (6.95,1.4) .. (10.93,3.29)   ;
\draw   (436.18,110.44) -- (436.8,132.79) -- (422.28,118.29) ;
\draw   (538.24,104.15) -- (539.72,126.46) -- (524.66,112.52) ;
\draw   (562.76,135.46) .. controls (562.76,113.17) and (580.95,95.1) .. (603.38,95.1) .. controls (625.81,95.1) and (644,113.17) .. (644,135.46) .. controls (644,157.74) and (625.81,175.81) .. (603.38,175.81) .. controls (580.95,175.81) and (562.76,157.74) .. (562.76,135.46) -- cycle ;
\draw   (630,94) -- (636.93,112.62) -- (620.41,103.93) ;
\draw    (617,174.47) -- (649.63,227.39) ;
\draw [shift={(636.46,206.04)}, rotate = 238.35] [color={rgb, 255:red, 0; green, 0; blue, 0 }  ][line width=0.75]    (10.93,-3.29) .. controls (6.95,-1.4) and (3.31,-0.3) .. (0,0) .. controls (3.31,0.3) and (6.95,1.4) .. (10.93,3.29)   ;
\draw    (650.18,29.4) -- (615.78,96) ;
\draw [shift={(636.19,56.48)}, rotate = 117.32] [color={rgb, 255:red, 0; green, 0; blue, 0 }  ][line width=0.75]    (10.93,-3.29) .. controls (6.95,-1.4) and (3.31,-0.3) .. (0,0) .. controls (3.31,0.3) and (6.95,1.4) .. (10.93,3.29)   ;
\draw (55.71,130.55) node [anchor=north west][inner sep=0.75pt]  [font=\small]  {$a_{0}$};
\draw (153.22,133.94) node [anchor=north west][inner sep=0.75pt]  [font=\small]  {$b_{0}$};
\draw (259.82,129.82) node [anchor=north west][inner sep=0.75pt]  [font=\small]  {$a_{1}$};
\draw (22.36,43.22) node [anchor=north west][inner sep=0.75pt]  [font=\small]  {$\Gamma _{R,l}$};
\draw (32.57,218.58) node [anchor=north west][inner sep=0.75pt]  [font=\small]  {$\Gamma _{R,l}^{*}$};
\draw (197.58,49.28) node [anchor=north west][inner sep=0.75pt]  [font=\small]  {$\Gamma _{R,0}$};
\draw (200.67,207.36) node [anchor=north west][inner sep=0.75pt]  [font=\small]  {$\Gamma _{R,0}^{*}$};
\draw (393.22,133.94) node [anchor=north west][inner sep=0.75pt]  [font=\small]  {$b_{n-1}$};
\draw (499.82,129.82) node [anchor=north west][inner sep=0.75pt]  [font=\small]  {$a_{n}$};
\draw (437.58,49.28) node [anchor=north west][inner sep=0.75pt]  [font=\small]  {$\Gamma _{R,n-1}$};
\draw (440.67,207.36) node [anchor=north west][inner sep=0.75pt]  [font=\small]  {$\Gamma _{R,n-1}^{*}$};
\draw (594.71,124.55) node [anchor=north west][inner sep=0.75pt]  [font=\small]  {$b_{n}$};
\draw (610.36,38.22) node [anchor=north west][inner sep=0.75pt]  [font=\small]  {$\Gamma _{R,r}$};
\draw (613.57,209.58) node [anchor=north west][inner sep=0.75pt]  [font=\small]  {$\Gamma _{R,r}^{*}$};
\end{tikzpicture}
\caption{The jump contours of the RH problem for $R$.}
\label{fig:RHP R}
\end{figure}
Define
\begin{align}\label{def: R}
R(z):=\left\{
\begin{aligned}
& S(z){P^{(p)}(z)}^{-1}, && z\in U^{(p)}, \ p\in\mathcal{I}_{e},\\
& S(z){P^{(\infty)}(z)}^{-1}, && z\in\mathbb{C}\setminus 
\cup_{p\in\mathcal{I}_{e}} U^{(p)}.
\end{aligned}
\right.
\end{align}
In view of RH problems for $S$, $P^{(p)}$ and $P^{(\infty)}$, it is then readily seen that $R$ solves the following RH problem.
\paragraph{RH problem for $R$}
\begin{itemize}
\item [\rm (a)] $R(z)$ is holomorphic for $z\in\mathbb{C}\setminus\Gamma_{R}$, where 
\begin{align}
\Gamma_{R}:=\left(\gamma_{R}\cup\gamma_{R}^{*}\cup\bigcup_{p\in\mathcal{I}_{e}}\partial U^{(p)}\right)\setminus\bigcup_{p\in\mathcal{I}_{e}}U^{(p)},
\end{align}
with 
\begin{align}
&\gamma_{R}:=\Gamma_{R,l}\cup\Gamma_{R,r}\cup
\bigcup_{\substack{j=0}}^{n-1}\Gamma_{R,j}, \qquad \gamma_{R}^*:=\Gamma_{R,l}^{*}\cup\Gamma_{R,r}^{*}\cup
\bigcup_{\substack{j=0}}^{n-1}\Gamma_{R,j}^{*};
\end{align}
see Figure \ref{fig:RHP R} for an illustration and the orientation.
\item [\rm (b)] For $z\in\Gamma_{R}$, we have 
\begin{align}
R_{+}(z)=R_{-}(z)J_{R}(z), 
\end{align}
where
\begin{align}\label{equ:J_{R}}
J_{R}(z)=\left\{
\begin{aligned}
& P^{(\infty)}(z)J_{S}(z){P^{(\infty)}(z)}^{-1}, &&z\in \Gamma_{R}\setminus \cup_{p\in\mathcal{I}_{e}} U^{(p)}, \\
& P^{(p)}(z){P^{(\infty)}(z)}^{-1}, &&z\in\cup_{p\in\mathcal{I}_{e}} \partial U^{(p)}.
\end{aligned}
\right.
\end{align}
\item [\rm (c)]  As $z\to\infty$, we have
\begin{align}\label{eq:R(z) asym as z to infty}
R(z)=I+\frac{R_1(s)}{z}+\mathcal{O}(z^{-2}),
\end{align}
where $R_1$ is independent of $z$.
\end{itemize}
For large positive $s$, we have the following estimate for the jump matrix $J_{R}$ 
defined in \eqref{equ:J_{R}}.
\begin{proposition}\label{prop:JRest}
As $s\to+\infty$, we have
\begin{align}\label{equ:J_{R} asy as s to infty}
J_{R}(z)=\left\{
\begin{aligned}
    & I+\mathcal{O}\left(e^{-c_{0}s|z|}\right), && z\in\Gamma_{R}\setminus\cup_{p\in\mathcal{I}_{e}}\partial U^{(p)},\\
& I+\frac{J_{R}^{(1)}(z)}{s}+\mathcal{O}(s^{-2}), && z\in\cup_{p\in\mathcal{I}_{e}}\partial U^{(p)},
\end{aligned}
\right.
\end{align}
where $c_0$ is a small positive constant and 
\begin{align}\label{equ:J_R in neighbor}
J_{R}^{(1)}(z)=\left\{
\begin{aligned}
&\pm \frac{P^{(\infty)}(z)e^{\left(-is\frac{\Omega_j}{2}+\sign{a_j}\frac{(\alpha-\beta)\pi i}{2}\right)\sigma_3}\Phi_{\rm Be,1}
e^{\left(is\frac{\Omega_j}{2}-\sign{a_j}\frac{(\alpha-\beta)\pi i} {2}\right)\sigma_3}{P^{(\infty)}(z)}^{-1}}{u_{a_j}(z)^{\frac{1}{2}}}, &z\in\partial U^{(a_j)},\\
&\pm \frac{P^{(\infty)}(z)e^{\left(-is\frac{\Omega_j}{2}+\sign{b_j}\frac{(\alpha-\beta)\pi i}{2}\right)\sigma_3}\sigma_3\Phi_{\rm Be,1}\sigma_3
e^{\left(is\frac{\Omega_j}{2}-\sign{b_j}\frac{(\alpha-\beta)\pi i} {2}\right)\sigma_3}{P^{(\infty)}(z)}^{-1}}{u_{b_j}(z)^{\frac{1}{2}}}, &z\in\partial U^{(b_j)}.
\end{aligned}
\right.
\end{align}
\end{proposition}
\begin{proof}
For $z\in\Gamma_{R}\setminus\cup_{p\in\mathcal{I}_{e}}U^{(p)}$, one has from 
\eqref{equ:J_{R}} that
\begin{align}
J_{R}(z)-I=P^{(\infty)}(z)(J_{S}(z)-I){P^{(\infty)}(z)}^{-1}.
\end{align}
On account of \eqref{def:Q(z)} and \eqref{equ:sol of Pinfty}, we have that $P^{(\infty)}(z)$ is bounded as $s\to+\infty$ uniformly for $z$ bounded away from the $\varrho$-neighborhood of $p\in\mathcal{I}_{e}$. This, together with \eqref{equ:JS-I}, gives us the first estimate in \eqref{equ:J_{R} asy as s to infty}. For $z\in\partial U^{(p)}$, by using \eqref{equ:PajPinfty-1 est} and \eqref{equ:PbjPinfty-1 est}, we arrive at the desired estimate.
\end{proof}

Combining Proposition \ref{prop:JRest} and the standard small norm theory \cites{DZ-Annals-1993}, we conclude that there exists a unique solution to the RH problem for $R$ for sufficiently large $s$. In addition, 
\begin{align}\label{equ: R(z) expansion as s to infty}
R(z)=I+\frac{R^{(1)}(z)}{s}+\mathcal{O}(s^{-2}), \quad s\to+\infty,
\end{align}
uniformly for $z\in\mathbb{C}\setminus{\Gamma_{R}}$. Together with \eqref{equ:J_{R} asy as s to infty} and the integral equation
\begin{align}\label{equ: R(z) BC expression}
R(z)=I+\frac{1}{2\pi i}\int_{\Gamma_{R}}\frac{R_{-}(\xi)(J_{R}(\xi)-I)}{\xi-z}\dif\xi, \quad 
z\in\mathbb{C}\setminus\Gamma_{R},
\end{align}
it is deduced that
\begin{align}\label{eq:R^(1)(z) except the disk edge}
R^{(1)}(z)=\frac{1}{2\pi i}\int_{\cup_{p\in\mathcal{I}_{e}}\partial U^{(p)}}\frac{J_{R}^{(1)}(\xi)}{\xi-z}\dif\xi, \quad z\in\mathbb{C}\setminus\cup_{p\in\mathcal{I}_{e}}\partial U^{(p)}.
\end{align}
It's readily seen from \eqref{equ:J_R in neighbor} that $J_R^{(1)}$ 
admits an analytic continuation from $\cup_{p\in\mathcal{I}_{e}}\partial U^{(p)}$ to  $\cup_{p\in\mathcal{I}_{e}} \overline{{U}^{(p)}}\setminus{\left\{p\right\}}$ and
\begin{align}
    J_R^{(1)}(z)=(J_R^{(1)})^{(-1)}_{p}(z-p)^{-1}+\mathcal{O}(1) \qquad z \to p.
\end{align}
For $z\in\mathbb{C}\setminus\cup_{p\in\mathcal{I}_{e}}U^{(p)}$, by the residue theorem, we have 
\begin{align}
    R^{(1)}(z)=\sum_{p\in\mathcal{I}_{e}}\frac{1}{z-p}(J_R^{(1)})^{(-1)}_{p}, 
\end{align}
and moreover
\begin{align}\label{eq: R^{(1)}(z) as z to infty}
    R^{(1)}(z)=\frac{1}{z}\sum_{p\in\mathcal{I}_{e}}(J_R^{(1)})^{(-1)}_{p}+\mathcal{O}(z^{-2}), \quad {\rm as} \ z\to \infty.
\end{align}
It's then inferred from \eqref{eq:R(z) asym as z to infty} and \eqref{eq: R^{(1)}(z) as z to infty} that
\begin{align}\label{equ: asy for R_1(s) and R(1)_1}
    R_1(s)=\frac{R_1^{(1)}}{s}+\mathcal{O}(s^{-2}), \quad \quad R_1^{(1)}=\sum_{p\in\mathcal{I}_{e}}(J_R^{(1)})^{(-1)}_{p}.
\end{align}
Using \eqref{equ: Pinfty expansion at z=p}, \eqref{equ: Pinftyinverse expansion at z=p}, \eqref{equ: u_aj expansion at a_j}, \eqref{equ: u_bj expansion at b_j} and \eqref{equ:J_R in neighbor}, we have for $j=0,1,\dots,n$,
\begin{align}
&(J_{R}^{(1)})_{a_j}^{(-1)}=c_{a_j}^{-\frac{1}{2}}\left(P^{(\infty)}\right)_{a_j}^{(-\frac{1}{4})}e^{\left(-is\frac{\Omega_j}{2}+\sign{a_j}\frac{(\alpha-\beta)\pi i}{2}\right)\sigma_3}\Phi_{\rm Be,1}
e^{\left(is\frac{\Omega_j}{2}-\sign{a_j}\frac{(\alpha-\beta)\pi i} {2}\right)\sigma_3}\left(P^{(\infty)}_{\rm inv}\right)_{a_j}^{(-\frac{1}{4})}, \\
&(J_{R}^{(1)})_{b_j}^{(-1)}=c_{b_j}^{-\frac{1}{2}}\left(P^{(\infty)}\right)_{b_j}^{(-\frac{1}{4})}e^{\left(-is\frac{\Omega_j}{2}+\sign{b_j}\frac{(\alpha-\beta)\pi i}{2}\right)\sigma_3}\sigma_3\Phi_{\rm Be,1}\sigma_3
e^{\left(is\frac{\Omega_j}{2}-\sign{b_j}\frac{(\alpha-\beta)\pi i} {2}\right)\sigma_3}\left(P^{(\infty)}_{\rm inv}\right)_{b_j}^{(-\frac{1}{4})}.
\end{align}
In particular, we can write the $(1,1)$ and $(2,2)$-entries of $(J_{R}^{(1)})_{p}^{(-1)}$ in a unified way by using \eqref{equ:Pinftya_j(-1/4)} and \eqref{equ:Pinftyb_j(-1/4)}, i.e.,
\begin{align}\label{equ: JR(1)p(-1) (11) and (22)-a}
\left[(J_{R}^{(1)})_{p}^{(-1)}\right]_{11}=-\left[(J_{R}^{(1)})_{p}^{(-1)}\right]_{22}=\frac{1}{16}c_{p}^{-\frac{1}{2}}
\left(\mathcal{E}_{p}^{(-\frac{1}{4})}\right)^{2}
\frac{\mathcal{G}\left(\vec{\mathcal{A}}(p)+\vec{d}\right)\mathcal{G}\left(-\vec{\mathcal{A}}(p)-\vec{d}\right)}{\mathcal{G}\left(\vec{\mathcal{A}}(\infty)+\vec{d}\right)\mathcal{G}\left(-\vec{\mathcal{A}}(\infty)-\vec{d}\right)},
\end{align}
where one also needs to use the relations \eqref{equ:mathcalD_p}, \eqref{def:Vj}, \eqref{equ:jump of mathcalG-1} and \eqref{equ:jump of mathcalG-2} for further simplifications.


We are now ready to prove our main results.

\section{Proofs of the main results} \label{sec: proof of theorems}
\subsection{Proof of Theorem \ref{thm: general case}}
By tracing back the transformations $X \to T \to S \to R$ given by \eqref{def: T}, \eqref{def: S} and \eqref{def: R}, it follows that, for sufficiently large $z$,
\begin{align}
X(z)=e^{is\ell\sigma_3}R(z)P^{(\infty)}(z)e^{-isg(z)\sigma_3}.
\end{align}
Taking $z\to\infty$ and using \eqref{eq:gAsy}, \eqref{equ:asy Pinfty at z=infty} and \eqref{eq:R(z) asym as z to infty}, we obtain
\begin{align}
X(z)=\left(I+e^{is\ell\sigma_3}\frac{P^{(\infty)}_{1}+R_1+is\gamma_0\sigma_3}{z}e^{-is\ell\sigma_3}+\mathcal{O}(z^{-2})\right)z^{-\beta\sigma_3}e^{-isz\sigma_3}.
\end{align}
Comparing the above formula with \eqref{equ:X asym at inf}, it follows from \eqref{equ: asy for R_1(s) and R(1)_1} that, as $s\to+\infty$,
\begin{align}
{(X_{1}(s))}_{11}&=is\gamma_0+{(P^{(\infty)}_{1})}_{11}+\frac{{(R^{(1)}_{1})}_{11}}{s}+\mathcal{O}(s^{-2}), 
\\
{(X_{1}(s))}_{22}&=-is\gamma_0+{(P^{(\infty)}_{1})}_{22}+\frac{{(R^{(1)}_{1})}_{22}}{s}+\mathcal{O}(s^{-2}).
\end{align}
This, together with Proposition \ref{prop: differential identity}, implies 
\begin{multline}
    \label{equ: differential identity again}
     \partial_s\log\mathcal{F}(s\Sigma)=-2s\gamma_0+i\left({(P^{(\infty)}_{1})}_{11}-{(P^{(\infty)}_{1})}_{22}\right)+\frac{i}{s}\left({(R^{(1)}_{1})}_{11}-{(R^{(1)}_{1})}_{22}\right)
     \\
     -\frac{\alpha^2-\beta^2}{s}+\mathcal{O}(s^{-2}).
\end{multline}

We next rewrite $i\left({(P^{(\infty)}_{1})}_{11}-{(P^{(\infty)}_{1})}_{22}\right)$ as a derivative of a function. Recall the holomorphic one-form $\vec{\omega}=(\omega_1, \omega_2, \dots,\omega_{n})$ given in \eqref{equ:holo differential}, one has 
\begin{align}
\omega_j=\sum_{m=1}^{n}\frac{z^{m-1}\dif z}{\sqrt{\mathcal{R}(z)}}(\tilde{\mathbb{A}}^{-1})_{mj},\qquad j=1,\dots,n.
\end{align}
Since $1/\sqrt{\mathcal{R}(z)}=z^{-n-1}+\mathcal{O}(z^{-n-2})$ as $z\to\infty$, it follows 
\begin{align}\label{eq:omegajz}
\omega_j=\frac{(\tilde{\mathbb{A}}^{-1})_{nj}}{z^2}\left(1+\mathcal{O}(z^{-1})\right)\dif z, \qquad 
 \ z\to\infty. 
\end{align}
\begin{lemma}
We have 
\begin{align}\label{equ: tildeA = omega_j-0}
-4\pi i\cdot\left(\tilde{\mathbb{A}}^{-1}\right)_{nj}=\Omega_j=2\pi\frac{\dif V_j(s)}{\dif s}, \quad j=1,\dots,n,
\end{align}
where $V_{j}$ is given in \eqref{def:Vj}.
\end{lemma}
\begin{proof}
It's readily seen from \eqref{equ: asy of p/mathcalR at infty} and \eqref{eq:gAsy} that $\dif g(z)=\mathrm{p}(z)\dif z/\sqrt{\mathcal{R}(z)}$ is a meromorphic 
Abelian differential of the second kind on $\mathcal{W}$. In terms of 
a local holomorphic coordinate $u=z^{-1}$, it follows from \eqref{equ: asy of p/mathcalR at infty} that 
\begin{align}\label{equ: asy of dg(u) as u to 0}
\dif g(u)=\mp(1+\gamma_0u^2+o(u^2))\frac{\dif u}{u^2}, \quad {\rm as} \ u\to 0,
\end{align}
on the first (respectively, second) sheet  of $\mathcal{W}$. Similarly, as $u\to0$, by \eqref{eq:omegajz},
\begin{align}
\omega_j=\mp\left(\tilde{\mathbb{A}}^{-1}\right)_{nj}(1+\mathcal{O}(u))\dif u
\end{align}
on the first (respectively, second) sheet  of $\mathcal{W}$. Applying the 
Riemann's bilinear identity (cf. \cite[Page 64, Equation (3.0.2)]{FK-Riemann-Surface}) to $\dif g(z)$ and $\omega_j$, it follows that
\begin{align}\label{equ: Riemann bilinear relation}
\sum_{k=1}^{n}\left[\left(\int_{A_k}\omega_j\right)\left(\int_{B_k}\dif g\right)-\left(\int_{B_k}\omega_j\right)\left(\int_{A_k}\dif g\right)\right]=2\pi i
\res_{u=0}(f\cdot\dif g),
\end{align}
where $f$ is a function such that $\dif f=\omega_k$ near $u=0$, and 
\begin{align}\label{equ: asy of f(u) at u=0}
f(u)=f(0)\mp\left(\tilde{\mathbb{A}}^{-1}\right)_{nj}u+\mathcal{O}(u^2), \quad {\rm as} \ u\to 0.
\end{align}
Together with \eqref{equ: asy of f(u) at u=0} and \eqref{equ: asy of dg(u) as u to 0}, we have
\begin{align}\label{equ: compu 2pi i resu=0}
2\pi i\res_{u=0}(f\cdot \dif g)= 4\pi i \cdot \left(\tilde{\mathbb{A}}^{-1}\right)_{nj}.
\end{align}
On the other hand, we see from \eqref{equ:A-cycle=0} and \eqref{equ: normalization relation of omega} that the left-hand side of \eqref{equ: Riemann bilinear relation} equals to $\int_{B_j}\dif g$.
A straightforward computation with the aid of \eqref{equ:Omega_j} and \eqref{equ:hatOmega_k} shows that 
\begin{align}\label{equ: direct compu for B-cycle dg}
\int_{B_j}\dif g=\sum_{l=1}^{j}2(-1)^{n-l}\int_{b_{l-1}}^{a_l}\frac{\mathrm{p}(z)}{|\mathcal{R}(z)|^{\frac{1}{2}}}\dif z=-\Omega_j. 
\end{align}
The first equality in \eqref{equ: tildeA = omega_j-0} then follows by inserting the above two formulae into \eqref{equ: Riemann bilinear relation}, and the second one is a consequence of the definition of $V_j$ given in \eqref{def:Vj}.
\end{proof}

To proceed, we see from \eqref{def:mathcalG} that 
\begin{align}\label{equ:nablamathcalG/mathcalG-a}
&\frac{\nabla\mathcal{G}\left(\vec{\mathcal{A}}(\infty)+\vec{d}\right)}{\mathcal{G}\left(\vec{\mathcal{A}}(\infty)+\vec{d}\right)}=\frac{\nabla\theta\left(\vec{\mathcal{A}}(\infty)+\vec{V}(s)+\vec{d}\right)}{\theta\left(\vec{\mathcal{A}}(\infty)+\vec{V}(s)+\vec{d}\right)}-
\frac{\nabla\theta\left(\vec{\mathcal{A}}(\infty)+\vec{d}\right)}{\theta\left(\vec{\mathcal{A}}(\infty)+\vec{d}\right)}.
\end{align}
As $\theta(\vec{z})=\theta(-\vec{z})$, if follows that 
$\partial_j\theta(\vec{z})=\partial_{z_j}\theta(z_j)=\partial_{z_j}\theta(-z_j)=-\partial_{j}\theta(-\vec{z})$, and thus
\begin{align}\label{equ:nablamathcalG/mathcalG-b}
&\frac{\nabla\mathcal{G}\left(-\vec{\mathcal{A}}(\infty)-\vec{d}\right)}{\mathcal{G}\left(-\vec{\mathcal{A}}(\infty)-\vec{d}\right)}=-\frac{\nabla\theta\left(\vec{\mathcal{A}}(\infty)-\vec{V}(s)+\vec{d}\right)}{\theta\left(\vec{\mathcal{A}}(\infty)-\vec{V}(s)+\vec{d}\right)}+
\frac{\nabla\theta\left(\vec{\mathcal{A}}(\infty)+\vec{d}\right)}{\theta\left(\vec{\mathcal{A}}(\infty)+\vec{d}\right)}.
\end{align}
Inserting the above two equalities into \eqref{equ:P(infty)1,11} and \eqref{equ:P(infty)1,22}, we obtain from \eqref{equ: tildeA = omega_j-0} that
\begin{align}\label{prop:iP(infty)1,11-iP(infty)1,22inter}
&i\left({(P^{(\infty)}_{1})}_{11}-{(P^{(\infty)}_{1})}_{22}\right)\nonumber\\
&=\frac{1}{2}
\begin{pmatrix}
    \frac{\dif V_1(s)}{\dif s} & \cdots & \frac{\dif V_{n}(s)}{\dif s}
\end{pmatrix}
\left(\frac{\nabla\theta\left(\vec{\mathcal{A}}(\infty)+\vec{V}(s)+\vec{d}\right)}{\theta\left(\vec{\mathcal{A}}(\infty)+\vec{V}(s)+\vec{d}\right)}-\frac{\nabla\theta\left(\vec{\mathcal{A}}(\infty)-\vec{V}(s)+\vec{d}\right)}{\theta\left(\vec{\mathcal{A}}(\infty)-\vec{V}(s)+\vec{d}\right)}\right)-2i\mathcal{D}_{\infty,1}\nonumber\\
&=\frac{1}{2}\left(\frac{\dif}{\dif s} \log\theta\left(\vec{\mathcal{A}}(\infty)+\vec{V}(s)+\vec{d}\right)+\frac{\dif}{\dif s} \log\theta\left(\vec{\mathcal{A}}(\infty)-\vec{V}(s)+\vec{d}\right)\right)-2i\mathcal{D}_{\infty,1}\nonumber\\
&=\frac{1}{2}\frac{\dif}{\dif s}\log\left[
\theta\left(\vec{\mathcal{A}}(\infty)+\vec{V}(s)+\vec{d}\right)
\theta\left(\vec{\mathcal{A}}(\infty)-\vec{V}(s)+\vec{d}\right)
\right]-2i\mathcal{D}_{\infty,1},
\end{align}
where $\mathcal{D}_{\infty,1}$ is given in \eqref{def:Dinfty1-a}. A further appeal to  \eqref{equ: vecA(infty)+vecd== mod Lambda} and \eqref{property:theta function} shows that
$$\theta\left(\vec{\mathcal{A}}(\infty)+\vec{V}(s)+\vec{d}\right)
\theta\left(\vec{\mathcal{A}}(\infty)-\vec{V}(s)+\vec{d}\right)=\theta(\vec{V}(s))^2,$$ 
which, together with  \eqref{prop:iP(infty)1,11-iP(infty)1,22inter}, implies that
\begin{align}\label{equ:iP(infty)1,11-iP(infty)1,22}
    i\left({(P^{(\infty)}_{1})}_{11}-{(P^{(\infty)}_{1})}_{22}\right)=\frac{\dif}{\dif s}\log\theta\left(\vec{V}(s)\right)-2i\mathcal{D}_{\infty,1}.
    \end{align}

Finally, we turn to evaluate the term ${(R^{(1)}_{1})}_{11}-{(R^{(1)}_{1})}_{22}$ in  \eqref{equ: differential identity again} and start with the following lemma.
\begin{lemma}\label{lemma: JR(1)p(-1)11}
With $\left[(J_{R}^{(1)})_{p}^{(-1)}\right]_{11}$, $p\in\mathcal{I}_{e}$, given in \eqref{equ: JR(1)p(-1) (11) and (22)-a}, we have  
\begin{align}\label{equ: JR(1)p(-1) (11) and (22)-b}
\left[(J_{R}^{(1)})_{a_j}^{(-1)}\right]_{11}&=\frac{i}{32}\frac{\prod_{k=0}^{n}(a_j-b_k)}{\mathrm{p}(a_j)}\eta\left(a_j, \vec{V}(s)\right), 
\\
\left[(J_{R}^{(1)})_{b_j}^{(-1)}\right]_{11}&=\frac{i}{32}\frac{\prod_{k=0}^{n}(b_j-a_k)}{\mathrm{p}(b_j)}\eta\left(b_j, \vec{V}(s)\right),\label{equ: JR(1)p(-1)b}
\end{align}
where $\eta(z,\vec{\mu})$ is defined in \eqref{def: h(z) and eta(z,mu)}.
\end{lemma}
\begin{proof}
We only give the proof of \eqref{equ: JR(1)p(-1) (11) and (22)-b}, since \eqref{equ: JR(1)p(-1)b} can be proved similarly.  By \eqref{equ: asy of mathcalE at z=a_j} and $c_{a_j}$ given in \eqref{equ: u_aj expansion at a_j}, we have
$$c_{a_j}^{-1/2}(\mathcal{E}_{a_j}^{(-\frac{1}{4})})^2=i\prod_{k=0}^{n}(a_j-b_k)/{2\mathrm{p}(a_j)}.$$ 
Also, recall the definition of $\mathcal{G}$ in \eqref{def:mathcalG}, a straightforward computation gives
\begin{align*}
\frac{\mathcal{G}\left(\vec{\mathcal{A}}(a_j)+\vec{d}\right)\mathcal{G}\left(-\vec{\mathcal{A}}(a_j)-\vec{d}\right)}{\mathcal{G}\left(\vec{\mathcal{A}}(\infty)+\vec{d}\right)\mathcal{G}\left(-\vec{\mathcal{A}}(\infty)-\vec{d}\right)}
=\eta\left(a_j;\vec{V}(s)\right),
\end{align*}
where we have made use of \eqref{equ: vecA(infty)+vecd== mod Lambda} and \eqref{property:theta function}. We then obtain \eqref{equ: JR(1)p(-1) (11) and (22)-b} by inserting the above two formulae into \eqref{equ: JR(1)p(-1) (11) and (22)-a}. 
\end{proof}
By \eqref{equ: asy for R_1(s) and R(1)_1} and \eqref{equ: JR(1)p(-1) (11) and (22)-a}, we have 
 \begin{align}\label{equ: compu i/s(R(1)1,11-R(1)1,22)}
 \frac{i}{s}\left({(R^{(1)}_{1})}_{11}-{(R^{(1)}_{1})}_{22}\right)=\frac{2i}{s}
 \sum_{j=0}^{n}\left\{\left[(J_{R}^{(1)})_{a_j}^{(-1)}\right]_{11}+\left[(J_{R}^{(1)})_{b_j}^{(-1)}\right]_{11}\right\}.
 \end{align}
Substituting \eqref{equ: JR(1)p(-1) (11) and (22)-b} and \eqref{equ: JR(1)p(-1)b} into \eqref{equ: compu i/s(R(1)1,11-R(1)1,22)}, it follows from the definition of $\mathcal{L}$ given in \eqref{def-intro-mathcalL} that
\begin{align}\label{equ: i/s(R(1)1,11-R(1)1,22)}
\frac{i}{s}\left({(R^{(1)}_{1})}_{11}-{(R^{(1)}_{1})}_{22}\right)=-\frac{1}{16s}\sum_{j=0}^{n}\left(\mathcal{L}\left(a_j; \vec{V}(s)\right)+\mathcal{L}\left(b_j; \vec{V}(s)\right)\right).
\end{align}

A combiniation of \eqref{equ: differential identity again}, \eqref{equ:iP(infty)1,11-iP(infty)1,22} and \eqref{equ: i/s(R(1)1,11-R(1)1,22)} gives 
\begin{multline}\label{equ: partial_slogF(s sigma)}
\partial_{s}\log\mathcal{F}(s\Sigma)=-2\gamma_0s-2i\mathcal{D}_{\infty,1}+\frac{\dif}{\dif s}\log\theta\left(\vec{V}(s)\right)\\-\frac{1}{16s}\sum_{j=0}^{n}\left(\mathcal{L}\left(a_j, \vec{V}(s)\right)+\mathcal{L}\left(b_j, \vec{V}(s)\right)\right)-\frac{\alpha^2-\beta^2}{s}+\mathcal{O}(s^{-2}).
\end{multline}
We obtain \eqref{asy result: general case} by integrating the above formula with respect to $s$ on both sides.

To show \eqref{equ: hat{mathcalL}_p}, we need the following lemma.

\begin{lemma}\label{lemma: mean value-1st prop}
    Let $\hat{s}>0$ and $\mathcal{Y}$: $\mathbb{R}^{n}/\mathbb{Z}^{n}\to\mathbb{R}$ be continuous. Then 
    \begin{align}\label{equ: mean value equality-pre}
    \hat{\mathcal{Y}}:=\lim_{T\to+\infty}\frac{1}{T}\int_{0}^{T}\mathcal{Y}(\vec{V}(t))\dif t\in \mathbb{R}
    \end{align}
is well-defined, and for all $s>0$,
\begin{align}\label{equ: mean value equality-a}
\int_{\hat{s}}^{s}\frac{\mathcal{Y}(\vec{V}(t))}{t}\dif t=\frac{1}{s}\int_{\hat{s}}^{s}\mathcal{Y}(\vec{V}(t))\dif t+\int_{\hat{s}}^{s}
\frac{1}{t}\left(\frac{1}{t}\int_{\hat{s}}^{t}\mathcal{Y}(\vec{V}(t'))\dif t'-\hat{\mathcal{Y}}\right)\dif t+\hat{\mathcal{Y}}\log s.
\end{align}
In particular, we have
\begin{align}\label{equ: mean value equality-b}
\int_{\hat{s}}^{s}\frac{\mathcal{Y}(\vec{V}(t))}{t}\dif t=\hat{\mathcal{Y}}\log s+o(\log s).
\end{align}
\end{lemma}
\begin{proof}
As $\mathbb{R}^{n}/\mathbb{Z}^{n}$ is compact and $\mathcal{Y}$ is continuous, it follows from \cite[Definition 5.1]{Ka-harmonic-analysis} that $\mathcal{Y}(\vec{V}(t))$ is an almost periodic function, and thus $\hat{\mathcal{Y}}\in\mathbb{R}$ is well-defined \cite[Page 176]{Ka-harmonic-analysis}. The equality \eqref{equ: mean value equality-a} can be verified by differentiating both sides with respect to $s$. Note that as $s\to+\infty$,  the first and the second terms on the right-hand side of \eqref{equ: mean value equality-a} are of order $\mathcal{O}(1)$ and $o(\log s)$, respectively, which in turn leads to \eqref{equ: mean value equality-b}.
\end{proof}
It follows from \eqref{property:theta function} and \eqref{def-intro-mathcalL} that $\mathcal{L}(p,\cdot)$, $p\in\mathcal{I}_{e}$, is real analytic. Applying Lemma \ref{lemma: mean value-1st prop} to $\mathcal{Y}=\mathcal{L}(p,\vec{V}(s))$ gives us \eqref{equ: hat{mathcalL}_p}.

This completes the proof of Theorem \ref{thm: general case}. \qed


\subsection{Proof of Theorem \ref{thm: Diophantine case}}
It suffices to prove \eqref{equ: hat{mathcalL}_p-Dio case}, which relies on the following lemma. 
\begin{lemma}\label{lemma: mean value-2nd prop}
Let $\mathcal{Y}:\mathbb{R}^{n}/\mathbb{Z}^{n}\to\mathbb{R}$ be analytic and assume that  $\vec{\Omega}\in \mathcal{S}_D$. For any $\hat{s}>0$, we have
\begin{align}\label{eq:intDio}
\frac{1}{t}\int_{\hat{s}}^t\mathcal{Y}(\vec{V}(t'))\dif t'=\hat{\mathcal{Y}}+\mathcal{O}(t^{-1}),\qquad t\to +\infty,
\end{align}
where $\hat{\mathcal{Y}}$ is given in \eqref{equ: mean value equality-pre}
\end{lemma}
\begin{proof}
Note that $\mathcal{Y}$ is analytic on $\mathbb{R}^{n}/\mathbb{Z}^{n}$, it follows 
\begin{align}\label{equ:Fourier expansion}
\mathcal{Y}(\vec{\mu})=\sum_{\vec{m}\in\mathbb{Z}^{n}}\ell_{\vec{m}}e^{2\pi i\vec{m}^{\rm T}\vec{\mu}},\qquad \vec{\mu}\in\mathbb{R}^{n}/\mathbb{Z}^{n},
\end{align}
where the Fourier coefficients $\{\ell_{\vec{m}}\}$ satisfy
\begin{align}\label{equ: Fourier coefficient bound}
|\ell_{\vec{m}}|\leqslant c_1e^{-c_2{\Vert\vec{m}\Vert}_{2}}, \quad \vec{m}\in\mathbb{Z}^{n},
\end{align}
for some positive numbers $c_1$ and $c_2$.   Let $\vec{\tilde{\zeta}}:={(\im\zeta_{1},\dots,\im\zeta_n)}^{\rm T}$, we obtain from \eqref{def:Vj} that
\begin{align}\label{equ:some equ necessary}
    \frac{1}{t}\int_{\hat{s}}^t\mathcal{Y}(\vec{V}(t'))\dif t'
    &=
    \frac{1}{t}\sum_{\vec{m}\in\mathbb{Z}^{n}}\int_{\hat{s}}^t \ell_{\vec{m}}e^{i\vec{m}^{\rm T}(t'\vec{\Omega}+\vec{\tilde{\zeta}})}\dif t'\nonumber\\
    &=\frac{1}{t}\sum_{\substack{\vec{m}\in\mathbb{Z}^{n} \\ \vec{m}^{\rm T}\vec{\Omega}\neq 0}}\ell_{\vec{m}}e^{i\vec{m}^{\rm T}\vec{\tilde{\zeta}}}
    \frac{e^{i\vec{m}^{\rm T}\vec{\Omega}t}-e^{i\vec{m}^{\rm T}\vec{\Omega}\hat{s}}}{i\vec{m}^{\rm T}\vec{\Omega}}
    -\frac{\hat{s}}{t}\sum_{\substack{\vec{m}\in\mathbb{Z}^{n} \\ \vec{m}^{\rm T}\vec{\Omega}= 0}}\ell_{\vec{m}}e^{i\vec{m}^{\rm T}\vec{\tilde{\zeta}}}
    \nonumber \\
    & \quad +
    \sum_{\substack{\vec{m}\in\mathbb{Z}^{n} \\ \vec{m}^{\rm T}\vec{\Omega}= 0}}\ell_{\vec{m}}e^{^{i\vec{m}^{\rm T}\vec{\tilde{\zeta}}}}.
\end{align}
As $\vec{\Omega}$ and $\vec{\tilde{\zeta}}$ are real-valued vectors, it follows from \eqref{condition: diophantine} and \eqref{equ:some equ necessary} that
\begin{align}
& \left\vert\frac{1}{t}\sum_{\substack{\vec{m}\in\mathbb{Z}^{n} \\ \vec{m}^{\rm T}\vec{\Omega}\neq 0}}\ell_{\vec{m}}e^{i\vec{m}^{\rm T}\vec{\tilde{\zeta}}}
    \frac{e^{i\vec{m}^{\rm T}\vec{\Omega}t}-e^{i\vec{m}^{\rm T}\vec{\Omega}\hat{s}}}{i\vec{m}^{\rm T}\vec{\Omega}}
    -\frac{\hat{s}}{t}\sum_{\substack{\vec{m}\in\mathbb{Z}^{n} \\ \vec{m}^{\rm T}\vec{\Omega}= 0}}\ell_{\vec{m}}e^{i\vec{m}^{\rm T}\vec{\tilde{\zeta}}}\right\vert
    \nonumber \\
&\lesssim \frac{1}{t}\sum_{\substack{\vec{m}\in\mathbb{Z}^{n} \\ \vec{m}^{\rm T}\vec{\Omega}\neq 0}}|\ell_{\vec{m}}|\cdot\delta_1^{-1}{\Vert\vec{m}\Vert}_{2}^{\delta_2}+\frac{\hat{s}}{t}\sum_{\substack{\vec{m}\in\mathbb{Z}^{n} \\ \vec{m}^{\rm T}\vec{\Omega}= 0}}|\ell_{\vec m}|
=\mathcal{O}(t^{-1}).
\end{align}
Moreover, we have
\begin{align}
\hat{\mathcal{Y}}
&=\lim_{T\to+\infty}\frac{1}{T}\int_{0}^{T}\mathcal{Y}(\vec{V}(t))\dif t
=\lim_{T\to+\infty}\frac{1}{T}\int_{0}^{T}\sum_{\vec{m}\in\mathbb{Z}^{n}}\ell_{\vec{m}}e^{i\vec{m}^{\rm T}(t\vec{\Omega}+\vec{\tilde{\zeta}})}\dif t\nonumber\\
&=\lim_{T\to+\infty}\frac{1}{T}\sum_{\substack{\vec{m}\in\mathbb{Z}^{n} \\ \vec{m}^{\rm T}\vec{\Omega}\neq 0}}\ell_{\vec{m}}e^{i\vec{m}^{\rm T}\vec{\tilde{\zeta}}}
\frac{e^{i\vec{m}^{\rm T}\vec{\Omega}T}-1}{i\vec{m}^{\rm T}\vec{\Omega}}+
\sum_{\substack{\vec{m}\in\mathbb{Z}^{n} \\ \vec{m}^{\rm T}\vec{\Omega}= 0}}\ell_{\vec{m}}e^{^{i\vec{m}^{\rm T}\vec{\tilde{\zeta}}}}\nonumber\\
&=\sum_{\substack{\vec{m}\in\mathbb{Z}^{n} \\ \vec{m}^{\rm T}\vec{\Omega}= 0}}\ell_{\vec{m}}e^{^{i\vec{m}^{\rm T}\vec{\tilde{\zeta}}}}.
\end{align}
Combining the above three formulae gives us \eqref{eq:intDio}.
\end{proof}
Applying Lemma \ref{lemma: mean value-2nd prop} to $\mathcal{Y}=\mathcal{L}(p,\vec{V}(s))$ with $p\in\mathcal{I}_{e}$, one has 
\begin{align}
\frac{1}{t}\int_{\hat{s}}^t\mathcal{L}(p,\vec{V}(t'))\dif t'=\hat{\mathcal{L}}_{p}+\mathcal{O}(t^{-1}).
\end{align}
Substituting the above estimate into the second term of the right-hand side of \eqref{equ: mean value equality-a} gives us \eqref{equ: hat{mathcalL}_p-Dio case}.

This completes the proof of Theorem \ref{thm: Diophantine case}. \qed 

\subsection{Proof of Theorem \ref{thm: ergodic case}}
We only need to show 
\begin{equation}\label{eq:Lpintegral}
    \hat{\mathcal{L}}_{p}=\frac{h(p)}{\mathrm{p}(p)}\int_{{[0,1)}^{n}}\eta(p; u_1, u_2,\dots u_n )\dif u_1\cdots\dif u_n, \qquad p\in \mathcal{I}_e,
\end{equation}
which follows from the following well-known Birkhoff's ergodic theorem; cf. \cite[Page 286, Theorem]{Arnold-DS}.
\begin{lemma}\label{lemma: Birkhoff}
Let $\mathcal{Y}: \mathbb{R}^{n}/\mathbb{Z}^{n}\to\mathbb{R}$ be an analytic function. If $\vec{\Omega}\in \mathcal{S}_{E}$, we have 
\begin{equation}\label{eq:Birkhoff}
   \bar{\mathcal{Y}}_{\rm time}=\bar{\mathcal{Y}}_{\rm space}, 
\end{equation}
where 
\begin{align}\label{equ:Y-time-average}
\bar{\mathcal{Y}}_{\rm time}:=\lim_{T\to+\infty}\frac{1}{T}\int_{0}^{T}\mathcal{Y}(\vec{V}(t)+\vec{x})\dif t 
\end{align}
with  $\vec{V}$ given in \eqref{def:Vj} is  the time average of  $\mathcal{Y}$ and 
\begin{align}
\bar{\mathcal{Y}}_{\rm space}:=\int_{\mathbb{R}^{n}/\mathbb{Z}^{n}}\mathcal{Y}(\vec{x})\dif\vec{x}=\int_{{[0,1)}^{n}}\mathcal{Y}(\vec{x})\dif x_1\cdots\dif x_n
\end{align}
is the space average of  $\mathcal{Y}$. 
\end{lemma}
\begin{proof}
The claim holds if $\mathcal{Y}$ is merely Riemann integrable, but we assume analyticity here to show how the rational independence of $\vec{\Omega}$ removes the influence of resonance in a simpler way.

The analyticity of $\mathcal{Y}$ implies that $\mathcal{Y}$ admits a Fourier expansion as shown in \eqref{equ:Fourier expansion}, where
\begin{align}\label{def:Fourier coeff integral}
\ell_{\vec{m}}=\int_{\mathbb{R}^{n}/\mathbb{Z}^{n}}\mathcal{Y}(\vec{\mu})e^{-2\pi i\vec{m}^{\rm T}\vec{\mu}}\dif\vec{\mu}, \qquad \vec{m}\in\mathbb{Z}^{n}.
\end{align}
Substituting \eqref{equ:Fourier expansion} and \eqref{def:Fourier coeff integral}
into \eqref{equ:Y-time-average}, it follows that 
\begin{align}
    \bar{\mathcal{Y}}_{\rm time}&=\sum_{\vec{m}\in\mathbb{Z}^{n}}\ell_{\vec{m}}e^{2\pi i\vec{m}^{\rm T}(\vec{x}+\vec{\tilde{\zeta}})}
    \lim_{T\to+\infty}\frac{1}{T}\int_{0}^{T}e^{2\pi it\vec{m}^{\rm T}\vec{\Omega}}\dif t\nonumber\\
    &=\sum_{\substack{\vec{m}\in\mathbb{Z}^{n} \\ \vec{m}^{\rm T}\vec{\Omega}= 0}}\ell_{\vec{m}}e^{2\pi i\vec{m}^{\rm T}(\vec{x}+\vec{\tilde{\zeta}})}+
    \sum_{\substack{\vec{m}\in\mathbb{Z}^{n}\\ \vec{m}^{\rm T}\vec{\Omega}\neq 0}}\ell_{\vec{m}}e^{2\pi i\vec{m}^{\rm T}(\vec{x}+\vec{\tilde{\zeta}})}\times 
    \lim_{T\to+\infty}\frac{1}{2\pi iT\vec{m}^{\rm T}\vec{\Omega}}\left(e^{2\pi iT\vec{m}^{\rm T}\vec{\Omega}}-1\right)\nonumber\\
    &=\sum_{\substack{\vec{m}\in\mathbb{Z}^{n} \\ \vec{m}^{\rm T}\vec{\Omega}= 0}}\ell_{\vec{m}}e^{2\pi i\vec{m}^{\rm T}(\vec{x}+\vec{\tilde{\zeta}})},
\end{align}
where the limit in the second equality vanishes due to the fact that  $\vec{m}^{\rm T}\vec{\Omega}$ is real.  As $\vec{\Omega}\in \mathcal{S}_{E}$, we have $\vec{m}\equiv\vec{0}$ if $\vec{m}^{\rm T}\vec{\Omega}= 0$, which yields $\bar{\mathcal{Y}}_{\rm time}=\ell_{\vec{0}}$. On account of \eqref{def:Fourier coeff integral}, we obtain \eqref{eq:Birkhoff} immediately. 
\end{proof}
Applying Lemma \ref{lemma: Birkhoff} to $\hat{\mathcal{L}}_{p}$ defined in  \eqref{equ: hat{mathcalL}_p}, we arrive at \eqref{eq:Lpintegral} and complete the proof of Theorem \ref{thm: ergodic case}. \qed

\subsection{Proof of Theorem \ref{thm: n=1 case}}\label{subsec: proof of coefficient of logs in case of n=1}
Since $\vec{\Omega}\in \mathcal{S}_D\cap \mathcal{S}_E$ for $n=1$, it follows from Theorems \ref{thm: Diophantine case} and \ref{thm: ergodic case} that we only need to show
\begin{equation}\label{equ: hat L(z) int for u}
    \hat{\mathcal{L}}_{p}=\frac{h(p)}{\mathrm{p}(p)}\int_{0}^{1}\eta(p,u)\dif u=2, \qquad p\in\{a_0,b_0,a_1,b_1\},
\end{equation}
where, by \eqref{def:Vj}, \eqref{def:choice of d}, \eqref{def:tau matrix} and \eqref{equ:Abel map on mathbbC},
\begin{align}\label{def: h(z) and p(z) in case of n=1}
h(z)=(z-a_0)(z-a_1)+(z-b_0)(z-b_1), \qquad \mathrm{p}(z)=(z-x_1)(z-x_2),
\end{align}
with $x_1\in(a_0,b_0)$ and $x_2\in(a_1,b_1)$, and 
\begin{align}\label{equ: eta}
\eta(z,V_1(s))=\frac{{\theta(0)}^2\theta({\mathcal{A}}(z)+V_1(s)+d)\theta({\mathcal{A}}(z)-V_1(s)+d)}{{\theta(V_1(s))}^2{\theta({\mathcal{A}}(z)+d)}^2},
\end{align}
with 
\begin{align}
\mathcal{A}(z)=\int_{a_0}^{z}\omega_1=\frac{1}{a_{1,0}}\int_{a_0}^{z}\frac{\dif\xi}{\sqrt{\mathcal{R}(\xi)}}, \qquad V_1(s)=\frac{s\Omega_{1}}{2\pi}+\frac{\im \zeta_{1}}{2\pi},
\end{align}
\begin{align}
  d=-\frac{1+\tau}{2}+\int_{a_0}^{z_1}\omega_1, \qquad 
\tau:=\tau_{11}=\oint_{B_1}\omega_1.
\end{align}

To evaluate the integral $\int_{0}^{1}\eta(z,u)\dif u$,  we introduce the first kind $\theta$-function $\theta_1$ by  
\begin{align}\label{def: theta_1 function}
\theta_1(z)=ie^{-\pi iz+\frac{\pi i\tau}{4}}\theta\left(z-\frac{1+\tau}{2}\right),
\end{align}
where $\theta$ is defined in \eqref{def:theta func} with $n=1$. The function $\theta_1(z)$ is odd, and its zero locates at $0$ modulo the lattice $(\mathbb{Z},\tau\mathbb{Z})$. By \eqref{equ: eta}, it follows  that 
\begin{align}
\eta\left(z,u+\frac{1+\tau}{2}\right)&=-\frac{{\theta(0)}^2}{{\theta(\mathcal{A}(z)+d)}^2}\cdot
\frac{\theta_1(\mathcal{A}(z)+u+d)\theta_1(\mathcal{A}(z)-u+d)}{{\theta_1(u)}^2}\nonumber\\
&=-\frac{{\theta(0)}^2{\theta_1(\mathcal{A}(z)+d)}^2}{{\theta(\mathcal{A}(z)+d)}^2{\theta_1'(0)}^2}\cdot u^{-2}+\mathcal{O}(u^{-1}), \qquad u\to 0,
\end{align}
where the second equality follows from the Taylor expansions  $\theta(z)=\theta(0)+\frac{\theta''(0)}{2}z^2+\cdots$ and $\theta_1(z)=\theta_1'(0)z+\frac{\theta_1'''(0)}{6}z^3+\cdots$ as $u\to 0$. The following lemma comes from \cite[Lemma A.1]{Fahs-Krasovsky-2024CPAM}.
\begin{lemma}
\label{lemma: equality between elliptic function and theta function}
If $f(z)$ is an elliptic function (i.e., $f(z)$ is a doubly periodic function with periods $1$ and $\tau$ ) with a single pole modulo the lattice $(\mathbb{Z},\tau\mathbb{Z})$ at $z=\frac{1+\tau}{2}$, and
\begin{align}
   f\left(u+\frac{1+\tau}{2}\right)=c_1u^{-2}+\mathcal{O}(u^{-1}), \qquad u\to 0,
\end{align}
then
\begin{align}
    f(z)=-c_1\left(\left(\frac{\theta^{'}(z)}{\theta(z)}\right)^{'}-\frac{\theta^{''}(0)}{\theta(0)}\right)+f(0),
\end{align}
and furthermore
\begin{align}
\int_{0}^{1}f(z)\dif z=c_1\frac{\theta''(0)}{\theta(0)}+f(0).
\end{align}
\end{lemma}
With the aid of Lemma \ref{lemma: equality between elliptic function and theta function}, 
it follows that for $p\in\{a_0,b_0,a_1,b_1\}$,
\begin{align}\label{equ: int of eta}
    \hat{\mathcal{L}}_{p}=\frac{h(p)}{\mathrm{p}(p)}\int_{0}^{1}\eta(z,u)\dif u=\frac{h(p)}{\mathrm{p}(p)}\left(1-\frac{\theta(0)\theta^{''}(0)}{{\theta^{'}_1(0)}^2}\frac{{\theta_1\left(\mathcal{A}\left(z\right)+d\right)}^2}{{\theta\left(\mathcal{A}\left(z\right)+d\right)^2}}\right).
\end{align}
To proceed, we need the following identities for the $\theta$-functions. 
\begin{lemma}\label{lemma: lemma to state key identities}
For $p\in\{a_0,b_0,a_1,b_1\}$, we have
    \begin{align}\label{equ: part (a) in the lemma to state key identities}
        \frac{{\theta\left(0\right)}^2}{{\theta^{'}_1(0)}^2}\frac{{\theta_1\left(\mathcal{A}\left(p\right)+d\right)^2}}{{\theta\left(\mathcal{A}\left(p\right)+d\right)}^2}h(p)=\frac{4}{a_{1,0}^2},
    \end{align}
where  $a_{1,0}$ is given in \eqref{equ:a_{k,j}}, and
    \begin{align}\label{equ: part (b) in the lemma to state key identities}
        -\frac{4\theta^{''}(0)}{a_{1,0}^2\theta(0)}=2\mathrm{p}(p)-h(p). 
    \end{align}
\end{lemma}
\begin{proof}
The proof of \eqref{equ: part (a) in the lemma to state key identities} is similar to that of \cite[Lemma 16, part (d)]{Fahs-Krasovsky-2024CPAM}. Let
\begin{align}
& f_{1}(z):=\left(\frac{\mathcal{E}(z)+\mathcal{E}^{-1}(z)}{2}\right)^2\frac{{\theta_1(\mathcal{A}(z)+d)}^2}{{\theta(\mathcal{A}(z)+d)}^2}-
\left(\frac{\mathcal{E}(z)-\mathcal{E}^{-1}(z)}{2}\right)^2\frac{{\theta_1(\mathcal{A}(z)-d)}^2}{{\theta(\mathcal{A}(z)-d)}^2},  \label{def: f_1(z)}\\
& f_2(z):=\sqrt{\mathcal{R}(z)}\left[\left(\frac{\mathcal{E}(z)+\mathcal{E}^{-1}(z)}{2}\right)^2\frac{{\theta_1(\mathcal{A}(z)+d)}^2}{{\theta(\mathcal{A}(z)+d)}^2}+
\left(\frac{\mathcal{E}(z)-\mathcal{E}^{-1}(z)}{2}\right)^2\frac{{\theta_1(\mathcal{A}(z)-d)^2}}{{\theta(\mathcal{A}(z)-d)}^2}\right],\label{def: f_2(z)}
\end{align}
where recall that $\sqrt{\mathcal{R}(z)}=\sqrt{(z-a_0)(z-a_1)(z-b_0)(z-b_1)}$ and $\mathcal{E}(z)=[(z-b_0)(z-b_1)]^{\frac{1}{4}}/[(z-a_0)(z-a_1)]^{\frac{1}{4}}$. It follows from Proposition \ref{prop:zeros of theta function} and the fact that $\theta_1(\cdot)/\theta(\cdot)$ is an elliptic function that $f_1(z)$ is meromorphic on $\mathbb{C}$. Also note that the fact the zeros of $\mathcal{E}(z)\pm{\mathcal{E}(z)}^{-1}$ and $\theta(\mathcal{A}(z)\pm d)$ are the same, thus $f_1(z)$ is an entire function. Analogous arguments shows that $f_2(z)$ is entire as well. By Liouville's theorem, it follows 
\begin{align}
&f_1(z)\equiv\lim_{z\to\infty}f_1(z)=0, \label{equ: f1 Liouville}\\
&f_2(z)\equiv\lim_{z\to\infty}f_2(z)=\frac{1}{a_{1,0}}\frac{{\theta_1'(0)}^2}{{\theta(0)}^2}+\frac{(b_1+b_0-a_0-a_1)^2}{16}\frac{{\theta_1(2d)}^2}{{\theta(2d)}^2},\label{equ: f2 Liouville}
\end{align}
where we have used the fact that  $\theta_1$ is odd (which implies that $\theta_1(0)=0$) and \eqref{equ: vecA(infty)+vecd== mod Lambda} in \eqref{equ: f1 Liouville}. 
Since
\begin{align}
f_1(z)=-\frac{1}{z^2}\left[-\frac{1}{a_{1,0}}\frac{{\theta_1'(0)}^2}{{\theta(0)}^2}+\frac{(b_1+b_0-a_0-a_1)^2}{16}\frac{{\theta_1(2d)}^2}{{\theta(2d)}^2}\right]+\mathcal{O}(z^{-3}), \quad z\to \infty, 
\end{align}
from which we obtain 
\begin{align}
\frac{1}{a^2_{1,0}}\frac{{\theta_1'(0)}^2}{{\theta(0)}^2}=\frac{(b_1+b_0-a_0-a_1)^2}{16}\frac{{\theta_1(2d)}^2}{{\theta(2d)^2}}.
\end{align}
Substituting this equality into \eqref{equ: f2 Liouville}, it follows that 
\begin{align}\label{equ: f_2(z) after simplification}
f_2(z)\equiv\frac{2}{a^2_{1,0}}\frac{{\theta_1'(0)}^2}{{\theta(0)}^2}.
\end{align}
By \eqref{def: f_2(z)} and Proposition \ref{prop:jump of mathcalA}, we also have
\begin{align}
f_2(p)=\frac{1}{2}h(p)\frac{{\theta_1(\mathcal{A}(p)+d)}^2}{{\theta(\mathcal{A}(p)+d)}^2}, \qquad 
p\in\{a_0,b_0,a_1,b_1\},
\end{align}
which gives us \eqref{equ: part (a) in the lemma to state key identities} by equating the above two formulae. 

To show \eqref{equ: part (b) in the lemma to state key identities}, similar to the proof of \cite[Equation (252)]{Fahs-Krasovsky-2024CPAM}, we have that \eqref{equ: part (b) in the lemma to state key identities} holds for $p=a_0$. It then suffices to show 
\begin{align}
2\mathrm{p}(a_0)-h(a_0)=2\mathrm{p}(a_1)-h(a_1)=2\mathrm{p}(b_0)-h(b_0)=2\mathrm{p}(b_1)-h(b_1),
\end{align}
where $\mathrm{p}(z)$ and $h(z)$ are given in \eqref{def: h(z) and p(z) in case of n=1}.
Indeed, 
\begin{equation}\label{cal: 2mathrmp(a_0)-h(a_0)}
2\mathrm{p}(a_0)-h(a_0)=a_0^2-2a_0(x_1+x_2)+a_0b_1+a_0b_0-b_0b_1+2x_1x_2.
\end{equation}
Note that
\begin{equation}
\frac{\mathrm{p}(z)}{\sqrt{{\mathcal{R}}(z)}}=1+\frac{1}{z}\left(\frac{a_0+b_0+a_1+b_1}{2}-(x_1+x_2)\right)+\mathcal{O}(z^{-2}), \qquad z\to \infty.
\end{equation}
This, together with the normalization conditions $\oint_{A_j}\mathrm{p}(z)/\sqrt{{\mathcal{R}}(z)}\dif z=0$, $j=0,1$, implies that  
\begin{align}
x_1+x_2=\frac{a_0+b_0+a_1+b_1}{2}.
\end{align}
Hence, by \eqref{cal: 2mathrmp(a_0)-h(a_0)},
\begin{equation}\label{cal: 2mathrmp(a_0)-h(a_0) (b)}
2\mathrm{p}(a_0)-h(a_0)=-(a_0a_1+b_0b_1)+2x_1x_2.
\end{equation}
A straightforward calculation shows that $2\mathrm{p}(p)-h(p)\equiv -a_0a_1-b_0b_1+2x_1x_2$ for $p\in\{b_0,a_1,b_1\}$, which leads us to the desired result.
\end{proof}
By Lemma \ref{lemma: lemma to state key identities}, it follows from \eqref{equ: int of eta} that 
\begin{align}
\hat{\mathcal{L}}_{p}=\frac{1}{\mathrm{p}(p)}\left(h(p)-\frac{4\theta^{''}(0)}{a_{1,0}^2\theta(0)}\right)=\frac{1}{\mathrm{p}(p)}(h(p)+2\mathrm{p}(p)-h(p))=2,
\end{align}
which is \eqref{equ: hat L(z) int for u}.
This completes the proof of Theorem \ref{thm: n=1 case}. \qed 

\addcontentsline{toc}{section}{Acknowledgments} 
\section*{Acknowledgments}
Taiyang Xu is partially supported by China Postdoctoral Science Foundation under grant number 2024M760480 and Shanghai Post-Doctoral Excellence Program under grant number 2024100. Lun Zhang is partially supported by National Natural Science Foundation of China under grant number 12271105. 

\appendix
\section{The confluent hypergeometric parametrix}\label{appendix:PCP}
The confluent hypergeometric parametrix  $\Phi_{\rm CH}(z)=\Phi_{\rm CH}(z;\alpha,\beta)$ with $\alpha$ and $\beta$ being parameters is a solution of the following RH problem. 

\paragraph{RH problem for $\Phi_{\rm CH}$} \label{appendix-A-paragraph-RHP Phi}

\begin{itemize}
\item[\rm (a)]
  $\Phi_{\rm CH}(z)$ is holomorphic for $z\in\mathbb{C}\backslash \Gamma_{\Phi_{\rm CH}}$, where $\Gamma_{\Phi_{\rm CH}}:=\cup_{j=1}^5 \Gamma_{\Phi_{\rm CH}, j}$ 
 with
    $$
\begin{aligned}
& \Gamma_{\Phi_{\rm CH}, 1}=e^{\frac{\pi i}{4}} \mathbb{R}^{+}, \quad \Gamma_{\Phi_{\rm CH}, 2}=e^{\frac{3 \pi i}{4}} \mathbb{R}^{+}, \quad \Gamma_{\Phi_{\rm CH}, 3}=e^{-\frac{3 \pi i}{4}} \mathbb{R}^{+}, \\
& \Gamma_{\Phi_{\rm CH}, 4}=e^{-\frac{\pi i}{2}} \mathbb{R}^{+}, \quad \Gamma_{\Phi_{\rm CH}, 5}=e^{-\frac{\pi i}{4}} \mathbb{R}^{+};
\end{aligned}
$$ 
see Figure \ref{fig:RHP Phi} for an illustration.
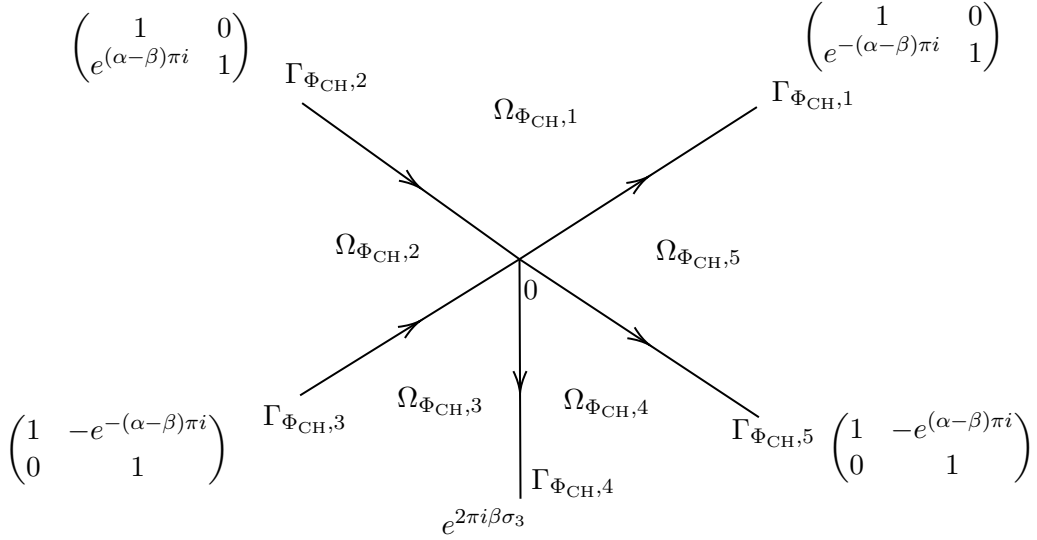
\begin{figure}[htbp]
\centering
\tikzset{every picture/.style={line width=0.75pt}} 
\begin{tikzpicture}[x=0.75pt,y=0.75pt,yscale=-1,xscale=1]
\draw    (221,72) -- (329.5,150.5) ;
\draw [shift={(280.11,114.77)}, rotate = 215.89] [color={rgb, 255:red, 0; green, 0; blue, 0 }  ][line width=0.75]    (10.93,-3.29) .. controls (6.95,-1.4) and (3.31,-0.3) .. (0,0) .. controls (3.31,0.3) and (6.95,1.4) .. (10.93,3.29)   ;
\draw    (329.5,150.5) -- (449,230) ;
\draw [shift={(394.25,193.57)}, rotate = 213.63] [color={rgb, 255:red, 0; green, 0; blue, 0 }  ][line width=0.75]    (10.93,-3.29) .. controls (6.95,-1.4) and (3.31,-0.3) .. (0,0) .. controls (3.31,0.3) and (6.95,1.4) .. (10.93,3.29)   ;
\draw    (329.5,150.5) -- (448,74) ;
\draw [shift={(393.79,109)}, rotate = 147.15] [color={rgb, 255:red, 0; green, 0; blue, 0 }  ][line width=0.75]    (10.93,-3.29) .. controls (6.95,-1.4) and (3.31,-0.3) .. (0,0) .. controls (3.31,0.3) and (6.95,1.4) .. (10.93,3.29)   ;
\draw    (220,219) -- (329.5,150.5) ;
\draw [shift={(279.84,181.57)}, rotate = 147.97] [color={rgb, 255:red, 0; green, 0; blue, 0 }  ][line width=0.75]    (10.93,-3.29) .. controls (6.95,-1.4) and (3.31,-0.3) .. (0,0) .. controls (3.31,0.3) and (6.95,1.4) .. (10.93,3.29)   ;
\draw    (330,271) -- (329.5,150.5) ;
\draw [shift={(329.78,217.75)}, rotate = 269.76] [color={rgb, 255:red, 0; green, 0; blue, 0 }  ][line width=0.75]    (10.93,-3.29) .. controls (6.95,-1.4) and (3.31,-0.3) .. (0,0) .. controls (3.31,0.3) and (6.95,1.4) .. (10.93,3.29)   ;
\draw (330,159.4) node [anchor=north west][inner sep=0.75pt]    {$0$};
\draw (454,59.4) node [anchor=north west][inner sep=0.75pt]    {$\Gamma _{\Phi_{\rm CH}, 1}$};
\draw (211,50.4) node [anchor=north west][inner sep=0.75pt]    {$\Gamma_{\Phi_{\rm CH}, 2}$};
\draw (315,68.4) node [anchor=north west][inner sep=0.75pt]    {$\Omega _{\Phi_{\rm CH}, 1}$};
\draw (236,136.4) node [anchor=north west][inner sep=0.75pt]    {$\Omega _{\Phi_{\rm CH}, 2}$};
\draw (267,213.4) node [anchor=north west][inner sep=0.75pt]    {$\Omega _{\Phi_{\rm CH}, 3}$};
\draw (350,214.4) node [anchor=north west][inner sep=0.75pt]    {$\Omega _{\Phi_{\rm CH}, 4}$};
\draw (396,139.4) node [anchor=north west][inner sep=0.75pt]    {$\Omega _{\Phi_{\rm CH}, 5}$};
\draw (469,18.4) node [anchor=north west][inner sep=0.75pt]    {$\begin{pmatrix}
1 & 0\\
e^{-( \alpha -\beta ) \pi i} & 1
\end{pmatrix}$};
\draw (482,225.5) node [anchor=north west][inner sep=0.75pt]    {$\begin{pmatrix}
1 & -e^{( \alpha -\beta ) \pi i}\\
0 & 1
\end{pmatrix}$};
\draw (103,24.4) node [anchor=north west][inner sep=0.75pt]    {$\begin{pmatrix}
1 & 0\\
e^{( \alpha -\beta ) \pi i} & 1
\end{pmatrix}$};
\draw (71,226.4) node [anchor=north west][inner sep=0.75pt]    {$\begin{pmatrix}
1 & -e^{-( \alpha -\beta ) \pi i}\\
0 & 1
\end{pmatrix}$};
\draw (200,222.4) node [anchor=north west][inner sep=0.75pt]    {$\Gamma _{\Phi_{\rm CH}, 3}$};
\draw (334,254.4) node [anchor=north west][inner sep=0.75pt]    {$\Gamma _{\Phi_{\rm CH}, 4}$};
\draw (434,229.4) node [anchor=north west][inner sep=0.75pt]    {$\Gamma _{\Phi_{\rm CH}, 5}$};
\draw (289,274.4) node [anchor=north west][inner sep=0.75pt]    {$e^{2\pi i\beta \sigma _{3}}$};
\end{tikzpicture}
\caption{The jump contours and regions in the RH problem for $\Phi_{\rm CH}$.}
\label{fig:RHP Phi}
\end{figure}

\item[\rm (b)]  For $z\in\Gamma_{\Phi_{\rm CH}}$, we have $\Phi_{{\rm CH}, +}(z)=\Phi_{{\rm CH}, -}(z)J_{\Phi_{\rm CH}}(z)$, where
\begin{equation}\label{jump4phi}
J_{\Phi_{\rm CH}}(z)=\begin{cases}{\begin{pmatrix}
    1 & 0 \\
e^{-(\alpha-\beta)\pi i} & 1
\end{pmatrix}}, & z \in \Gamma_{\Phi_{\rm CH}, 1}, \\[.4cm]
\begin{pmatrix}
    1 & 0 \\
e^{(\alpha-\beta)\pi i} & 1
\end{pmatrix}, & z \in \Gamma_{\Phi_{\rm CH}, 2}, \\[.4cm]
\begin{pmatrix}
    1 & -e^{-(\alpha-\beta)\pi i} \\
0 & 1
\end{pmatrix}, & z \in \Gamma_{\Phi_{\rm CH}, 3}, \\[.4cm]
e^{2 \pi i \beta \sigma_3}, & z \in \Gamma_{\Phi_{\rm CH}, 4}, \\
\begin{pmatrix}
    1 & -e^{(\alpha-\beta)\pi i} \\
0 & 1
\end{pmatrix}, & z \in \Gamma_{\Phi_{\rm CH}, 5} .\end{cases}
\end{equation}

\item [\rm (c)] As $z\rightarrow\infty$,  we have
\begin{equation}\label{phi aym inf}
  \Phi_{\rm CH}(z)=\left(I+\frac{\Phi_{{\rm CH},1}}{z}+\mathcal{O}\left(\frac{1}{z^2}\right)\right) z^{-\beta \sigma_3} e^{-\frac{i}{2} z \sigma_3},
\end{equation}
where 
\begin{equation}\label{def:PhiCH1}
    \Phi_{{\rm CH},1}=\begin{pmatrix}
        -\left(\alpha^2-\beta^2\right) i & -e^{-\pi i \beta} \frac{\Gamma(1+\alpha-\beta)}{i \Gamma(\alpha+\beta)} \\
e^{\pi i \beta} \frac{\Gamma(1+\alpha+\beta)}{i \Gamma(\alpha-\beta)} & \left(\alpha^2-\beta^2\right) i
    \end{pmatrix}.
\end{equation}

\item [\rm (d)]
If $\alpha>-\frac{1}{2}$ and $ 2\alpha \not\in \mathbb{N}$, we have as $z\to 0$,
\begin{equation}\label{eq: Phi0}
\Phi_{\rm CH}(z)= \Phi^{(0)}_{\rm CH}(z)z^{\alpha \sigma_3} C_j, \qquad ~z\in \Omega_{\Phi_{\rm CH}, j}, \quad j=1,2,3,4,5, 
\end{equation}
for some constant matrix $C_j$, where $\Phi_{\rm CH}^{(0)}(z)$ is holomorphic in a neighborhood of the origin, the regions $\Omega_{\Phi_{\rm CH}, j}$ are illustrated in Figure \ref{fig:RHP Phi}, and the branch for $z^{\alpha \sigma_3}$ is chosen such that $\arg z\in (-\pi/2, 3\pi/2)$.
If $ 2\alpha \in \mathbb{N}$, we have as $z\to 0$,
  \begin{equation}\label{eq: Phi02}
  \Phi_{\rm CH}(z)= \widehat{\Phi}^{(0)}_{\rm CH}(z)z^{\alpha \sigma_3}
  \begin{pmatrix}
                                 1 &(-1)^{2\alpha}\frac{\sin(\pi(\alpha+\beta))}{\pi} \log  z \\
                                 0 &1
                                 \end{pmatrix} 
                                 \widehat{C}_j, ~z\in \Omega_{\Phi_{\rm CH}, j}, ~j=1,2,3,4,5, \end{equation}
for some constant matrix $\hat{C}_j$,  where $\widehat{\Phi}_{\rm CH}^{(0)}(z)$ is holomorphic near the origin and the branches for $z^{\alpha \sigma_3}$ and $\log z$ are chosen such that $\arg z\in (-\pi/2, 3\pi/2)$.
\end{itemize}
This model RH problem has appeared in \cites{TC-Its-Kra-Duke2011, Its-Kra-ContempMath, Xu-Zhao-CMP-2020}, which can also be obtained from the model RH problem in \cites{Dei-Its-Kra-Annals2011, Foul-Mart-Sousa-JofApp} by performing a suitable transformation. The function $\Phi_{\rm CH}$ can be constructed explicitly by using the confluent hypergeometric function. More precisely, we have
\begin{align}\label{solution4phi}
\Phi_{\rm CH}(z)
    & =   e^{-\frac{i z}{2}}\begin{pmatrix}
e^{-\frac{\pi i(\alpha+\beta)}{2}} \frac{\Gamma(1+\alpha-\beta)}{\Gamma(1+2 \alpha)} \phi(\alpha+\beta, 1+2 \alpha; i z) & -e^{\frac{\pi i(\alpha-\beta)}{2}} \frac{\Gamma(2 \alpha)}{\Gamma(\alpha+\beta)} \phi(-\alpha+\beta, 1-2 \alpha; i z) \\
e^{-\frac{\pi i(\alpha-\beta)}{2}} \frac{\Gamma(1+\alpha+\beta)}{\Gamma(1+2 \alpha)} \phi(1+\alpha+\beta, 1+2 \alpha; i z) & e^{\frac{\pi i(\alpha+\beta)}{2}} \frac{\Gamma(2 \alpha)}{\Gamma(\alpha-\beta)} \phi(1-\alpha+\beta, 1-2 \alpha; i z)
 \end{pmatrix}
 \nonumber 
 \\
& ~~\times z^{\alpha \sigma_3} \begin{pmatrix}
1 & \frac{\sin (\pi(\alpha+\beta))}{\sin (2 \pi \alpha)} \\
0 & 1
\end{pmatrix}, \qquad z\in \Omega_{\Phi,1},
\end{align}
for $\alpha>-1/2$ and $2 \alpha \notin \mathbb{N}$. Here, $\phi(a,b;z)$ is the confluent hypergeometric function defined in \eqref{equ:expansion for phi}. The expression of  $\Phi_{\rm CH}$ in the other regions is then determined by using \eqref{solution4phi} and the jump condition \eqref{jump4phi}. If $2\alpha\in \mathbb{N}$, $\Phi_{\rm CH}$ can be built in a similar way and we omit the details here.

\section{The Bessel parametrix}\label{appendix:Bessel}
The Bessel parametrix  $\Phi_{\rm Be}$ is a solution of the following RH problem. 
\paragraph{RH problem for $\Phi_{\rm Be}$}
\begin{itemize}
    \item[\rm (a)] $\Phi_{\rm Be}(z)$ is holomorphic for $z\in\mathbb{C}\setminus\Gamma_{\rm Be}$, where 
    $\Gamma_{\rm Be}:=(-\infty,0)\cup e^{\frac{2\pi i}{3}}(0,+\infty)\cup e^{-\frac{2\pi i}{3}}(0,+\infty)$ with the orientation illustrated in the Figure \ref{fig:RHP Phi_{Be}}.

\begin{figure}[h]
\centering
\tikzset{every picture/.style={line width=0.75pt}} 
\begin{tikzpicture}[x=0.75pt,y=0.75pt,yscale=-1,xscale=1]
\draw    (224,151) -- (384,150) ;
\draw [shift={(310,150.46)}, rotate = 179.64] [color={rgb, 255:red, 0; green, 0; blue, 0 }  ][line width=0.75]    (10.93,-3.29) .. controls (6.95,-1.4) and (3.31,-0.3) .. (0,0) .. controls (3.31,0.3) and (6.95,1.4) .. (10.93,3.29)   ;
\draw    (284,50) -- (384,150) ;
\draw [shift={(338.24,104.24)}, rotate = 225] [color={rgb, 255:red, 0; green, 0; blue, 0 }  ][line width=0.75]    (10.93,-3.29) .. controls (6.95,-1.4) and (3.31,-0.3) .. (0,0) .. controls (3.31,0.3) and (6.95,1.4) .. (10.93,3.29)   ;
\draw    (384,150) -- (290,250) ;
\draw [shift={(341.79,194.9)}, rotate = 133.23] [color={rgb, 255:red, 0; green, 0; blue, 0 }  ][line width=0.75]    (10.93,-3.29) .. controls (6.95,-1.4) and (3.31,-0.3) .. (0,0) .. controls (3.31,0.3) and (6.95,1.4) .. (10.93,3.29)   ;
\draw (392,141.4) node [anchor=north west][inner sep=0.75pt]    {$0$};
\draw (328,46.4) node [anchor=north west][inner sep=0.75pt]    {$\begin{pmatrix}
1 & 0\\
1 & 1
\end{pmatrix}$};
\draw (333,207.4) node [anchor=north west][inner sep=0.75pt]    {$\begin{pmatrix}
1 & 0\\
1 & 1
\end{pmatrix}$};
\draw (218,106.4) node [anchor=north west][inner sep=0.75pt]    {$\begin{pmatrix}
0 & 1\\
-1 & 0
\end{pmatrix}$};
\end{tikzpicture}
\caption{The jump contour $\Gamma_{\rm Be}$ of the RH problem for $\Phi_{\mathrm{Be}}$.}
\label{fig:RHP Phi_{Be}}
\end{figure}
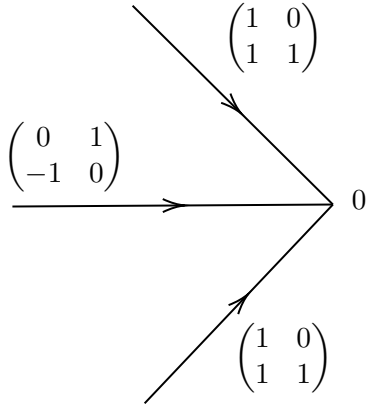

\item[\rm (b)] For $z\in \Gamma_{\rm Be}$, we have  $
    \Phi_{{\rm Be}, +}(z)= \Phi_{{\rm Be}_-}(z)J_{\Phi_{\rm Be}}(z)$,
  where
\begin{equation}\label{jump of PhiBe}
      J_{\Phi_{\rm Be}}(z)=\left\{\begin{array}{ll}
  {\begin{pmatrix}
0 & 1\\
-1 & 0
\end{pmatrix}},& {z\in(-\infty,0)}, \\
  {
\begin{pmatrix}
1 & 0\\
1 & 1
\end{pmatrix}},& z\in e^{\pm\frac{2}{3}\pi i}(0,+\infty).
\end{array}\right.
\end{equation}

\item[\rm (c)]As $z\to\infty$, we have
\begin{align}\label{equ:asy of PhiBe at infty}
\Phi_{\rm Be}(z)=\left(\pi z^{\frac{1}{2}}\right)^{-\frac{\sigma_3}{2}}M
\left(I+\frac{\Phi_{{\rm Be}, 1}}{z^{\frac{1}{2}}}+
\frac{\Phi_{{\rm Be},2}}{z}+\mathcal{O}\left(z^{-\frac{3}{2}}\right)\right)e^{z^{\frac{1}{2}}\sigma_3},
\end{align}
where
\begin{align}\label{equ: Bessel coe M,be1,be2}
M=\frac{1}{\sqrt{2}}\begin{pmatrix} 1 & i\\ i& 1 \end{pmatrix}, \quad
\Phi_{{\rm Be}, 1}=\frac{1}{8}\begin{pmatrix}-1 & -2i \\ -2i & 1 \end{pmatrix}, \quad
\Phi_{{\rm Be}, 2}=-\frac{3}{128}\begin{pmatrix} 1 & -4i \\ 4i & 1 \end{pmatrix}.
\end{align}
\item[\rm (d)] As $z\to 0$, we have 
\begin{align}\label{asy: PhiBe at 0}
\Phi_{\rm Be}(z)=\left\{
\begin{aligned}
&\begin{pmatrix}
    \mathcal{O}(1) & \mathcal{O}(\log |z|)\\
    \mathcal{O}(z) & \mathcal{O}(\log |z|)
\end{pmatrix}, & \vert \arg z  \vert<\frac{2\pi}{3},\\
&\begin{pmatrix}
    \mathcal{O}(\log |z|) & \mathcal{O}(\log |z|)\\
    \mathcal{O}(\log |z|) & \mathcal{O}(\log |z|)
\end{pmatrix}, & \frac{2\pi}{3}<\vert \arg z  \vert<\pi.
\end{aligned}
\right.
\end{align}
\end{itemize}
By \cite{KMAV-AdV-2004}, the unique solution to the above RH problem is given by
\begin{align}\label{equ:PhiBe in cone}
\Phi_{\rm Be}(z)=\left\{
\begin{aligned}
&\begin{pmatrix}
I_0(z^{\frac{1}{2}}) & \frac{i}{\pi}K_0(z^{\frac{1}{2}}) \\
\pi iz^{\frac{1}{2}}I_0'(z^{\frac{1}{2}}) & -z^{\frac{1}{2}}K_0'(z^{\frac{1}{2}})
\end{pmatrix}, & |\arg z|<\frac{2}{3}\pi,\\
&\frac{1}{2}\begin{pmatrix}
H_{0}^{(1)}(-iz^{\frac{1}{2}}) &  H_{0}^{(2)}(-iz^{\frac{1}{2}}) \\
\pi z^{\frac{1}{2}}{H_{0}^{(1)}}'(-iz^{\frac{1}{2}}) & \pi z^{\frac{1}{2}}{H_{0}^{(2)}}'(-iz^{\frac{1}{2}})
\end{pmatrix}, & \frac{2}{3}\pi<\arg z<\pi,\\
&\frac{1}{2}\begin{pmatrix}
H_{0}^{(2)}(iz^{\frac{1}{2}}) &  -H_{0}^{(1)}(iz^{\frac{1}{2}}) \\
-\pi z^{\frac{1}{2}}{H_{0}^{(2)}}'(iz^{\frac{1}{2}}) & \pi z^{\frac{1}{2}}{H_{0}^{(1)}}'(iz^{\frac{1}{2}})
\end{pmatrix}, & -\pi<\arg z<-\frac{2}{3}\pi,
\end{aligned}
\right.
\end{align}
where $I_0$ and $K_0$ are the modified first and second Bessel functions respectively, and $H_0^{(1)}$, $H_0^{(2)}$ are the Hankel functions of the first
and second kind; see \cite{NISTbook} for definitions and more properties of these special functions. 

\end{document}